\documentclass[onecolumn]{IEEEtran}
\usepackage{soul}
\usepackage[mathscr]{eucal}
\usepackage[cmex10]{amsmath}
\usepackage{epsfig,epsf,psfrag}
\usepackage{amssymb,amsmath,amsthm,amsfonts,latexsym}
\usepackage{amsmath,graphicx,bm,xcolor,url,overpic}
\usepackage[caption=false]{subfig} 
\usepackage{fixltx2e}
\usepackage{array}
\usepackage{verbatim}
\usepackage{bm}
\usepackage{algorithmic}
\usepackage{algorithm}
\usepackage{verbatim}
\usepackage{textcomp}
\usepackage{mathrsfs}
\usepackage{epstopdf}

\newcommand{\openone}{\leavevmode\hbox{\small1\normalsize\kern-.33em1}}

\catcode`~=11 \def\UrlSpecials{\do\~{\kern -.15em\lower .7ex\hbox{~}\kern .04em}} \catcode`~=13 

\allowdisplaybreaks[4]

\newcommand{\nn}{\nonumber}

\newcommand{\calA}{\mathcal{A}}
\newcommand{\calB}{\mathcal{B}}

\newcommand{\calD}{\mathcal{D}}

\newcommand{\calL}{\mathcal{L}}
\newcommand{\calM}{\mathcal{M}}
\newcommand{\calN}{\mathcal{N}}

\newcommand{\calP}{\mathcal{P}}

\newcommand{\calR}{\mathcal{R}}
\newcommand{\calS}{\mathcal{S}}
\newcommand{\calT}{\mathcal{T}}

\newcommand{\calX}{\mathcal{X}}
\newcommand{\calY}{\mathcal{Y}}
\newcommand{\calZ}{\mathcal{Z}}

\newcommand{\bx}{\mathbf{x}}

\newcommand{\rmb}{\mathrm{b}}

\newcommand{\rmd}{\mathrm{d}}

\newcommand{\rme}{\mathrm{e}}

\newcommand{\rmH}{\mathrm{H}}

\newcommand{\rmK}{\mathrm{K}}

\newcommand{\rmQ}{\mathrm{Q}}

\newcommand{\rmT}{\mathrm{T}}

\newcommand{\rmV}{\mathrm{V}}

\newcommand{\bbN}{\mathbb{N}}

\newcommand{\bbR}{\mathbb{R}}

\DeclareMathAlphabet{\mathbsf}{OT1}{cmss}{bx}{n}
\DeclareMathAlphabet{\mathssf}{OT1}{cmss}{m}{sl}

\newcommand{\rvR}{\mathsf{R}}

\DeclareSymbolFont{bsfletters}{OT1}{cmss}{bx}{n}  
\DeclareSymbolFont{ssfletters}{OT1}{cmss}{m}{n}
\DeclareMathSymbol{\bsfGamma}{0}{bsfletters}{'000}
\DeclareMathSymbol{\ssfGamma}{0}{ssfletters}{'000}
\DeclareMathSymbol{\bsfDelta}{0}{bsfletters}{'001}
\DeclareMathSymbol{\ssfDelta}{0}{ssfletters}{'001}
\DeclareMathSymbol{\bsfTheta}{0}{bsfletters}{'002}
\DeclareMathSymbol{\ssfTheta}{0}{ssfletters}{'002}
\DeclareMathSymbol{\bsfLambda}{0}{bsfletters}{'003}
\DeclareMathSymbol{\ssfLambda}{0}{ssfletters}{'003}
\DeclareMathSymbol{\bsfXi}{0}{bsfletters}{'004}
\DeclareMathSymbol{\ssfXi}{0}{ssfletters}{'004}
\DeclareMathSymbol{\bsfPi}{0}{bsfletters}{'005}
\DeclareMathSymbol{\ssfPi}{0}{ssfletters}{'005}
\DeclareMathSymbol{\bsfSigma}{0}{bsfletters}{'006}
\DeclareMathSymbol{\ssfSigma}{0}{ssfletters}{'006}
\DeclareMathSymbol{\bsfUpsilon}{0}{bsfletters}{'007}
\DeclareMathSymbol{\ssfUpsilon}{0}{ssfletters}{'007}
\DeclareMathSymbol{\bsfPhi}{0}{bsfletters}{'010}
\DeclareMathSymbol{\ssfPhi}{0}{ssfletters}{'010}
\DeclareMathSymbol{\bsfPsi}{0}{bsfletters}{'011}
\DeclareMathSymbol{\ssfPsi}{0}{ssfletters}{'011}
\DeclareMathSymbol{\bsfOmega}{0}{bsfletters}{'012}
\DeclareMathSymbol{\ssfOmega}{0}{ssfletters}{'012}

\newcommand{\hatT}{\hat{T}}

\newcommand{\hatx}{\hat{x}}
\newcommand{\hatX}{\hat{X}}

\newcommand{\haty}{\hat{y}}
\newcommand{\hatY}{\hat{Y}}

\newcommand{\barQ}{\bar{Q}}

\newcommand{\bmu}{\bm{\mu}}

\newcommand{\bSigma	}{\bm{\Sigma}}

\DeclareMathOperator*{\argmin}{arg\,min}
\DeclareMathOperator*{\argsup}{arg\,sup}

\newtheorem{theorem}{Theorem} 
\newtheorem{lemma}[theorem]{Lemma}

\newtheorem{definition}{Definition}

\usepackage{picture,xcolor,cite}

\usepackage[ colorlinks = true,
             linkcolor = blue,
             urlcolor  = blue,
             citecolor = red,
             anchorcolor = green,]{hyperref}
\usepackage{soul,color}
\usepackage{epstopdf}
\newcommand{\myfoot}[1]{\footnote{\color{red}\bf #1}}
\soulregister\em7
\soulregister\cite7
\soulregister\ref7
\soulregister\eqref7
\soulregister\underline7
\soulregister\emph7
\soulregister\footnote7
\soulregister\myfoot7

\begin{document}
\title{Non-Asymptotic Converse Bounds and Refined Asymptotics for Two Lossy Source Coding Problems}

\author{\IEEEauthorblockN{Lin Zhou and Mehul Motani}\\
\thanks{The authors are with the Department of Electrical and Computer Engineering, National University of Singapore (NUS). Emails: lzhou@u.nus.edu; motani@nus.edu.sg.}
\thanks{Part of this paper will be presented at Globecom 2017~\cite{zhou2017fy,zhou2017kaspishort}.}

}

\maketitle

\flushbottom
\allowdisplaybreaks[1]

\begin{abstract}

In this paper, we revisit two multi-terminal lossy source coding problems: the lossy source coding problem with side information available at the encoder and one of the two decoders, which we term as the Kaspi problem (Kaspi, 1994), and the multiple description coding problem with one semi-deterministic distortion measure, which we refer to as the Fu-Yeung problem (Fu and Yeung, 2002). For the Kaspi problem, we first present the properties of optimal test channels. Subsequently, we generalize the notion of the distortion-tilted information density for the lossy source coding problem to the Kaspi problem and prove a non-asymptotic converse bound using the properties of optimal test channels and the well-defined distortion-tilted information density. Finally, for discrete memoryless sources, we derive refined asymptotics which includes the second-order, large and moderate deviations asymptotics. In the converse proof of second-order asymptotics, we apply the Berry-Esseen theorem to the derived non-asymptotic converse bound. The achievability proof follows by first proving a type-covering lemma tailored to the Kaspi problem, then properly Taylor expanding the well-defined distortion-tilted information densities and finally applying the Berry-Esseen theorem. We then generalize the methods used in the Kaspi problem to the Fu-Yeung problem. As a result, we obtain the properties of optimal test channels for the minimum sum-rate function, a non-asymptotic converse bound and refined asymptotics for discrete memoryless sources. Since the successive refinement problem is a special case of the Fu-Yeung problem, as a by-product, we obtain a non-asymptotic converse bound for the successive refinement problem, which is a strict generalization of the non-asymptotic converse bound for successively refinable sources (Zhou, Tan and Motani, 2017).
\end{abstract}


\section{Introduction}
In this paper, we revisit two multi-terminal lossy source coding problems: the lossy source coding problem with side information available at the encoder and one of the two decoders, which we term as the Kaspi problem~\cite[Theorem 1]{kaspi1994}, and the multiple description coding problem with one semi-deterministic distortion measure, which we term as the Fu-Yeung problem~\cite{fu2002rate}. The setting of the Kaspi problem is shown in Figure \ref{systemmodel} where we have one encoder $f$ and two decoders $\phi_1,\phi_2$. The side information $Y^n$ is available to the encoder $f$ and the decoder $\phi_2$ but \emph{not} to the decoder $\phi_1$. The encoder $f$ compresses the source $X^n$ into a message $S$ given the side information $Y^n$. Decoder $\phi_1$ aims to recover source sequence $X^n$ within distortion level $D_1$ under distortion measure $d_1$ using the message $S$. Decoder $\phi_2$ aims to recover $X^n$ within distortion level $D_2$ under distortion measure $d_2$ using the message $S$ and the side information $Y^n$. The rate-distortion function for discrete memoryless sources (DMSes) under bounded distortion measures was derived by Kaspi in~\cite[Theorem 1]{kaspi1994}. 

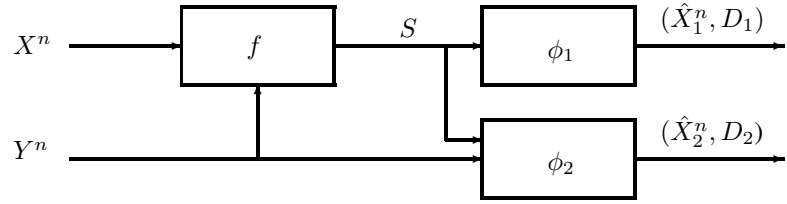
\begin{figure}[t]
\centering
\setlength{\unitlength}{0.5cm}
\scalebox{1}{
\begin{picture}(20,6)
\linethickness{1pt}
\put(0.5,4.8){\makebox{$X^n$}}
\put(0.5,2){\makebox{$Y^n$}}
\put(5,4){\framebox(4,2)}
\put(6.7,4.8){\makebox{$f$}}
\put(2,5){\vector(1,0){3}}
\put(2,2){\vector(1,0){11}}
\put(7,2){\vector(0,1){2}}
\put(13,1){\framebox(4,2)}
\put(13,4){\framebox(4,2)}
\put(14.7,4.7){\makebox{$\phi_1$}}
\put(14.7,1.7){\makebox{$\phi_2$}}
\put(9,5){\vector(1,0){4}}
\put(11,5.5){\makebox(0,0){$S$}}
\put(12,5){\line(0,-1){2.5}}
\put(12,2.5){\vector(1,0){1}}
\put(17,2){\vector(1,0){4}}
\put(17.7,2.5){\makebox{$(\hat{X}_2^n,D_2$)}}
\put(17,5){\vector(1,0){4}}
\put(17.7,5.5){\makebox{$(\hat{X}_1^n,D_1)$}}
\end{picture}}
\caption{The Lossy Source Problem with Side Information Available to the Encoder and One of the Two Decoders~\cite{kaspi1994}.}
\label{systemmodel}
\end{figure}

The setting for the Fu-Yeung problem is shown in Figure~\ref{systemmodelmd}. There are two encoders and three decoders. Each encoder $f_i,~i=1,2$ has access to a source sequence $X^n$ and compresses it into a message $S_i,~i=1,2$. Decoder $\phi_1$ aims to recover $X^n$ with distortion level $D_1$ using the encoded message $S_1$ from encoder $f_1$. Decoder $\phi_2$ aims to recover $X^n$ with distortion level $D_2$ using encoded messages $S_1$ and $S_2$. Decoder $\phi_3$ aims to recover $Y^n$, which is a symbolwise deterministic function  of the source sequence $X^n$. The Fu-Yeung problem is a special case of multiple description coding problem and the El Gamal-Cover inner bound is tight~\cite{gamal1982achievable} for the Fu-Yeung problem. Note that the successive refinement problem~\cite{rimoldi1994} is a special case of the Fu-Yeung problem by letting $Y$ be a constant.

\begin{figure}[tb]
\centering
\setlength{\unitlength}{0.5cm}
\scalebox{1}{
\begin{picture}(26,9)
\linethickness{1pt}
\put(2,5.5){\makebox{$X^n$}}
\put(6,1){\framebox(4,2)}

\put(6,7){\framebox(4,2)}
\put(7.7,1.8){\makebox{$f_2$}}
\put(7.7,7.8){\makebox{$f_1$}}
\put(1,5){\line(1,0){3.5}}
\put(4.5,5){\line(0,1){3}}
\put(4.5,8){\vector(1,0){1.5}}
\put(4.5,5){\line(0,-1){3}}
\put(4.5,2){\vector(1,0){1.5}}
\put(15,4){\framebox(4,2)}
\put(15,1){\framebox(4,2)}
\put(15,7){\framebox(4,2)}
\put(16.7,4.8){\makebox{$\phi_2$}}
\put(16.7,7.7){\makebox{$\phi_1$}}
\put(16.7,1.7){\makebox{$\phi_2$}}
\put(10,2){\vector(1,0){5}}
\put(12,2.5){\makebox(0,0){$S_2$}}
\put(10,8){\vector(1,0){5}}
\put(12,8.5){\makebox(0,0){$S_1$}}
\put(13,8){\line(0,-1){2.6}}
\put(13,5.4){\vector(1,0){2}}
\put(13,2){\line(0,1){2.6}}
\put(13,4.6){\vector(1,0){2}}
\put(19,2){\vector(1,0){4}}
\put(20,2.5){\makebox{$Y^n$}}
\put(19,8){\vector(1,0){4}}
\put(19.7,8.5){\makebox{$(\hat{X}_1^n,D_1)$}}
\put(19,5){\vector(1,0){4}}
\put(19.7,5.5){\makebox{$(\hat{X}_2^n,D_2)$}}
\end{picture}}
\caption{The Multiple Description Coding Problem with One Semi-Deterministic Distortion Measure~\cite{fu2002rate}.}
\label{systemmodelmd}
\end{figure}
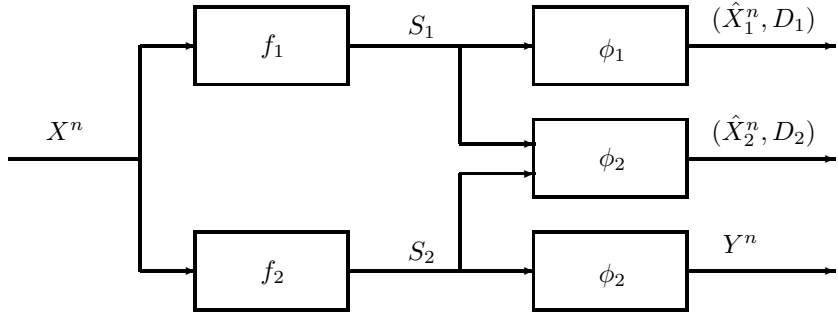
In this paper, we revisit these two lossy source coding problems and present results regarding the properties of optimal test channels, non-asymptotic converse bounds and refined asymptotics for a DMS.

\subsection{Related Works}

We first summarize the works regarding the properties of optimal test channels and non-asymptotic bounds for lossy source coding problems. In \cite{csiszar1974}, Csisz\'ar presented a parametric representation and the properties of optimal test channels for the rate-distortion function~\cite{berger1971rate} of the lossy source coding problem~\cite{shannon1959coding}. In \cite{tuncel2003comp}, Tuncel and Rose derived the properties of the optimal test channels for rate-distortion function of the successive refinement problem~\cite{rimoldi1994}. Kostina and Tuncel~\cite{kostina2017suc} extended the result in \cite{tuncel2003comp} to continuous sources. In \cite{kostina2012converse}, Kostina and Verd\'u derived non-asymptotic converse bounds for the lossy source coding problem and the lossy source coding with encoder and decoder side information. Further, in \cite{kostina2012fixed}, Kostina and Verd\'u derived non-asymptotic achievability bounds for the lossy source coding problem. In \cite{watanabe2015}, Watanabe, Kuzuoka and Tan presented non-asymptotic achievability bounds for the Wyner-Ahlswede-K\"orner problem of the lossless source coding problem with coded side information~\cite{wyner1975,ahlswede1975}, the Wyzer-Ziv problem of the lossy source coding problem with decoder side information~\cite{wyner1976rate} and the Gelfand-Pinsker problem of the channel coding problem with encoder state information~\cite{Gelfand1980}. In a concurrent work of \cite{watanabe2015}, Yassaee, Aref and Gohari~\cite{yassaee2013technique} proposed a novel technique to derive non-asymptotic achievability bounds using stochastic encoding and decoding ideas.

Next, we briefly summarize other works related to the Kaspi problem or the multiple description coding problem. Heegard and Berger~\cite{heegard1985} and Kaspi~\cite{kaspi1994} independently derived the rate-distortion function for the lossy source coding problem with side information available at one of the two decoders (and \emph{not} at the encoder and the other decoder), which is termed as the Kaspi/Heegard-Berger problem. Perron, Diggavi and Telatar derived the explicit formula for the rate-distortion function of the Kaspi problem for a Gaussian memoryless source (GMS) under quadratic distortion measures in \cite{perron2005kaspi} and for a binary memoryless erasure source under Hamming distortion measures in \cite{perron2006kaspi}. Further, Perron, Diggavi and Telatar~\cite{perron2009} consider the lossy source coding problem with Gaussian or erased side information. Tian and Diggavi~\cite{Tian2006} derived explicit formula for the rate-distortion function of the Kaspi/Heegard-Berger problem~\cite{kaspi1994,heegard1985} for a doubly symmetric binary source under Hamming distortion measures. The multiple description coding problem for a binary source under Hamming distortion measures were considered by Wolf, Wyner and Ziv~\cite{wolf1980source} and by Zhang and Berger~\cite{zhang1987new}. The rate-distortion region of the multiple description coding problem for GMS under quadratic distortion measures was characterized by Ozarow~\cite{ozarow1980source}. Venkataramani, Kramer and Goyal considered the multiple description coding problem with many channels in \cite{Kramer2003} while Venkataramanan and Pradhan considered the multiple description coding problem with feedforward in \cite{pradhan2007,pradhan2008}. Other works (non-exhausted) on the multiple description coding problem include \cite{ahlswede1986multiple,zamir1999,witsenhausen1980}.

\subsection{Main Contributions}
In this paper, we revisit two multi-terminal lossy source coding problems: the Kaspi problem~\cite[Theorem 1]{kaspi1994} of the lossy source coding problem with side information available at the encoder and one of the two decoders and the Fu-Yeung problem~\cite{fu2002rate} of the multiple description coding problem with one semi-deterministic distortion measure. Our main contributions can be summarized as follows.

For the Kaspi problem, we first present the properties of optimal test channels and a parametric representation for the rate-distortion function. Second, we generalize the notion of the distortion-tilted information density for the lossy source coding problem~\cite{kostina2012converse} to the Kaspi problem and derive a non-asymptotic converse bound. The non-asymptotic converse bound holds for abstract sources, not restricted to DMSes. Finally, for DMSes under bounded distortion measures, we derive refined asymptotics, i.e., the second-order~(cf.~\cite{kostina2012fixed,ingber2011dispersion,no2016}), the large deviations~(cf.~\cite{Marton74,ihara2000error}) and the moderate deviations~(cf.~\cite{altugwagner2014,polyanskiy2010channel}) asymptotics. The converse proof of the second-order asymptotics follows by applying the Berry-Esseen theorem to the derived non-asymptotic converse bound and the achievability part relies on a type-covering lemma tailored to the Kaspi problem, proper uses of the Taylor's expansions, and the properties of the distortion-tilted information density. Since the Kaspi problem is a generalization of the lossy source coding problem and the lossy source coding problem with encoder and decoder side information, our second-order asymptotics for the Kaspi problem is a generalization of \cite{kostina2012fixed} and \cite{le2014second} for DMSes. We illustrate this point via numerical examples.

We then extend the methods used in the Kaspi problem to the Fu-Yeung problem. First, we present the optimal test channels for the minimum sum-rate function. Subsequently, we generalize the rate-distortion-tilted information for the successive refinement problem~\cite{zhousuc2016} to the Fu-Yeung problem and present a non-asymptotic converse bound. This non-asymptotic bound, when specialized to the case where $|\calY|=1$, gives a non-asymptotic converse bound for the successive refinement problem, which is a strict generalization of \cite[Lemma 9]{zhou2016second}. Finally, we derive refined asymptotics of the Fu-Yeung problem for DMSes under bounded distortion measures and illustrate our results using numerical examples.

\subsection{Organization of the Paper}
The rest of the paper is organized as follows. We set up the notation, formulate the problems and present existing results for the Kaspi problem and the Fu-Yeung problem in Section \ref{sec:pform}. Subsequently, in Section \ref{sec:mainresults4kaspi}, we present our main results for the Kaspi problem: the properties of optimal test channels for the rate-distortion function, a non-asymptotic converse bound, refined asymptotics for DMSes under bounded distortion measures and several numerical examples. In Section \ref{sec:mainresults4fy}, we present our main results for the Fu-Yeung problem: the properties of optimal test channels for the  minimum sum-rate function, a non-asymptotic converse bound and its specialization to the successive refinement problem, refined asymptotics for DMSes under bounded distortion measures and several numerical examples. In Sections \ref{sec:proof4kaspi} and \ref{sec:proof4fy}, we present the proofs of second-order asymptotics for the Kaspi problem and the Fu-Yeung problem respectively. Finally, in Section \ref{sec:conclusion}, we conclude the paper and discuss future research directions. For the smooth presentation of the main results, the proofs of all supporting lemmas are deferred to the appendices.

\section{Problem Formulation and Existing Results}
\label{sec:pform}
\subsection*{Notation}
Random variables and their realizations are in capital (e.g.,\ $X$) and lower case (e.g.,\ $x$) respectively. All sets are denoted in calligraphic font (e.g.,\ $\mathcal{X}$). We use $\calX^{\mathrm{c}}$ to denote the complement of $\calX$. Let $X^n:=(X_1,\ldots,X_n)$ be a random vector of length $n$. We use $\exp(x)$ to denote $e^x$. All logarithms are base $e$. We use $\rmQ(\cdot)$ to denote the standard Gaussian complementary cumulative distribution function (cdf) and  $\rmQ^{-1}(\cdot )$ its inverse. We use standard asymptotic notation such as $O(\cdot)$  and $o(\cdot)$. We use $\bbR$, $\bbR_+$ and $\bbN$ to denote the set of all real numbers, non-negative real numbers and natural numbers respectively. For mutual information, we use $I(X;Y)$ and $I(P_{X},P_{Y|X})$ interchangeably and for conditional mutual information, we use $I(X;Y|Z)$ and $I(P_{X|Z},P_{Y|XZ}|P_Z)$ interchangeably.

The set of all probability distributions on a finite set $\calX$ is denoted as $\calP(\calX)$ and the set of all conditional probability distributions from $\calX$ to $\calY$ is denoted as $\calP(\calY|\calX)$. Given $P\in\calP(\calX)$ and $V\in\calP(\calY|\calX)$, we use $P\times V$ to denote the joint distribution induced by $P$ and $V$. Given a conditional distribution $P_{Z|XY}$ in $\calP(\calZ|\calX\times\calY)$ and a realization $y\in\calY$, we se $P_{Z|X,Y=y}(\cdot|\cdot)$ and $P_{Z|XY}(\cdot|\cdot,y)$ interchangeably. In terms of the method of types, we use the notation as~\cite{Tanbook}. Given a sequence $x^n$, the empirical distribution is denoted as $\hat{T}_{x^n}$. The set of types formed from length $n$ sequences in $\calX$ is denoted as $\calP_{n}(\calX)$. Given $P\in\calP_{n}(\calX)$, the set of all sequences of length $n$ with type $P$ is denoted as $\calT_{P}$.

\subsection{The Kaspi Problem}
The setting of the Kaspi problem is shown in Figure \ref{systemmodel}. In the Kaspi problem, we consider a correlated memoryless source with distribution $P_{XY}$ on the alphabet $\calX\times\calY$. Thus, $(X^n,Y^n)$ is an i.i.d. sequence where each $(X_i,Y_i)$ is generated according to the distribution $P_{XY}$. We assume that the reproduction alphabets for decoders $\phi_1$ and $\phi_2$ are $\hat{\calX}_1$ and $\hat{\calX}_2$ respectively.

\begin{definition}
\label{def:code4kaspi}
An $(n,M)$-code for the Kaspi problem consists of one encoder 
\begin{align}
f:\calX^n\times\calY^n\to \calM:=\{1,\ldots,M\}
\end{align}
and two decoders
\begin{align}
\phi_1&:\calM\to\hat{\calX}_1^n\\
\phi_2&:\calM\times\calY^n\to \hat{\calX}_2^n.
\end{align}
\end{definition}

Hence, $\hat{X}_1^n=\phi_1\big(f(X^n,Y^n)\big)$ and $\hat{X}_2^n=\phi_2\big(f(X^n,Y^n),Y^n\big)$. For $i=1,2$, let $d_i:\calX\times\hat{\calX}_i\to [0,\infty]$ be two distortion measures and let the distortion between $x^n$ and $\hat{x}_i^n$ be defined as $d_i(x^n,\hatx_i^n):=\frac{1}{n}\sum_{j=1}^n d_i(x_j,\hatx_{i,j})$. 

Following \cite{kaspi1994}, the first-order coding rate is defined as follows.
\begin{definition}
\label{firstorder}
A rate $R$ is said to be $(D_1,D_2)$-achievable for the Kaspi problem if there exists a sequence of $(n,M)$-codes such that
\begin{align}
\limsup_{n\to\infty} \frac{\log M}{n}\leq R,
\end{align}
and
\begin{align}
\limsup_{n\to\infty} \mathbb{E}\big[d_i(X^n,\hat{X}_i^n)\big]\leq D_i,~i=1,2.
\end{align}
The minimum $(D_1,D_2)$-achievable rate is denoted as $R_{\rmK}^*(D_1,D_2)$.
\end{definition}

Kaspi~\cite[Theorem 1]{kaspi1994} derived the minimum $(D_1,D_2)$-achievable rate for a DMS under bounded distortion measures. Define
\begin{align}
R_{\rmK}(P_{XY},D_1,D_2)
&:=\min_{\substack{P_{\hat{X}_1|XY},~P_{\hatX_2|XY\hatX_1}:\\\mathbb{E}[d_1(X,\hatX_1)]\leq D_1,\\\mathbb{E}[d_2(X,\hatX_2)]\leq D_2}} I(XY;\hat{X}_1)+I(X;\hat{X}_2|Y\hat{X}_1)\label{kaspiratefunc},
\end{align}

\begin{theorem}
The minimum $(D_1,D_2)$-achievable rate for the Kaspi problem satisfies
\begin{align}
R_{\rmK}^*(D_1,D_2)=R_{\rmK}(P_{XY},D_1,D_2).
\end{align}
\end{theorem}
We refer to $R_{\rmK}(P_{XY},D_1,D_2)$ as the Kaspi rate-distortion function. Note that $R_{\rmK}(P_{XY},D_1,D_2)$ is convex and non-increasing in both $D_1$ and $D_2$. From \cite[(3.4)]{kaspi1994}, we obtain
\begin{align}
 I(XY;\hat{X}_1)+I(X;\hat{X}_2|Y\hat{X}_1)=I(X;\hat{X}_1\hat{X}_2|Y)+I(\hat{X}_1;Y)\label{equivalent}.
\end{align}
Hence the objective function in \eqref{kaspiratefunc} is convex in $(P_{\hat{X}_1|XY},P_{\hatX_2|XY\hat{X}_1})$. We remark that the explicit formulas of the Kaspi rate-distortion function was derived by Perron, Diggavi and Telatar for a GMS under quadratic distortion measures~\cite{perron2005kaspi} and a binary memoryless erasure source under Hamming distortion measures~\cite{perron2006kaspi}.

\subsection{Fu-Yeung Problem}
The setting of the Fu-Yeung problem is shown in Figure \ref{systemmodelmd}. In the Fu-Yeung problem, we consider a memoryless source with distribution $P_{X}$ supported on an alphabet $\calX$. Hence $X^n$ is an i.i.d. sequence where each $X_i$ is generated according to $P_{X}$. Let reproduction alphabets for decoders $\phi_1,\phi_2$ be $\hat{\calX}_1$ and $\hat{\calX}_2$ respectively. Fix a finite set $\calY$ and define a deterministic function $g:\calX \to \calY$. Let $Y_i=g(X_i),~i\in[1:n]$. Note that $P_Y$ is induced by the source distribution $P_X$ and the deterministic function $g$, i.e., for $y\in\calY$, $P_{Y}(y)=\sum_{x:g(x)=y}P_X(x)$. We assume that for each $y$, $P_Y(y)>0$. Decoder $\phi_3$ is required to recover $Y^n=g(X^n)=(g(Y_1),\ldots,g(Y_n))$ losslessly and the decoded sequence is denoted as $\hat{Y}^n$. We follow the definitions of codes and the rate-distortion region in~\cite{fu2002rate}.
\begin{definition}
\label{def:code4fy}
An $(n,M_1,M_2)$-code for the Fu-Yeung problem consists of two encoders:
\begin{align}
f_1:&\calX^n\to\calM_1:=\{1,\ldots,M_1\},\\
f_2:&\calX^n\to\calM_2:=\{1,\ldots,M_2\},
\end{align}
and three decoders:
\begin{align}
\phi_1&:\calM_1\to\mathcal{\hat{X}}_1^n,\\
\phi_2&:\calM_1\times\calM_2\to \mathcal{\hat{X}}_2^n,\\
\phi_3&:\calM_2\to \calY^n.
\end{align}
\end{definition}

Using the encoding and decoding functions, we have $\hatX_1^n=\phi_1(f_1(X^n))$, $\hatX_2^n=\phi_2(f_1(X^n),f_2(X^n))$ and $\hatY^n=\phi_3(f_2(X^n))$. Let $d_{\rmH}$ be the hamming distortion measure and let $d_{\rmH}(Y^n,\hat{Y}^n):=\frac{1}{n}\sum_{i=1}^nd_{\rmH}(Y_i,\hat{Y}_i)$. Let $d_1,d_2$, $d_1(x^n,\hatx_1^n)$ and $d_2(x^n,\hatx_2^n)$ be the same as defined after Definition \ref{def:code4kaspi}. Then the first-order coding region for the Fu-Yeung problem is defined as follows.
\begin{definition}
A rate pair $(R_1,R_2)$ is said to be $(D_1,D_2)$-achievable for the Fu-Yeung problem if there exists a sequence of $(n,M_1,M_2)$-codes such that
\begin{align}
\limsup_{n\to\infty}\frac{\log M_i}{n}\leq R_i,~i=1,2,\label{def:ratecons4fy}
\end{align}
and
\begin{align}
\limsup_{n\to\infty} \mathbb{E}[d_i(X^n,\hat{X}_{i}^n)]&\leq D_i,~i=1,2,\\
\lim_{n\to\infty} \mathbb{E}[d_{\rmH}(Y^n,\hat{Y}^n)]&=0.
\end{align}
The closure of the set of all $(D_1,D_2)$-achievable rate pairs is called the first-order $(D_1,D_2)$-coding region and denoted as $\calR_{\mathrm{FY}}(D_1,D_2|P_X)$.
\end{definition}

The first-order coding region $\calR_{\mathrm{FY}}(D_1,D_2|P_X)$ was characterized by Fu and Yeung in~\cite{fu2002rate} for DMS. In particular, Fu and Yeung~\cite{fu2002rate} showed that the El-Gamal-Cover inner bound~\cite{gamal1982achievable} for the multiple description coding problem is tight for this case. Let $\calP(P_X,D_1,D_2)$ be the set of all pairs of conditional distributions $(P_{\hatX_1|X},P_{\hatX_2|X\hatX_1})$ such that $\mathbb{E}[d_1(X,\hatX_1)]\leq D_1$ and $\mathbb{E}[d_2(X,\hatX_2)]\leq D_2$. Given a pair of conditional distributions $(P_{\hatX_1|X},P_{\hatX_2|X\hatX_1})$, let $\calR(P_{\hatX_1|X},P_{\hatX_2|X\hatX_1})$ be the collection of rate pairs such that 
\begin{align}
R_1&\geq I(X;\hatX_1),\\
R_2&\geq H(Y),\\
R_1+R_2&\geq H(Y)+I(\hat{X}_1;Y)+I(X;\hat{X}_1\hat{X}_2|Y)
\end{align}

\begin{theorem}
\label{mdfirst}
The first-order coding region for the Fu-Yeung problem is
\begin{align}
\nn\calR_{\mathrm{FY}}(D_1,D_2|P_X)
&=\bigcup_{\substack{(P_{\hatX_1|X},P_{\hatX_2|X\hatX_1})\in \calP(P_X,D_1,D_2)}}\calR(P_{\hatX_1|X},P_{\hatX_2|X\hatX_1}).
\end{align}
\end{theorem}
As argued in \cite{fu2002rate}, the Kaspi problem~\cite[Theorem 1]{kaspi1994} and the successive refinement problem~\cite[Theorem 1]{rimoldi1994} are both special cases of the Fu-Yeung problem~(Theorem~\ref{mdfirst}) in the first-order asymptotics for DMSes under bounded distortion measures.

\section{Main Results for the Kaspi Problem}
\label{sec:mainresults4kaspi}
\subsection{Optimal Test Channels}
Throughout the section and its proofs, we are interested in the distortion levels $(D_1,D_2)$ such that $R_{\rmK}(P_{XY},D_1,D_2)$ (cf. \eqref{kaspiratefunc}) is finite. Further, we assume that there exists test channels $(P_{\hatX_1|XY}^*,P_{\hatX_2|XY\hatX_1}^*)$ achieving $R_{\rmK}(P_{XY},D_1,D_2)$. Given any (conditional) distributions $(P_{\hatX_1|XY},P_{\hatX_2|XY\hatX_1})$, let $P_{\hatX_1}$, $P_{X\hatX_1}$, $P_{X\hatX_2}$ and $P_{Y\hatX_1}$, $P_{\hatX_2|Y\hatX_1}$, $P_{Y\hatX_1\hatX_2}$ be induced by $P_{XY}$, $P_{\hatX_1|XY}$ and $P_{\hatX_2|XY\hatX_1}$. For a given distortion pair $(D_1,D_2)$, let $(\lambda_1^*,\lambda_2^*)$ be optimal solutions to the dual problem of $R_{\rmK}(P_{XY},D_1,D_2)$, i.e.,
\begin{align}
\lambda_1^*&:=\frac{\partial R(P_{XY},D,D_2)}{\partial D}\Big|_{D=D_1}\label{dualoptimal1},\\
\lambda_2^*&:=\frac{\partial R(P_{XY},D_1,D)}{\partial D}\Big|_{D=D_2}\label{dualoptimal2}.
\end{align}
Given $(x,y,\hatx_1)$ and arbitrary distributions $(Q_{\hatX_1},Q_{\hatX_2|Y\hatX_1})$, let
\begin{align}
\alpha_2(x,y,\hatx_1|Q_{\hatX_2|Y\hatX_1})
&:=\Big\{\mathbb{E}_{Q_{\hatX_2|Y\hatX_1}}\Big[\exp(-\lambda_2^*d_2(x,\hatX_2))\Big|Y=y,\hatX_1=\hatx_1\Big]\Big\}^{-1},\label{def:a2q4kaspi}\\
\alpha(x,y|Q_{\hatX_1},Q_{\hatX_2|Y\hatX_1})&:=\Bigg\{\mathbb{E}_{Q_{\hatX_1}}\Bigg[\frac{\exp\big(-\lambda_1^*d_1(x,\hatX_1)\big)}{\alpha_2(x,y,\hatX_1|Q_{\hatX_2|Y\hatX_1})}\Bigg]\Bigg\}^{-1}\label{def:a1q4kaspi}.
\end{align}

\begin{lemma}[Properties of Optimal Test Channels]
\label{opttest4kaspi}
A pair of test channels $(P_{\hatX_1|XY}^*,P_{\hatX_2|XY\hatX_1}^*)$ achieves $R_{\rmK}(P_{XY},D_1,D_2)$ if and only if 
\begin{itemize}
\item  For all $(x,y,\hatx_1)$,
\begin{align}
P_{\hatX_1|XY}^*(\hatx_1|x,y)=\frac{\alpha(x,y|P_{\hatX_1}^*,P_{\hatX_2|Y\hatX_1}^*)P_{\hatX_1}^*(\hatx_1)\exp(-\lambda_1^*d_1(x,\hatx_1))}{\alpha_2(x,y,\hatx_1|P_{\hatX_2|Y\hatX_1}^*)}\label{optcond14kaspi},
\end{align}
\item For all $(x,y,\hatx_1,\hatx_2)$ such that $P_{\hatX_1}^*(\hatx_1)>0$,
\begin{align}
P_{\hatX_2|XY\hatX_1}^*(\hatx_2|x,y,\hatx_1)=\alpha_2(x,y,\hatx_1|P_{\hatX_2|Y\hatX_1}^*)P_{\hatX_2|Y\hatX_1}^*(\hatx_2|y,\hatx_1)\exp(-\lambda_1^*d_2(x,\hatx_2))\label{optcond24kaspi}.
\end{align}
\end{itemize}
Further, if $(P_{\hatX_1|XY}^*,P_{\hatX_2|XY\hatX_1}^*)$ achieves $R_{\rmK}(P_{XY},D_1,D_2)$, then we have 
\begin{align}
R_{\rmK}(P_{XY},D_1,D_2)=\mathbb{E}[\log \alpha(X,Y)]-\lambda_1^*D_1-\lambda_2^*D_2\label{kaspipara},
\end{align}
\end{lemma}
The proof of Lemma \ref{opttest4kaspi} is similar to \cite[Properties 1-3]{kostina2012converse}, \cite[Lemma 1]{watanabe2015second}, \cite[Lemma 3]{zhou2016second} and  thus omitted.

Now fix a pair of optimal test channels $(P_{\hatX_1|XY}^*,P_{\hatX_2|XY\hatX_1}^*)$. For simplicity, let
\begin{align}
\alpha_2(x,y,\hatx_1)&:=\alpha_2(x,y,\hatx_1|P_{\hatX_2|Y\hatX_1}^*)\label{def:a24kaspi},\\
\alpha(x,y)&:=\alpha(x,y|P_{\hatX_1}^*,P_{\hatX_2|Y\hatX_1}^*),\label{def:a4kaspi}.
\end{align}
Further, for any $\hatx_1$ and arbitrary distribution $Q_{\hatX_2|Y\hatX_1}$, define
\begin{align}
\nu_1(\hatx_1)&:=\mathbb{E}_{P_{XY}}\left[\frac{\alpha(X,Y)}{\alpha_2(X,Y,\hatx_1)}\exp(-\lambda_1^*d_1(X,\hatx_1))\right]\label{def:nu1},\\
\nu_2(\hatx_1,Q_{\hatX_2|Y\hatX_1})&:=\mathbb{E}_{P_{XY}\times Q_{\hatX_2|Y\hatX_1}}\Big[\alpha(X,Y)\exp(-\lambda_1^*d_1(X,\hatx_1)-\lambda_2^*d_2(X,\hatX_2))\Big|\hatX_1=\hatx_1\Big]\label{def:nu2}.
\end{align}
Then we have the following results regarding the properties of the quantities in \eqref{def:nu1} and \eqref{def:nu2}.
\begin{lemma}
\label{nule1}
For any $\hatx_1$ and arbitrary distribution $Q_{\hatX_2|Y\hatX_1}$, we have
\begin{align}
\nu_1(\hatx_1)&\leq 1,\label{nu1le1},\\
\nu_2(\hatx_1,Q_{\hatX_2|Y\hatX_1})&\leq 1.\label{nu2le1}
\end{align}
\end{lemma}

We remark that Lemma \ref{nule1} holds for both discrete and continuous sources. The proof of Lemma \ref{nule1} is inspired by \cite[Lemma 1.4]{csiszar1974}, \cite[Lemma 5]{tuncel2003comp}, \cite{kostina2017suc} and available in Appendix \ref{proofnule1}. Invoking \eqref{nu2le1} in Lemma \ref{nule1}, as we show in Lemma \ref{oneshotconverse4kaspi}, we can prove a non-asymptotic converse bound for the Kaspi problem.

\subsection{Distortion-Tilted Information Density and Non-Asymptotic Converse Bound} 

Now we introduce the tilted information density which takes a similar role as distortion-tilted information density does for the lossy source coding problem \cite{kostina2012fixed} (See also \cite{ingber2011dispersion}). Recall the definition of $\alpha(\cdot)$ in \eqref{def:a4kaspi}. Define the $(D_1,D_2)$-tilted information density for the Kaspi problem as
\begin{align}
\jmath_{\mathrm{K}}(x,y|D_1,D_2,P_{XY})
&:=\log \alpha(x,y)-\lambda_1^*D_1-\lambda_2^*D_2\label{def:kaspitilt}.
\end{align}
We remark that for all optimal test channels, the value in \eqref{def:kaspitilt} remains the same, which can be verified in a similar manner as \cite[Lemma 2]{watanabe2015second}. The properties of the $(D_1,D_2)$-tilted information density follows from Lemma \ref{opttest4kaspi}. For example, invoking \eqref{optcond14kaspi} and \eqref{optcond24kaspi}, we obtain that for all $(x,y,\hatx_1,\hatx_2)$ such that $P_{\hatX_1}^*(\hatx_1)P_{\hatX_2|Y\hatX_1}^*(\hatx_2|y,\hatx_1)>0$,
\begin{align}
\jmath_{\mathrm{K}}(x,y|D_1,D_2,P_{XY})
\nn&=\log\frac{P_{\hatX_1|XY}^*(\hatx_1|x,y)}{P_{\hatX_1}^*(\hatx_1)}+\log\frac{P_{\hatX_2|XY\hatX_1}^*(\hatx_2|x,y,\hatx_1)}{P_{\hatX_2|Y\hatX_1}^*(\hatx_2|y,\hatx_1)}\\
&\qquad+\lambda_1^*(d_1(x,\hatx_1)-D_1)+\lambda_2^*(d_2(x,\hatx_2)-D_2)\label{expandj4kaspi}.
\end{align}

Define the joint excess-distortion probability as
\begin{align}
\varepsilon_n^{\rmK}(D_1,D_2):=\Pr\Big(d_1(X^n,\hat{X}_1^n)>D_1~\mathrm{or}~d_2(X^n,\hat{X}_2^n)>D_2\Big)\label{defexcessp}.
\end{align}

Invoking Lemma \ref{nule1} and \eqref{def:kaspitilt}, we can prove the following one-shot converse bound for the Kaspi problem.

\begin{lemma}
\label{oneshotconverse4kaspi}
For any $(1,M)$-code, we have that for any $\gamma\geq 0$,
\begin{align}
\varepsilon_1^{\rmK}(D_1,D_2)
&\geq \Pr\left(\jmath_{\mathrm{K}}(X,Y|D_1,D_2,P_{XY})\geq \log M+\gamma\right)-\exp(-\gamma).
\end{align}
\end{lemma}
The proof of Lemma \ref{oneshotconverse4kaspi} is given in Appendix \ref{proofoneshot}. We remark that Lemma \ref{oneshotconverse4kaspi} is a generalization of the one-shot bounds by Kostina and Verd\'u for the lossy source coding~\cite[Theorem 1]{kostina2012converse} and the lossy source coding with encoder and decoder side information~\cite[Theorem 4]{kostina2012converse}.

For any memoryless source, from the definition in \eqref{def:kaspitilt}, we obtain that
\begin{align}
\jmath_{\rmK}(X^n,Y^n|D_1,D_2,P_{XY}^n)
&=\sum_{i=1}^n \jmath_{\rmK}(X_i,Y_i|D_1,D_2,P_{XY}).
\end{align}
Hence, invoking Lemma \ref{oneshotconverse4kaspi}, we obtain that for any $\gamma\geq 0$,
\begin{align}
\varepsilon_n^{\rmK}(D_1,D_2)
&\geq \Pr\left(\sum_{i=1}^n\jmath_{\rmK}(X,Y|D_1,D_2,P_{XY})\geq \log M+n\gamma\right)-\exp(-n\gamma)\label{nshot}.
\end{align}
We remark that the bound in \eqref{nshot} plays a central role in the converse proof the second-order and moderate deviations asymptotics. The bound in \eqref{nshot} holds not only for DMSes but also for any continuous memoryless sources.

\subsection{Second-Order Asymptotics for DMSes}
In this subsection, we formulate and present second-order asymptotics of the Kaspi problem for DMSes under bounded distortion measures, i.e., $\calX$, $\calY$, $\hat{\calX_1}$, $\hat{\calX_2}$ are all finite and $\max_{x,\hatx_i}d_i(x,\hatx_i)~,i=1,2$ is finite. Let $\epsilon\in(0,1)$ be fixed.
\begin{definition}
\label{secondorder}
A rate $L$ is said to be second-order $(D_1,D_2,\epsilon)$-achievable for the Kaspi problem if there exists a sequence of $(n,M)$-codes such that
\begin{align}
\limsup_{n\to\infty} \frac{\log M-nR_{\rmK}(P_{XY},D_1,D_2)}{\sqrt{n}}\leq L,
\end{align}
and
\begin{align}
\limsup_{n\to\infty} \varepsilon_n^{\rmK}(D_1,D_2)\leq \epsilon.
\end{align}
The infimum second-order $(D_1,D_2,\epsilon)$-achievable rate is called the optimal second-order coding rate and denoted as $L^*(D_1,D_2,\epsilon)$.
\end{definition}

The central goal for this subsection is to characterize $L^*(D_1,D_2,\epsilon)$. Note that in Definition \ref{firstorder}, the average distortion criterion is used while in Definition \ref{secondorder}, the excess-distortion probability is considered. The reason is that for second-order asymptotics, we need a probability to quantify. To be specific, the excess-distortion probability plays a similar role as error probability for the lossless source coding problem~\cite{hayashi2008} or channel coding problems~\cite{hayashi2009information,polyanskiy2010thesis}. Let $\rmV(D_1,D_2|P_{XY})$ be the distortion-dispersion function for the Kaspi problem, i.e., 
\begin{align}
\rmV(D_1,D_2|P_{XY})&:=\mathrm{Var}\big[\jmath_{\rmK}(X,Y|D_1,D_2,P_{XY})\big]\label{dispersion4kaspi}.
\end{align}

We impose following conditions:
\begin{enumerate}
\item\label{kaspi:cond1} The distortion levels are chosen such that $R_\rmK(P_{XY},D_1,D_2)>0$ is finite;
\item\label{kaspi:condend} $Q_{XY}\to R_\rmK(Q_{XY},D_1,D_2)$ is twice differentiable in the neighborhood of $P_{XY}$ and the derivative is bounded.
\end{enumerate}
\begin{theorem}
\label{kaspisecond}
Under conditions (\ref{kaspi:cond1}) and (\ref{kaspi:condend}), the optimal second-order coding rate for the Kaspi problem is
\begin{align}
L^*(D_1,D_2,\epsilon)=\sqrt{\rmV(D_1,D_2|P_{XY})}\rmQ^{-1}(\epsilon).
\end{align}
\end{theorem}
The proof of Theorem \ref{kaspisecond} is given in Section \ref{sec:proof4kaspi}. In the converse part, we apply the Berry-Esseen Theorem to the non-asymptotic bound in \eqref{nshot}.  In the achievability part, we first prove a type-covering lemma tailored for the Kaspi problem. Subsequently, we make use of the properties of $\jmath_{\rmK}(x,y|D_1,D_2,P_{XY})$ in Lemma \ref{opttest4kaspi}, appropriate Taylor expansions, and the Berry-Esseen Theorem.

We remark that the distortion-tilted information density for the Kaspi problem $\jmath_{\rmK}(x,y|D_1,D_2,P_{XY})$ reduces to the $D_1$-tilted information density for the lossy source coding problem~\cite{kostina2012converse} or the $D_2$-tilted information density for the lossy source coding problem with encoder and decoder side information~\cite{kostina2012converse} for particular choices of distortion levels $(D_1,D_2)$. Hence, our result in Theorem \ref{kaspisecond} is a strict generalization of the second-order coding rate for the lossy source coding problem~\cite{kostina2012} and the conditional lossy source coding problem~\cite{le2014second} for DMSes under bounded distortion measures. We illustrate this point in Section \ref{sec:numerical} via numerical examples.

\subsection{Large and Moderate Deviations Asymptotics for DMSes}
The optimal error exponent function is defined as follows.
\begin{definition} 
\label{defee}
{\em  A non-negative number $E$ is said to be an {\em $(R,D_1,D_2)$-achievable error exponent} if there exists a sequence of $(n,M)$-codes such that,
\begin{align}
\limsup_{n\to\infty}\frac{1}{n}\log M&\leq  R,\\
\liminf_{n\to\infty}-\frac{\log \varepsilon_n^\rmK(D_1,D_2)}{n}&\geq E.
\end{align}
The supremum of all $(R,D_1,D_2)$-achievable error exponents is called the optimal error exponent and denoted as $E^*(R,D_1,D_2)$.}
\end{definition}
Define the function
\begin{align}
\label{defeef}
F(P_{XY},R,D_1,D_2):=\inf_{Q_{XY}:R_\rmK(Q_{XY},D_1,D_2)\geq R} D(Q_{XY}\|P_{XY}).
\end{align}
\begin{theorem}
\label{eekaspi}
The optimal error exponent for the Kaspi problem is
\begin{align}
E^*(R,D_1,D_2)&=F(P_{XY},R,D_1,D_2).
\end{align}
\end{theorem}
The proof of Theorem \ref{eekaspi} is omitted due to similarity to \cite[Theorem 12]{zhou2015second}. The achievability part follows by invoking the type-covering Lemma (cf. Lemma \ref{coveringkaspi}) while the converse part follows from a strong converse result implied by \eqref{nshot} and the change-of-measure technique.

In the following, we define and present our result for the moderate deviations constant.
\begin{definition}
{\em  
Consider any sequence $\{\rho_n\}_{n\in\bbN}$ satisfying
\begin{align}
 \lim_{n\to\infty}\rho_n&=0,~\lim_{n\to\infty}\sqrt{n}\rho_n=\infty. \label{eqn:cond_rho}
\end{align}
A non-negative number $\nu$ is said to be a {\em $(R,D_1, D_2)$-achievable moderate deviations constant (with respect to $\{ \rho_n\}_{n=1}^{\infty}$)} if there exists a sequence of $(n,M)$-codes such that
\begin{align}
\limsup_{n\to\infty}\frac{1}{n\rho_n}(\log M-nR_\rmK(P_{XY},D_1,D_2))\leq 1,\\
\liminf_{n\to\infty}-\frac{\log \varepsilon^\rmK_n(D_1,D_2)}{n\rho_n^2}\geq \nu.
\end{align}
The supremum of all $(R,D_1, D_2)$-achievable moderate deviations constants is called the optimal moderate deviations constant and denoted as $\nu^*(R,D_1, D_2)$.
}
\end{definition}

\begin{theorem}
\label{theoremmdc}
Given the conditions in Theorem \ref{kaspisecond} and the assumption that the distortion-dispersion function $\mathrm{V}(D_1,D_2|P_{XY})$ is positive, the optimal moderate deviations constant for the Kaspi problem is
\begin{align}
\nu^*(R,D_1, D_2)=\frac{1}{2\mathrm{V}(D_1,D_2|P_{XY})}.
\end{align}
\end{theorem}
We observe that similarly to second order asymptotics (Theorem \ref{kaspisecond}), the distortion-dispersion function $\mathrm{V}(D_1,D_2|P_{XY})$ (cf. \eqref{dispersion4kaspi}) is a fundamental quantity that governs the speed of convergence of the excess-distortion probability to zero. The proof of Theorem \ref{theoremmdc} is done by applying the moderate deviations principle/theorem (cf. \cite[Theorem~3.7.1]{dembo2009large}) to the information spectrum bounds in the second-order asymptotics and omitted here due to similarity to \cite[Theorem 7]{zhou2015second}. We remark that the proof of Theorem \ref{theoremmdc} can also be done in a similar manner as \cite{tan2012moderate} using Euclidean information theory \cite{borade2008}.

\subsection{Numerical Examples}
\label{sec:numerical}
In this subsection, we consider two DMSes and Hamming distortion measures.

\subsubsection{Asymmetric Correlated Source}

In order to illustrate our results in Lemma \ref{opttest4kaspi} and Theorem \ref{kaspisecond}, we consider the following source. Let $\calX=\{0,1\}$, $\calY=\{0,1,\rme\}$ and $P_X(0)=P_X(1)=\frac{1}{2}$. Let $Y$ be the output of passing $X$ through a Binary Erasure Channel (BEC) with erasure probability $p$, i.e., $P_{Y|X}(y|x)=1-p$ if $x=y$ and $P_{Y|X}(\rme|x)=p$. The explicit formula of the Kaspi rate-distortion function for the above correlated source under Hamming distortion measures was derived by Perron, Diggavi and Telatar in \cite{perron2006kaspi}. Here we only recall the \emph{non-degenerate} result, i.e., the case where the distortion levels $(D_1,D_2)$ are chosen such that $\lambda_1^*>0$ and $\lambda_2^*>0$.

Define the set
\begin{align}
\calD_{\mathrm{bec}}:=
\Bigg\{(D_1,D_2)\in\bbR_+^2: D_1\leq \frac{1}{2},~D_1-\frac{1-p}{2}\leq D_2\leq p D_1
\Bigg\}
\end{align}
\begin{lemma}
If $(D_1,D_2)\in\calD_{\mathrm{bec}}$, then the Kaspi rate-distortion function for the above asymmetric correlated source under Hamming distortion measures is
\begin{align}
R_{\rmK}(P_{XY},D_1,D_2)&=\log 2-(1-p)H_b\Bigg(\frac{D_1-D_2}{1-p}\Bigg)-pH_b\Bigg(\frac{D_2}{p}\Bigg).
\end{align}
\end{lemma}
Hence, for $(D_1,D_2)\in\calD_{\mathrm{bec}}$, using the definitions of $\lambda_1^*$ in \eqref{dualoptimal1} and $\lambda_2^*$ in \eqref{dualoptimal2}, we obtain
\begin{align}
\lambda_1^*&=\log\frac{(1-p)-(D_1-D_2)}{1-p}-\log\frac{D_1-D_2}{1-p}=\log\frac{(1-p)-(D_1-D_2)}{D_1-D_2}\label{call1*},\\
\lambda_2^*&=\log\frac{p-D_2}{p}+\log\frac{D_1-D_2}{1-p}-\log\frac{(1-p)-(D_1-D_2)}{1-p}-\log\frac{D_2}{p}\\
&=-\lambda_1^*+\log\frac{p-D_2}{D_2}\label{call2*}.
\end{align}
Then, using the definitions of $\alpha_2(\cdot)$ in \eqref{def:a24kaspi} and $\alpha(\cdot)$ in \eqref{def:a4kaspi}, we have
\begin{align}
\alpha_2(0,0,0)&=\alpha_2(0,0,1)=\alpha_2(1,1,0)=\alpha_2(1,1,1)=\alpha_2(0,\rme,0)=\alpha_2(1,\rme,1)=1,\\
\alpha_2(1,0,0)&=\alpha_2(1,0,1)=\alpha_2(0,1,0)=\alpha_2(0,1,1)=\alpha_2(1,\rme,0)=\alpha_2(0,\rme,1)=\exp(\lambda_2^*),
\end{align}
and
\begin{align}
\alpha(0,0)&=\alpha(1,1)=\frac{2}{1+\exp(-\lambda_1^*)},\\
\alpha(0,\rme)&=\alpha(1,\rme)=\frac{2}{1+\exp(-\lambda_1^*-\lambda_2^*)}.
\end{align}

It can be verified easily that \eqref{optcond14kaspi}, \eqref{optcond24kaspi}, \eqref{kaspipara} and \eqref{nu1le1} hold. In the following, we will verify that \eqref{nu2le1} holds for arbitrary $Q_{\hatX_2|Y\hatX_1}$ and $\hatx_1$. As a first step, we can verify that for any $(y,\hatx_1,\hatx_2)$, we have
\begin{align}
\sum_x P_{XY}(x,y)\alpha(x,y)\exp(-\lambda_1^*d_1(x,\hatx_1)-\lambda_2^*d_2(x,\hatx_2))
&\leq \sum_x P_{XY}(x,y)\frac{\alpha(x,y)}{\alpha_2(x,y,\hatx_1)}\exp(-\lambda_1^*d_1(x,\hatx_1))\label{verify1}.
\end{align}
Then, for any distribution $Q_{\hatX_2|Y\hatX_1}$, using the definitions of $\nu_1(\cdot)$ in \eqref{def:nu1} and $\nu_2(\cdot)$ in \eqref{def:nu2}, multiplying $Q_{\hatX_2|Y\hatX_1}(\hatx_2|y,\hatx_1)$ over both sides of \eqref{verify1}, and summing over $(y,\hatx_2)$, we obtain that
\begin{align}
\nu_2(\hatx_1,Q_{\hatX_2|Y\hatX_1})\leq \nu_1(\hatx_1)\leq 1\label{verify2},
\end{align}
where the second inequality in \eqref{verify2} follows from \eqref{nu1le1}.

Using the definition of $\jmath_{\rmK}(\cdot)$ in \eqref{def:kaspitilt}, we have
\begin{align}
\jmath_{\rmK}(0,0|D_1,D_2,P_{XY})
&=\jmath_{\rmK}(1,1|D_1,D_2,P_{XY})\\
&=\log \alpha(0,0)-\lambda_1^*D_1-\lambda_2^*D_2\label{j1cal4kaspi},
\end{align}
and
\begin{align}
\jmath_{\rmK}(0,\rme|D_1,D_2,P_{XY})
&=\jmath_{\rmK}(1,\rme|D_1,D_2,P_{XY})\\
&=\log \alpha(0,\rme)-\lambda_1^*D_1-\lambda_2^*D_2\label{j2cal4kaspi}.
\end{align}

Further, using the definition the distortion-dispersion function $\rmV(D_1,D_2|P_{XY})$ in \eqref{dispersion4kaspi}, we have
\begin{align}
\rmV(D_1,D_2|P_{XY})
&=\mathrm{Var}[\jmath_\rmK(X,Y|D_1,D_2,P_{XY})]\\
&=\mathrm{Var}[\log \alpha(X,Y)]\\
\nn&=(1-p)\Big(\log\alpha(0,0)-(1-p)\log\alpha(0,0)-p\log\alpha(0,\rme)\Big)^2\\
&\qquad+p\Big(\log\alpha(0,\rme)-(1-p)\log\alpha(0,0)-p\log\alpha(0,\rme)\Big)^2\\
&=[(1-p)p^2+p(1-p)^2]\Big(\log \alpha(0,0)-\log\alpha(0,\rme)\Big)^2\\
&=p(1-p)\Bigg(\log\frac{p-D_2}{p}-\log\frac{(1-p)-(D_1-D_2)}{1-p}\Bigg)^2.
\end{align}

Define $R_n(D_1,D_2,\epsilon):=R_{\rmK}(P_{XY},D_1,D_2)+\frac{L^*(D_1,D_2,\epsilon)}{\sqrt{n}}$. In order to illustrate out result, we plot $R_n(D_1,D_2,\epsilon)$ for $p=0.3$, $D_1=0.2$, $D_2=0.05$ and $\epsilon=0.05,0.5,0.95$ in Figure \ref{asymsecond}. Note that with this choice, $(D_1,D_2)\in\calD_{\mathrm{bec}}$ and $R_n(D_1,D_2,0.5)=R_{\rmK}(P_{XY},D_1,D_2)$.
\begin{figure}
\centering
\includegraphics[width=12cm]{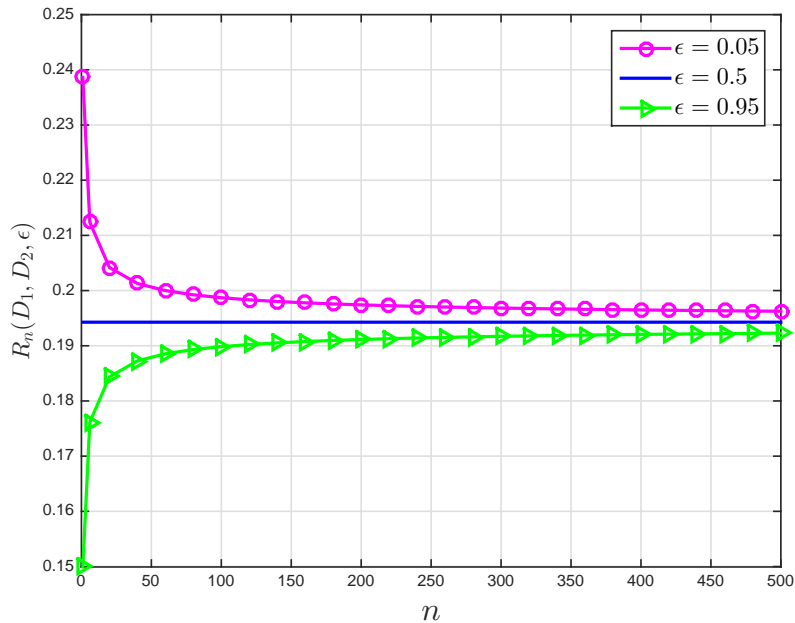}
\caption{Rate $R_n(D_1,D_2,\epsilon)$ when $p=0.3$, $D_1=0.2$ and $D_2=0.05$.}
\label{asymsecond}
\end{figure}

\subsubsection{Doubly Symmetric Binary Source (DSBS)}
In this example, we show that under certain distortion levels, the Kaspi rate-distortion function reduces to the rate-distortion function~\cite{shannon1959coding} (see also \cite[Theorem 3.5]{el2011network}) and the conditional rate-distortion function~\cite[Eq. (11.2)]{el2011network}. We consider the DSBS where $\calX=\calY=\{0,1\}$, $P_{XY}(0,0)=P_{XY}(1,1)=\frac{1-p}{2}$ and $P_{XY}(0,1)=P_{XY}(1,0)=\frac{p}{2}$ for some $p\in[0,\frac{1}{2}]$. 
\begin{lemma}
\label{kaspidsbs}
Depending on the distortion levels $(D_1,D_2)$, the Kaspi rate-distortion function for the DSBS with Hamming distortion measures satisfies
\begin{itemize}
\item $D_1\geq \frac{1}{2}$ and $D_2\geq p$
\begin{align}
R_{\rmK}(P_{XY},D_1,D_2)=0.
\end{align}
\item $D_1<\frac{1}{2}$ and $D_2\geq \min\{p,D_1\}$
\begin{align}
R_{\rmK}(P_{XY},D_1,D_2)=\log 2-H_b(D_1),
\end{align}
where $H_b(x)=-x\log x-(1-x)\log (1-x)$ is the binary entropy function.
\item $D_1\geq D_2+\frac{1-2p}{2}$ and $D_2<p$ 
\begin{align}
R_{\rmK}(P_{XY},D_1,D_2)&=H_b(p)-H_b(D_2).
\end{align}
\end{itemize}
\end{lemma}

When $D_1<\frac{1}{2}$ and $D_2<pD_1$, the Kaspi rate-distortion function reduces to the rate-distortion function for the lossy source coding problem. Thus, the distortion-tilted information density for the Kaspi problem reduces to the $D_1$-tilted information density~\cite[Definition 1]{kostina2012converse}, i.e.,
\begin{align}
\jmath_{\rmK}(x,y|D_1,D_2,P_{XY})&=\log 2-H_b(D_1)\label{def:d1tilt}.
\end{align} 
Hence, $L^*(D_1,D_1|P_{XY})=0$. When $D_1\geq D_2+\frac{1-2p}{2}$ and $D_2<p$, the Kaspi rate-distortion function reduces to the conditional rate-distortion function. Under the optimal test channel, we have $\hatX_1=0/1$ and $X\to\hatX_2\to Y$ forms a Markov chain. In this case, the distortion-tilted information density for the Kaspi problem reduces to the conditional distortion-tilted information density~\cite[Definition 5]{kostina2012converse} (see also \cite{le2014second}), i.e., 
\begin{align}
\jmath_{\rmK}(x,y|D_1,D_2,P_{XY})
&=-\log P_{X|Y}(x|y)-H_b(D_2).
\end{align}
Hence,
\begin{align}
\rmV(D_1,D_2|P_{XY})
&=\mathrm{Var}[-\log P_{X|Y}(X|Y)]\\
&=(1-p)(-\log(1-p)-H_b(p))^2+p(-\log p-H_b(p))^2:=\rmV(p),
\end{align}
and
\begin{align}
L^*(D_1,D_2,\epsilon)=\sqrt{\rmV(p)}\rmQ^{-1}(\epsilon).
\end{align}

\section{Main Results for the Fu-Yeung Problem}
\label{sec:mainresults4fy}
\subsection{Preliminaries}

Throughout the section and its proofs, given any distributions $(P_{\hatX_1|X},P_{\hatX_2|X\hatX_1})$, let $P_{XY}$, $P_{X|Y}$, $P_{\hatX_1}$, $P_{Y\hatX_1}$, $P_{X\hatX_1}$, $P_{X\hatX_2}$, $P_{XY\hatX_1}$, $P_{\hatX_1|XY}$ and $P_{\hatX_2|Y\hatX_1}$ be induced by $P_X$, $P_{\hatX_1|X}$, $P_{\hatX_2|X\hatX_1}$ and the deterministic function $g:\calX\to\calY$. Recalling the definition of $\calP(P_X,D_1,D_2)$ above Theorem \ref{mdfirst}. For subsequent analyses, we need the following definition
\begin{align}
\rvR_{\rm{FY}}(R_1,D_1,D_2|P_X)&:=\min_{\substack{(P_{\hatX_1|X},P_{\hatX_2|X\hatX_1})\\\in\calP(P_X,D_1,D_2):\\R_1\geq I(X;\hatX_1)}} I(\hatX_1;Y)+I(X;\hatX_1\hatX_2|Y)\label{rvrmin}.
\end{align} 
Invoking Theorem \ref{firstorder}, we conclude that given $R_1$, the minimum achievable sum rate is given by $\rvR_{\rm{FY}}(R_1,D_1,D_2|P_X)+H(P_Y)$. Hence, any optimal test channel for $\rvR_{\rm{FY}}(R_1,D_1,D_2|P_X)$ is also an optimal test channel for the minimum sum rate function in the Fu-Yeung problem. From Theorem \ref{mdfirst}, we also conclude that the minimum rate for encoder $f_1$ is the rate-distortion function $R(P_X,D_1)$~\cite{cover2012elements}, i.e.,
\begin{align}
R(P_X,D_1)=\min_{P_{\hatX_1|X}:\mathbb{E}_{P_X\times P_{\hatX_1|X}}[d_1(X,\hatX_1)\leq D_1]}I(X;\hatX_1)\label{def:ratedisd1},
\end{align}
and the minimum rate for encoder $f_2$ is the entropy $H(P_Y)$. With these three observations regarding the minimum rates and minimum sum-rates, we obtain boundary rate pairs for second-order asymptotics in Theorem \ref{mddsecondregion}.

Note that \eqref{rvrmin} is a convex optimization problem. Let $(s^*,t_1^*,t_2^*)$ be the optimal solutions to the dual problem, i.e.,
\begin{align}
s^*:=-\frac{\partial \rvR_{\mathrm{FY}}(R,D_1,D_2|P_X)}{\partial R}\Bigg|_{R=R_1}\label{def:s*},\\
t_1^*:=-\frac{\partial \rvR_{\mathrm{FY}}(R_1,D,D_2|P_X)}{\partial D}\Bigg|_{D=D_1}\label{def:t1*},\\
t_2^*:=-\frac{\partial \rvR_{\mathrm{FY}}(R_1,D_1,D|P_X)}{\partial D}\Bigg|_{D=D_2}\label{def:t2*}.
\end{align}

Given distributions $(Q_{\hatX_1},Q_{\hatX_2|Y\hatX_1})$ and $(x,y,\hatx_1)$, define
\begin{align}
\beta_2(x,y,\hatx_1|Q_{\hatX_2|Y\hatX_1})
&:=\Big\{\mathbb{E}_{Q_{\hatX_2|Y\hatX_1}}\Big[\exp(-t_2^*d_2(x,\hatX_2))\big|Y=y,\hatX_1=\hatx_1\Big]\Big\}^{-1},\label{def:b2q}\\
\beta(x,y|Q_{\hatX_1},Q_{\hatX_2|Y\hatX_1})&:=\Bigg\{\mathbb{E}_{Q_{\hatX_1}}\Bigg[\exp\Bigg(-\frac{t_1^*d_1(x,\hatX_1)+\log\beta_2(x,y,\hatX_1|Q_{\hatX_2|Y\hatX_1})}{1+s^*}\Bigg)\Bigg]\Bigg\}^{-1}\label{def:bq}.
\end{align}

Define the bivariate generalization of the Gaussian cdf as follows:
\begin{align}
\Psi(x,y,\bmu,\mathbf{\Sigma})&:=\int_{-\infty}^{x}\int_{-\infty}^{y}\calN(\bx; \bmu;\bSigma)\, \rmd \bx.
\end{align}
Here, $\calN(\bx; \bmu;\bSigma)$ is the pdf of a bivariate Gaussian with mean $\bmu$ and covariance matrix $\bSigma$~\cite[Chapter 1]{Tanbook}. 


\subsection{Optimal Test Channels}
We first present the properties of the optimal test channels achieving \eqref{rvrmin}.
\begin{lemma}
\label{proptest4fy}
A pair of test channels $(P^*_{\hatX_1|X},P_{\hatX_2|X\hatX_1}^*)$ achieves $\rvR_{\rm{FY}}(R_1,D_1,D_2|P_X)$ if and only if
\begin{itemize}
\item For all $(x,y,\hatx_1,\hatx_2)$ such that $y=g(x)$,
\begin{align}
P_{\hatX_1|X}^*(\hatx_1|x)&=\beta(x,y|P_{\hatX_1}^*,P_{\hatX_2|Y\hatX_1}^*)P^*_{\hatX_1}(\hatx_1)\exp\Bigg(-\frac{t_1^*d_1(x,\hatx_1)+\log\beta_2(x,y,\hatx_1|P^*_{\hatX_2|Y\hatX_1})}{1+s^*}\Bigg)\label{optcond2},
\end{align}
\item For all $(x,y,\hatx_1,\hatx_2)$ such that $y=g(x)$ and $P_{\hatX_1|X}^*(\hatx_1|x)>0$
\begin{align}
P_{\hatX_2|X\hatX_1}^*(\hatx_2|x,\hatx_1)&=\beta_2(x,y,\hatx_1|P_{\hatX_2|Y\hatX_1}^*)P^*_{\hatX_2|Y\hatX_1}(\hatx_2|y,\hatx_1)\exp(-t_2^*d_2(x,\hatx_2))\label{optcond1}.
\end{align}
\item For all $(x,\hatx_1,\hatx_2)$ such that $P_{\hatX_1|X}^*(\hatx_1|x)=0$, $P_{\hatX_2|X\hatX_1}^*(\cdot|x,\hatx_1)$ can be arbitrary distribution.
\end{itemize} 
Further, if a pair of channels $(P^*_{\hatX_1|X},P_{\hatX_2|X\hatX_1}^*)$ achieves $\rvR_{\rm{FY}}(R_1,D_1,D_2)$, then we have
\begin{itemize}
\item The parametric representation of $\rvR_{\rm{FY}}(R_1,D_1,D_2|P_X)$ is
\begin{align}
\rvR_{\rm{FY}}(R_1,D_1,D_2|P_X)
&=(1+s^*)\mathbb{E}_{P_{XY}}[\log\beta(X,Y|P_{\hatX_1}^*,P_{\hatX_2|Y\hatX_1^*})]-s^*R_1-t_1^*D_1-t_2^*D_2\label{para4fy},
\end{align}
\item For $(x,y,\hatx_1,\hatx_2)$ such that $y=g(x)$ and $P_{\hatX_1}^*(\hatx_1)P_{\hatX_2|Y\hatX_1}^*(\hatx_2|g(x),\hatx_1)>0$, we have
\begin{align}
\nn&(1+s^*)\log \beta(x,y|P_{\hatX_1}^*,P_{\hatX_2|Y\hatX_1}^*)\\
&=(1+s^*)\log\frac{P_{\hatX_1|X}^*(\hatx_1|x)}{P_{\hatX_1}^*(\hatx_1)}+t_1^*d_1(x,\hatx_1)+\log\frac{P_{\hatX_2|X\hatX_1}^*(\hatx_2|x,\hatx_1)}{P_{\hatX_2|Y\hatX_1}^*(\hatx_2|y,\hatx_1)}+t_2^*d_2(x,\hatx_2)\label{expantypical}.
\end{align}
\end{itemize}
\end{lemma}
The proof of Lemma \ref{proptest4fy} is similar to \cite[Lemma 1.4]{csiszar1974}, \cite[Lemma 3]{watanabe2015second}, \cite[Lemma 3]{zhou2016second} and given in Appendix \ref{proofproptest4fy} for completeness.

Now fix a pair of optimal test channels $(P_{\hatX_1}^*,P_{\hatX_2|Y\hatX_1}^*)$. Given $(x,y,\hatx_1)$, for simplicity, define
\begin{align}
\beta_2(x,y,\hatx_1)
&:=\beta_2(x,y,\hatx_1|P_{\hatX_2|Y\hatX_1}^*),\label{def:b2}\\
\beta(x,y)&:=\beta(x,y|P_{\hatX_1}^*,P_{\hatX_2|Y\hatX_1}^*)\label{def:b}\\\imath_1(x,y,\hatx_1)&=\log \beta(x,y)-\frac{1}{1+s^*}\log\beta_2(x,y,\hatx_1)\label{def:i14fy}\\
\imath_2(x,y,\hatx_1)&:=\log \beta(x,y)+\frac{s^*}{1+s^*}\log \beta_2(x,y,\hatx_1)\label{def:i24fy}.
\end{align}
Further, given any $\hatx_1$ and arbitrary conditional distribution $Q_{\hatX_2|Y\hatX_1}$, define
\begin{align}
w_1(\hatx_1)
&:=\mathbb{E}_{P_{XY}}\bigg[\exp\bigg(\imath_1(X,Y,\hatx_1)-\frac{t_1^*d_1(X,\hatx_1)}{1+s^*}\bigg)\bigg]\label{def:w1},\\
w_2(\hatx_1,Q_{\hatX_2|Y\hatX_1})&:=\mathbb{E}_{P_{XY}\times Q_{\hatX_2|Y\hatX_1}}\bigg[\exp\big(\imath_2(X,Y,\hatx_1)-\frac{t_1^*}{1+s^*}d_1(X,\hatx_1)-t_2^*d_2(X,\hatX_2)\big)\Big|\hatX_1=\hatx_1\bigg].\label{def:w2}
\end{align}
Similarly as \cite{watanabe2015second}, we can show that, for any pair of optimal test channels $(P_{\hatX_1|X}^*,P_{\hatX_2|X\hatX_1}^*)$, the value of $\beta(x,y)$ remains unchanged and hence it is well defined. From now on, fix a pair of test channels $(P^*_{\hatX_1|X},P_{\hatX_2|X\hatX_1}^*)$ such that that i) \eqref{optcond2}, \eqref{optcond1} hold; ii) for any $(y,\hatx_1)$ such that $P_{Y\hatX_1}^*(y,\hatx_1)=0$, the induced distribution defined as $P_{\hatX_2|Y\hatX_1}^*(\hatx_2|y,\hatx_1):=\sum_x P_X(x)1\{y=g(x)\}P_{\hatX_2|X\hatX_1}^*(\hatx_2|x,\hatx_1)$ satisfies
\begin{align}
P_{\hatX_2|Y\hatX_1}^*=\argsup_{Q_{\hatX_2|Y\hatX_1}} \mathbb{E}_{P_{X|y}}\bigg[\beta(X,y)\beta_2^{-\frac{1}{1+s^*}}(X,y,\hatx_1|Q_{\hatX_2|Y\hatX_1})\exp\Big(-\frac{t_1^*}{1+s^*}d_1(X,\hatx_1)\Big)\bigg]\label{chooseopt}.
\end{align}
We remark that the choice of $P_{\hatX_2|X\hatX_1}^*$ satisfying \eqref{chooseopt} is possible due to the fact that sets $\{x:g(x)=y\}$ is disjoint for each $y\in\calY$.

In the following, we present an important property of the quantities in \eqref{def:w1} and \eqref{def:w2}.
\begin{lemma}
\label{wleq1}
For the given pair of test channel $(P^*_{\hatX_1|XY},P_{\hatX_2|X\hatX_1}^*)$ satisfying \eqref{optcond2}, \eqref{optcond1}, and \eqref{chooseopt}, we have that for any $\hatx_1\in\hat{\calX}_1$ and arbitrary distribution $Q_{\hatX_2|Y\hatX_1}$, 
\begin{align}
w_2(\hatx_1,Q_{\hatX_2|Y\hatX_1})\leq w_1(\hatx_1)\leq 1\label{w2leqw1le1}.
\end{align}
\end{lemma}
The proof of Lemma \ref{wleq1} is inspired by \cite[Lemma 5]{tuncel2003comp}, \cite{kostina2017suc} and available in Appendix \ref{proofwleq1}. We remark that Lemma \ref{proptest4fy} and Lemma \ref{wleq1} hold for abstract source, not restricted to DMSes. For simplicity, we only prove the results in Lemma \ref{proptest4fy} and Lemma \ref{wleq1} for DMSes. However, it can be extended easily to continuous sources by replacing the log-sum inequality with its continuous version~\cite[Lemma 2.14]{donsker1975asymptotic}. As we shall show later in Lemma \ref{oneshotconverse4fy}, the result in Lemma \ref{wleq1} leads to a non-asymptotic converse bound for the Fu-Yeung problem.

\subsection{Rate-Distortion-Tilted Information Density and Non-Asymptotic Converse Bound}
Recall that $P_{XY}$ is induced by $P_X$ and the deterministic function $g:\calX\to\calY$. Given $(R_1,D_1,D_2)$, using the definition of $\beta(\cdot)$ in \eqref{def:b}, we define the rate-distortion-tilted information density as follows:
\begin{align}
\jmath_{\rm{FY}}(x,y|R_1,D_1,D_2,P_X):=(1+s^*)\log \beta(x,y)-s^*R_1-t_1^*D_1-t_2^*D_2\label{def:j4fy}.
\end{align}
The properties of $\jmath_{\rm{FY}}(x,y|R_1,D_1,D_2,P_X)$ follows from Lemma \ref{proptest4fy}. For example, invoking \eqref{para4fy}, we obtain that
\begin{align}
\rvR_{\rm{FY}}(R_1,D_1,D_2|P_X)
&=\mathbb{E}_{P_{XY}}[\jmath_{\rm{FY}}(X,Y|R_1,D_1,D_2,P_X)]\\
&=\mathbb{E}_{P_X}[\jmath_{\rm{FY}}(X,g(X)|R_1,D_1,D_2,P_X)].
\end{align}
Let $\jmath(x,D_1|P_X)$ be the $D_1$-tilted information density~\cite{kostina2012fixed}, i.e.,
\begin{align}
\jmath(x,D_1|P_X):=-\log \Big(\sum_{\hatx_1}P_{\hatX_1}^*(\hatx_1)\exp(-t^*(d_1(x,\hatx_1)-D_1))\Big)\label{djkostina},
\end{align}
where $P_{\hatX_1}^*$ is induced by the source distribution $P_X$ and the optimal test channel $P_{\hatX_1|X}^*$ for the rate-distortion function $R(P_X,D_1)$ (cf. \eqref{def:ratedisd1}) and $t^*=-\frac{\partial R(P_X,D)}{\partial D}|_{D=D_1}$.

Define the joint error and excess-distortion probability for the Fu-Yeung problem as
\begin{align}
\varepsilon_n^{\rm{FY}}(D_1,D_2):=\Pr(d_1(X^n,\hatX_1^n)>D_1~\mathrm{or}~d_2(X^n,\hatX_2^n)>D_2~\mathrm{or}~\hatY^n\neq Y^n)\label{def:error4fy}.
\end{align}
In the following, we will present a non-asymptotic converse bound for the Fu-Yeung problem. Given any $\gamma\geq 0$, define the following sets:
\begin{align}
\calA_1&:=\{(x,y):\jmath(x,D_1|P_X)\geq \log M_1+\gamma\},\\
\calA_2&:=\{(x,y):-\log P_Y(y)\geq \log M_2+\gamma\},\\
\calA_3&:=\{(x,y):\jmath_{\rm{FY}}(x,y|R_1,D_1,D_2,P_X)\geq \log M_1M_2+s^*\log M_1+(1+s^*)\gamma\}.
\end{align}
\begin{lemma}
\label{oneshotconverse4fy}
Any $(1,M_1,M_2)$-code for the Fu-Yeung problem satisfies that for any $\gamma\geq 0$,
\begin{align}
\varepsilon_1^{\rm{FY}}(D_1,D_2)
&\geq \Pr((X,Y)\in(\calA_1\cup\calA_2\cup\calA_3))-4\exp(-\gamma).
\end{align}
\end{lemma}
The proof of Lemma \ref{oneshotconverse4fy} is given in Appendix \ref{proofoneshot4fy}. We remark that Lemma \ref{wleq1} plays an important role in the proof of Lemma \ref{oneshotconverse4fy}. This can be made clear by the following definitions. Given $(x,y,\hatx_1,\hatx_2)$, using the definitions of $\imath_1(\cdot)$ in \eqref{def:i14fy} and $\imath_2(\cdot)$ in \eqref{def:i24fy}, we define
\begin{align}
\jmath_1(x,y,\hatx_1,D_1)&:=\imath_1(x,y,\hatx_1)-\frac{t_1^*D_1}{1+s^*}\label{def:j14fy},\\
\jmath_2(x,y,\hatx_1,D_1,D_2)&:=\imath_2(x,y,\hatx_1)-\frac{t_1^*D_1}{1+s^*}-t_2^*D_2\label{def:j24fy}.
\end{align}
Using the definition of the rate-distortion-tilted information density in \eqref{def:j4fy}, we conclude that
\begin{align}
\jmath_{\rm{FY}}(x,y|R_1,D_1,D_2,P_X)=s^*\jmath_1(x,y,\hatx_1,D_1)+\jmath_2(x,y,\hatx_1,D_1,D_2)-s^*R_1\label{j=j1j24fy}.
\end{align}
In the proof of Lemma \ref{oneshotconverse4fy}, we make use of \eqref{j=j1j24fy} and the fact that $\Pr(A+B\geq c+d)\leq \Pr(A\geq c)+\Pr(B\geq d)$ for any variables $(A,B)$ and constants $(c,d)$.

Recall the setting of the Fu-Yeung problem in Figure \ref{systemmodelmd}. Note that when $Y$ is a constant, i.e. $|\calY|=1$, we recover the setting of the successive refinement problem~\cite{rimoldi1994}. Define a $(n,M_1,M_2)$-code for the successive refinement problem in a similar manner as Definition \ref{def:code4fy} (cf. \cite[Definition 1]{zhou2016second}). Let $\varepsilon_n^{\rm{SR}}(D_1,D_2)$ be the joint excess-distortion probability. Further, let $\rvR_{\rm{SR}}(R_1,D_1,D_2|P_X)$ be the minimum sum-rate function and let $\jmath_{\rm{SR}}(x|R_1,D_1,D_2,P_X)$ be the rate-distortion-tilted information density for the successive refinement problem~(cf. \cite{zhou2016second} and \cite{zhousuc2016}). For the case when $|\calY|=1$, we have that
\begin{align}
\rvR_{\rm{SR}}(R_1,D_1,D_2|P_X)&=\rvR_{\rm{FY}}(R_1,D_1,D_2|P_X),\\
\jmath_{\rm{SR}}(x|R_1,D_1,D_2,P_X)&=\jmath_{\rm{FY}}(x,g(x)|R_1,D_1,D_2,P_X)\label{def:j4sr}.
\end{align}
We remark that although the definition of the rate-distortion-tilted information density for the successive refinement problem in \eqref{def:j4sr} appears different from the definition in \cite[Definition 6]{zhou2016second}, the two quantities in \eqref{def:j4sr} and \cite[Definition 6]{zhou2016second} share same properties (cf. \cite[Lemma 3]{zhou2016second}) and are thus essentially the same. Invoking Lemma \ref{oneshotconverse4fy} with $\calY=\{1\}$, we obtain the following non-asymptotic converse bound for the successive refinement problem.
\begin{lemma}
\label{lemma:oneshot4sr}
Any $(1,M_1,M_2)$-code for the successive refinement problem satisfies that for any $\gamma\geq 0$,
\begin{align}
\varepsilon_1^{\rm{SR}}(D_1,D_2)
\nn&\geq \Pr\Big\{\jmath(X,D_1|P_X)\geq \log M_1+\gamma~\mathrm{or}\\
&\qquad\quad\jmath_{\rm{SR}}(X|R_1,D_1,D_2,P_X)\geq \log M_1M_2+s^*\log M_1+(1+s^*)\gamma\Big\}-4\exp(-\gamma)\label{converseoneshot4sr}.
\end{align}
\end{lemma}
We remark that the non-asymptotic converse bound in \eqref{converseoneshot4sr} can be used to establish converse results for second-order and moderate deviations asymptotics for memoryless sources (both discrete and continuous). The non-asymptotic converse bound in \eqref{converseoneshot4sr} is a strict generalization of \cite[Lemma 9]{zhou2016second} which is only valid for the successively refinable source-distortion triplets. Invoking Lemma \ref{lemma:oneshot4sr}, we have the potential to establish tight second-order and moderate deviations asymptotics for non-successively refinable continuous memoryless sources, such as the symmetric GMS under the quadratic distortion measures~\cite{chowberger}. For DMSes under arbitrary bounded distortion measures and GMSes under quadratic distortion measures as shown in \cite{zhou2016second}, we can make use of Lemma \ref{lemma:oneshot4sr} to simplify the converse proof significantly.

\subsection{Second-Order Asymptotics for DMSes}
In this subsection, we consider DMSes under bounded distortion measures, i.e., $\calX$, $\hat{\calX_1}$, $\hat{\calX_2}$ are all finite sets and $\max_{x,\hatx_i}d_i(x,\hatx_i),~i=1,2$ is finite. Let $\epsilon\in(0,1)$ be fixed. In the following, we give a formal definition of second-order coding region for the Fu-Yeung problem.
\begin{definition}
Given $(R_1,R_2,D_1,D_2,\epsilon)$, a pair $(L_1,L_2)$ is said to be second-order achievable for the Fu-Yeung problem if there exists a sequence of $(n,M_1,M_2)$-codes such that
\begin{align}
\limsup_{n\to\infty}\frac{\log M_i-nR_i}{\sqrt{n}}\leq L_i,~i=1,2,
\end{align}
and
\begin{align}
\limsup_{n\to\infty} \varepsilon_n^{\rm{FY}}(D_1,D_2)\leq \epsilon.
\end{align}
Given $(R_1,R_2,D_1,D_2,\epsilon)$, the set of all second-order achievable pairs is called the second-order $(R_1,R_2,D_1,D_2,\epsilon)$-coding region and denoted as $\calL(R_1,R_2,D_1,D_2,\epsilon)$.
\end{definition}

Before presenting the characterization of $\calL(R_1,R_2,D_1,D_2,\epsilon)$, we need several definitions. Recall that $P_{XY}$ and $P_Y$ are induced by $P_X$ and the deterministic function $g:\calX\to\calY$. Let the source dispersion~\cite{strassen1962asymptotische} be
\begin{align}
\rmV(P_Y)
&=\sum_y P_Y(y)\big(-\log P_Y(y)-H(P_Y)\big)^2\\
&=\sum_x P_X(x)\big(-\log P_Y(g(x))-H(P_Y)\big)^2.
\end{align}

Recall that $\rmV(P_X,D_1)$ is the distortion-dispersion function~\cite{kostina2012}, which is the variance of the distortion-tilted information density $\jmath(x,D_1|P_X)$~\cite{kostina2012converse} (cf. \eqref{djkostina}). Let the rate-distortion-dispersion function be
\begin{align}
\rmV(R_1,D_1,D_2|P_X):=\mathrm{Var}\Big[\jmath_{\rm{FY}}(X,g(X)|R_1,D_1,D_2,P_X)-\log P_Y(Y)\Big].
\end{align}
Define two covariance matrices:
\begin{align}
\mathbf{V}_1(R_1,D_1,D_2|P_X)
&:=\mathrm{Cov}\Big([\jmath(X,D_1|P_X),\jmath_{\rm{FY}}(X,g(X)|R_1,D_1,D_2,P_X)-\log P_Y(g(X))]^T\Big),\\
\mathbf{V}_2(R_1,D_1,D_2|P_X)
&:=\mathrm{Cov}\Big([\jmath_{\rm{FY}}(X,g(X)|D_1,D_2,P_X)-\log P_Y(g(X)),-\log P_Y(g(X))]^T\Big).
\end{align}

We impose the following conditions:
\begin{enumerate}
\item \label{cond1}$(Q_X,D_1')\to R(Q_X,D_1')$ is twice differentiable in the neighborhood of $(P_X,D_1)$ and the derivatives are bounded;
\item $(Q_X,R_1',D_1',D_2')\to \rvR_{\rm{FY}}(R_1',D_1',D_2'|Q_X)$ is twice differentiable in the neighborhood of $(P_X,R_1,D_1,D_2)$ and the derivatives are bounded;
\item $s^*$ in \eqref{def:s*}, $t_1^*$ in \eqref{def:t1*} and $t_2^*$ in \eqref{def:t2*} are well defined;
\item $\infty>\rvR_{\rm{FY}}(R(P_X,D_1),D_1,D_2|P_X)>R(P_X,D_1)>0$;
\item  \label{condf} The covariance matrices $\mathbf{V}_1(R_1,D_1,D_2|P_X)$ and $\mathbf{V}_2(R_1,D_1,D_2|P_X)$ are positive semi-definite.
\end{enumerate}

\begin{theorem}
\label{mddsecondregion}
Under conditions \eqref{cond1} to \eqref{condf}, depending on $(R_1,R_2)$, the second-order $(R_1,R_2,D_1,D_2,\epsilon)$-coding region for the Fu-Yeung problem satisfies
\begin{itemize}
\item Case (i): $R_1=R(P_X,D_1)$ and $R_1+R_2>\rvR_{\rm{FY}}(R_1,D_1,D_2|P_X)+H(P_Y)$
\begin{align}
\calL(R_1,R_2,D_1,D_2,\epsilon)=\big\{(L_1,L_2):L_1\geq \sqrt{\rmV(P_X,D_1)}\rmQ^{-1}(\epsilon)\big\}.
\end{align}
\item Case (ii): $R_1=R(P_X,D_1)$ and  $R_1+R_2=\rvR_{\rm{FY}}(R_1,D_1,D_2|P_X)+H(P_Y)$
\begin{align}
\calL(R_1,R_2,D_1,D_2,\epsilon)=\big\{(L_1,L_2):\Psi(L_1,(1+s^*)L_1+L_2;\mathbf{0};\mathbf{V}_1(R_1,D_1,D_2|P_X))\geq 1-\epsilon\big\}.
\end{align}
\item Case (iii): $R(P_X,D_1)<R_1<\rvR_{\rm{FY}}(R_1,D_1,D_2|P_X)$ and $R_1+R_2=\rvR_{\rm{FY}}(R_1,D_1,D_2|P_X)+H(P_Y)$
\begin{align}
\calL(R_1,R_2,D_1,D_2,\epsilon)=\big\{(L_1,L_2):(1+s^*)L_1+L_2\geq \sqrt{\rmV(R_1,D_1,D_2|P_X)}\rmQ^{-1}(\epsilon)\big\}.
\end{align}
\item Case (iv) $R_1=\rvR_{\rm{FY}}(R_1,D_1,D_2|P_X)$ and $R_2=H(P_Y)$
\begin{align}
\calL(R_1,R_2,D_1,D_2,\epsilon)=\big\{(L_1,L_2):\Psi((1+s^*)L_1+L_2,L_2;\mathbf{0},\mathbf{V}_2(R_1,D_1,D_2|P_X))\geq 1-\epsilon\big\}.
\end{align}
\item Case (v) $R_1>\rvR_{\rm{FY}}(R_1,D_1,D_2|P_X)$ and $R_2=H(P_Y)$
\begin{align}
\calL(R_1,R_2,D_1,D_2,\epsilon)=\big\{(L_1,L_2):L_2\geq \sqrt{\rmV(P_Y)}\rmQ^{-1}(\epsilon)\big\}.
\end{align}
\end{itemize}
\end{theorem}
The proof of Theorem \ref{mddsecondregion} is provided in Section \ref{sec:proof4fy}.

We remark that the second-order coding region for discrete successive refinement problem~\cite[Theorem 5]{zhou2016second} is recovered by cases (i) to (iii) in Theorem \ref{mddsecondregion} by letting $Y$ be a constant. Further, let $R_{\rm{FY}}(D_1,D_2|P_X)$ be defined as
\begin{align}
R_{\rm{FY}}(D_1,D_2|P_X):=\min_{P_{\hatX_1|X},P_{\hatX_2|X\hatX_1}\in\calP(P_X,D_1,D_2)} I(\hatX_1;Y)+I(X;\hatX_1\hatX_2|Y)\label{def:rfyd1d2},
\end{align}
where $\calP(P_X,D_1,D_2)$ was defined above Theorem \ref{mdfirst}. It can be verified easily that $R_{\rm{FY}}(D_1,D_2|P_X)$ is the solution of $R_1$ to $R_1=\rvR_{\rm{FY}}(R_1,D_1,D_2|P_X)$ and thus $s^*=0$. Hence, the expressions in cases (iv) can be simplified so that $R_1$ is not included in the results.

\subsection{Large and Moderate Deviations Asymptotics for DMSes}
The optimal error exponent function is defined as follows.
\begin{definition} \label{defee4fy}
{\em  A  non-negative number $E$ is said to be an {\em $(R_1,R_2,D_1,D_2)$-achievable error exponent} if there exists a sequence of $(n,M_1,M_2)$-codes such that \eqref{def:ratecons4fy} holds and
\begin{align}
\liminf_{n\to\infty}-\frac{\log \varepsilon_n^{\rm{FY}}(D_1,D_2)}{n}\geq E.
\end{align}
The supremum of all $(R_1,R_2,D_1,D_2)$-achievable error exponents is called the optimal error exponent and denoted as $E^*(R_1,R_2,D_1,D_2)$. }
\end{definition}
Given a rate-distortion tuple $(R_1,D_1,D_2)$, define
\begin{align}
\label{defeef4fy}
F(P_X,R_1,R_2,D_1,D_2)&:=\inf_{Q_X:\rvR_{FY}(R_1,D_1,D_2|Q_X)+H(Q_Y)\geq R_1+R_2} D(Q_X\|P_X),\\
&=\min\Big\{\min_{Q_X:R(Q_X,D_1)> R_1} D(Q_X\|P_X),\min_{\substack{ Q_X: R(Q_X,D_1)\leq R_1\\ \rvR_{FY}(R_1,D_1,D_2|Q_X)+H(Q_Y)\geq R_1+R_2}}D(Q_X\|P_X)\Big\}\label{equivalentee},
\end{align}
where $Q_Y$ is induced by $Q_X$ and the deterministic function $g:\calX\to\calY$ and \eqref{equivalentee} holds since if $R_1<R(Q_X,D_1)$, then $\rvR_{\rm{FY}}(R_1,D_1,D_2|Q_X)$ is infeasible and thus equals $\infty$.
\begin{theorem}
\label{eegray}
The optimal error exponent for the Fu-Yeung problem is
\begin{align}
E^*(R_1,R_2,D_1,D_2)&=F(P_X,R_1,R_2,D_1,D_2).
\end{align}
\end{theorem}
The proof of Theorem \ref{eegray} is omitted due to the similarity to \cite[Theorem 13]{zhou2015second}. Theorem \ref{eegray} can be proven by invoking the type-covering lemma (cf. Lemma \ref{mddtypecovering}) in the achievability part and by using the strong converse result and the change of measure technique in the converse part. We remark that the result in \cite{kanlis1996error} is a special case of Theorem \ref{eegray} by letting $|\calY|=1$.

The moderate deviations constant for the Fu-Yeung problem is defined as follows.
\begin{definition}
{\em  
Consider any sequence $\{\rho_n\}_{n\in\bbN}$ satisfying \eqref{eqn:cond_rho}. Let $\theta_i$ for $i=1,2$ be positive numbers. A number $\nu\ge 0$ is said to be a {\em $(R_1,R_2, D_1, D_2)$-achievable moderate deviations constant} if there exists a sequence of $(n,M_1,M_2)$-codes such that
\begin{align}
\limsup_{n\to\infty}\frac{1}{n\rho_n}(\log M_i-nR_i)&\leq \theta_i,~i=1,2,\\
\liminf_{n\to\infty}-\frac{\log \varepsilon_n^{\rm{FY}}(D_1,D_2)}{n\rho_n^2}&\geq \nu.
\end{align}
The supremum of all $(R_1,R_2, D_1, D_2)$-achievable moderate deviations constants is denoted as $\nu^*(R_1,R_2|D_1, D_2)$.
}
\end{definition}

\begin{theorem}
\label{mdcfy}
Given the conditions in Theorem \ref{mddsecondregion} and the assumption that $\rmV(P_X,D_1)$ and $\rmV(R_1,D_1,D_2|P_X)$ are positive, depending on $(R_1,R_2)$, the moderate-deviations constant for the Fu-Yeung problem satisfies that
\begin{itemize}
\item Case (i): $R_1=R(P_X,D_1)$ and $R_1+R_2>\rvR(R_1,D_1,D_2|P_X)+H(P_Y)$
\begin{align}
\nu^*(R_1,R_2|D_1, D_2)=\frac{\theta_1^2}{2\rmV(P_X,D_1)}.
\end{align}
\item Case (ii): $R_1=R(P_X,D_1)$ and  $R_1+R_2=\rvR(R_1,D_1,D_2|P_X)+H(P_Y)$
\begin{align}
\nu^*(R_1,R_2|D_1, D_2)&=\min\Big\{\frac{\theta_1^2}{2\rmV(P_X,D_1)},\frac{\big((1+s^*)\theta_1+\theta_2\big)^2}{2\rmV(R_1,D_1,D_2|P_X)}\Big\}
\end{align}
\item Case (iii): $R(P_X,D_1)<R_1<\rvR(R_1,D_1,D_2|P_X)$ and $R_1+R_2=\rvR(R_1,D_1,D_2|P_X)+H(P_Y)$
\begin{align}
\nu^*(R_1,R_2|D_1, D_2)=\frac{\big((1+s^*)\theta_1+\theta_2\big)^2}{2\rmV(R_1,D_1,D_2|P_X)}.
\end{align}
\item Case (iv): $R_1=\rvR(R_1,D_1,D_2|P_X)$ and $R_2=H(P_Y)$
\begin{align}
\nu^*(R_1,R_2|D_1, D_2)&=\min\Big\{\frac{\big((1+s^*)\theta_1+\theta_2\big)^2}{2\rmV(R_1,D_1,D_2|P_X)},\frac{\theta_2^2}{2\rmV(P_Y)}\Big\}\label{case4}.
\end{align}
\item Case (v): $R_1>\rvR(R_1,D_1,D_2|P_X)$ and $R_2=H(P_Y)$
\begin{align}
\nu^*(R_1,R_2|D_1, D_2)=\frac{\theta_2^2}{2\rmV(P_Y)}.
\end{align}
\end{itemize}
\end{theorem}
The proof of Theorem \ref{mdcfy} is similar to Theorem \ref{theoremmdc}, \cite[Theorem 6]{zhou2016second} and thus omitted. We remark the moderate deviations asymptotics for the successive refinement problem for DMSes~\cite[Theorem 12]{zhou2016second} is a special case of Theorem \ref{mdcfy} by letting $Y$ be a constant.

\subsection{Numerical Examples}
\label{sec:caseiinum}

We consider the numerical example inspired by \cite{perron2006kaspi} and calculate the moderate deviations constant for case (iv) in Theorem \ref{mdcfy} . Let $\calS_1=\{0,1\}$ and $\calS_2=\{0,1,\rme\}$. Let $S_1$ take values in $\calS_1$ with equal probability and let $P_{S_2|S_1}(s_2|s_1)=(1-p)1\{s_1=s_2\}+p1\{s_2=\rme\}$ where $1\{\}$ is the indicator function. 
Let the source be $X=(S_1,S_2)$ and the deterministic function be $Y=g(X)=g(S_1,S_2)=S_2$. Let $\hatX_1=\hatX_2=\{0,1\}$ and the distortion measures be $d_1(x,\hatx_1)=1\{s_1=\hatx_1\}$ and $d_2(x,\hatx_2)=1\{s_2=\hatx_2\}$. Choose $(p,D_1,D_2)$ such that $D_1\leq \frac{1}{2}$ and $D_1-\frac{1-p}{2}\leq D_2\leq p D_1$. For this case, using the definitions of $s^*$ in \eqref{def:s*}, $t_1^*$ in \eqref{def:t1*} and $t_2^*$ in \eqref{def:t2*}, we have
\begin{align}
s^*&=0,\\
t_1^*&=\log\big((1-p)/(D_1-D_2)-1\big),\\
t_2^*&=-t_1^*+\log\big(p/D_2-1\big).
\end{align} 
Recall that $H_\rmb(\cdot)$ is the binary entropy function. Let
\begin{align}
\alpha_0&:=\log\big(2/(1+\exp(-t_1^*))\big)-t_1^*D_1-t_2^*D_2,\\
\alpha&:=\log\big(2/(1+\exp(-t_1^*-t_2^*)\big)-t_1^*D_1-t_2^*D_2,\\
g_1(p,D_1,D_2)&:=\log 2-(1-p)H_\rmb((D_1-D_2)/(1-p))-pH_\rmb(D_2/p),\\
g_2(p,D_1,D_2)&:=p(1-p)\Big\{\log(1-D_2/p)-\log(1-(D_1-D_2)/(1-p))\Big\}^2,\\
g_3(p,D_1,D_2)&:=(1-p)\alpha_0\log\frac{2}{1-p}+p\alpha\log\frac{1}{p}.
\end{align}
Then, it can be verified that
\begin{align}
H(P_Y)&=(1-p)\log 2+H_\rmb(p),\\
\rmV(P_Y)&=p(1-p)\Big(\log\frac{2p}{1-p}\Big)^2,\\
\rmV(R_1,D_1,D_2|P_X)&=g_2(p,D_1,D_2)+\rmV(P_Y)+2\Big(g_3(p,D_1,D_2)-H(P_Y)g_1(p,D_1,D_2)\Big).
\end{align}
Thus, the moderate deviations constant for case (iv) can be calculated by invoking \eqref{case4}.

\section{Proof of Second-Order Asymptotics for the Kaspi Problem (Theorem \ref{kaspisecond})}
\label{sec:proof4kaspi}

\subsection{Achievability Coding Theorem}

In this section, we first prove a type covering lemma for the Kaspi problem. Then, invoking the type covering lemma, we can prove an upper bound on the excess-distortion probability. Finally, invoking the Berry-Esseen theorem together with proper Taylor expansions, we manage to prove an achievable second-order coding rate.

Define a constant
\begin{align}
c=\Big(8|\calX|\cdot|\calY|\cdot|\hat{\calX}_1|\cdot|\hat{\calX}_2|+6\Big).
\end{align}

We are now ready to present the type covering lemma for Kaspi problem.
\begin{lemma}
\label{coveringkaspi}
There exists a set $\calB\in\hat{\calX}_1^n$ such that for each $(x^n,y^n)\in\calT_{Q_{XY}}$, if we let
\begin{align}
z^*=\argmin_{\hat{x}_1^n\in\calB}d_1(x^n,\hat{x}_1^n),
\end{align}
then we have
\begin{itemize}
\item The distortion between $x^n$ and $z^*$ is upper bounded by $D_1$, i.e., 
\begin{align}
d_1(x^n,z^*)\leq D_1,
\end{align}
\item There exists a set $\calB(z^*,y^n)\in\hat{\calX}_2^n$ such that
\begin{align}
\min_{\hat{x}_2^n\in\calB(z^*,y^n)}d_2(x^n,\hat{x}_2^n)\leq D_2.
\end{align}
\item The size of the set $\calB\cup\calB(z^*,y^n)$ is upper bounded as follows:
\begin{align}
\log \big|\calB\cup\calB(z^*,y^n)\big|\leq nR(Q_{XY},D_1,D_2)+c
\log (n+1).
\end{align}
\end{itemize}
\end{lemma}
The proof of Lemma \ref{coveringkaspi} is similar to the proofs in \cite[Lemma 8]{no2016} and \cite[Lemma 10]{zhou2015second} and thus omitted.

Given an $(n,M)$-code, define
\begin{align}
R_n:=\frac{1}{n}\log M-(c+|\calX|\cdot|\calY|)\frac{\log (n+1)}{n}\label{defrn}.
\end{align}
Invoking Lemma \ref{coveringkaspi}, we can upper bound the excess-distortion probability as follows.
\begin{lemma}
\label{uppexcess}
There exists an $(n,M)$-code whose excess-distortion probability can be upper bounded as:
\begin{align}
\varepsilon_n^{\rmK}(D_1,D_2)&\leq \Pr\Big(R_n<R_{\rmK}(\hat{T}_{X^nY^n},D_1,D_2)\Big).
\end{align}
\end{lemma}
\begin{proof}
The proof of Lemma \ref{uppexcess} is similar to the proofs of \cite[Lemma 5]{watanabe2015second} and \cite[Lemma 11]{zhou2016second}. Consider the following coding scheme. Given source sequence pair $(x^n,y^n)$, the encoder first calculates the joint type $\hat{T}_{x^ny^n}$, which can be transmitted reliably using at most $|\calX|\cdot|\calY|\log (n+1)$ nats. Then the encoder calculates $R_\rmK(\hat{T}_{x^ny^n},D_1,D_2)$ and declares an error if $nR_{\rmK}(\hat{T}_{x^ny^n},D_1,D_2)+c\log (n+1)+|\calX|\cdot|\calY|\log (n+1)>\log M$. Otherwise, the encoder chooses a set $\calB$ satisfying the properties specified in Lemma \ref{coveringkaspi} and sends the index of $z^*=\argmin_{\hat{x}_1^n\in\calB} d_1(x^n,\hat{x}_1^n)$. Subsequently, the decoder chooses a set $\calB(z^*,y^n)$ satisfying the properties specified in Lemma \ref{coveringkaspi} and sends the index of $\argmin_{\hat{x}_2^n\in\calB(z^*,y^n)}d_2(x^n,\hat{x}_2^n)$. Lemma \ref{coveringkaspi} implies that the decoding is error free if $nR_{\rmK}(\hat{T}_{x^ny^n},D_1,D_2)+c\log (n+1)+|\calX|\cdot|\calY|\log (n+1)\leq \log M$. The proof of Lemma \ref{uppexcess} is now complete.
\end{proof}

Given a joint probability mass function (pmf) $P_{XY}$, let $\calS=\mathrm{supp}(P_{XY})$ and $|\calS|=m$. Let us sort $P_{XY}(x,y)$ in an decreasing order for all $(x,y)\in\calX\times\calY$, and for all $i\in[1:m]$, let $(x_i,y_i)$ be the pair such that $P_{XY}(x_i,y_i)$ is the $i$-th largest. Let $\Theta(P_{XY})$ be a joint distribution defined on $\calS$ such that $\Theta_i(P_{XY})=P_{XY}(x_i,y_i)$ for all $i\in[1:m]$. Define the typical set
\begin{align}
\calA_n(P_{XY})&:=\bigg\{Q_{XY}\in\calP_n(\calX\times\calY):\|\Theta_i(Q_{XY})-\Theta_i(P_{XY})\|_{\infty}\leq \sqrt{\frac{\log n}{n}},~\forall~i\in[1:m-1]\bigg\}\label{def:calan}.
\end{align}
According to \cite[Proposition 2]{watanabe2015second} (See also \cite[Lemma 22]{tan2014state}), we obtain
\begin{align}
\Pr\Big(\hat{T}_{X^nY^n}\notin \calA_n(P_{XY})\Big)&\leq \frac{2(m-1)}{n^2}\label{atypicalprob}.
\end{align}

Now choose 
\begin{align}
\frac{1}{n}\log M=R_{\rmK}(P_{XY},D_1,D_2)+\frac{L}{\sqrt{n}}+\Big(c+|\calX|\cdot|\calY|\Big)\frac{\log (n+1)}{n}.
\end{align}

Invoking \eqref{defrn}, we obtain
\begin{align}
R_n&=R_{\rmK}(P_{XY},D_1,D_2)+\frac{L}{\sqrt{n}}.
\end{align}

The following lemma is essential in the proof.
\begin{lemma}
\label{lemmataylor}
For $(x^n,y^n)$ such that $\hat{T}_{x^ny^n}\in\calA_n(P_{XY})$, we have
\begin{align}
R_{\rmK}(\hat{T}_{x^ny^n},D_1,D_2)
&=\frac{1}{n}\sum_{i=1}^n \jmath_{\rmK}(x_i,y_i|D_1,D_2,P_{XY})+O\Big(\frac{\log n}{n}\Big)\label{taylorexpand}.
\end{align}
\end{lemma}
The proof of Lemma \ref{lemmataylor} is given in Appendix \ref{prooflemmataylor}.

Define $\xi_n=\frac{\log n}{n}$. Invoking Lemma \ref{uppexcess}, we obtain
\begin{align}
\varepsilon_n^{\rmK}(D_1,D_2)
&\leq \Pr\Big(R_n<R_{\rmK}(\hat{T}_{X^nY^n},D_1,D_2),\hat{T}_{X^nY^n}\in\calA_n(P_{XY})\Big)+\Pr\Big(\hat{T}_{X^nY^n}\notin \calA_n(P_{XY})\Big)\\
&\leq \Pr\Big(R_{\rmK}(P_{XY},D_1,D_2)+\frac{L}{\sqrt{n}}<\frac{1}{n}\sum_{i=1}^n \jmath_{\rmK}(X_i,Y_i|D_1,D_2,P_{XY})+O(\xi_n)\Big)+\frac{2(m-1)}{n^2}\\
&\leq \Pr\bigg(\frac{1}{\sqrt{n}}\sum_{i=1}^n \big(\jmath_{\rmK}(X_i,Y_i|D_1,D_2,P_{XY})-R_{\rmK}(P_{XY},D_1,D_2)\big)>L+O(\xi_n\sqrt{n})\bigg)+\frac{2(m-1)}{n^2}\\
&\leq \rmQ\Bigg(\frac{L+O(\xi_n\sqrt{n})}{\sqrt{\rmV(D_1,D_2|P_{XY})}}\Bigg)+\frac{6\rmT(D_1,D_2|P_{XY})}{\sqrt{n}\rmV^{3/2}(D_1,D_2|P_{XY})}+\frac{2(m-1)}{n^2}\label{achfinal},
\end{align}
where \eqref{achfinal} follows from Berry-Esseen theorem and $\rmT(D_1,D_2|P_{XY})$ is the third absolute moment of $\jmath_{\rmK}(X,Y|D_1,D_2,P_{XY})$, which is finite for DMSes.

Therefore, if $L$ satisfies 
\begin{align}
L\geq \sqrt{\rmV(D_1,D_2|P_{XY})}\rmQ^{-1}(\epsilon),
\end{align}
by noting that $O(\xi_n\sqrt{n})=O(\log n/\sqrt{n})$, we obtain 
\begin{align}
\limsup_{n\to\infty} \varepsilon_n^{\rmK}(D_1,D_2)\leq \epsilon.
\end{align}
Thus, the optimal second-order coding rate satisfies that $L^*(\epsilon,D_1,D_2)\leq\sqrt{\rmV(D_1,D_2|P_{XY})}Q^{-1}(\epsilon)$.

\subsection{Converse Coding Theorem}
The converse part follows by applying the Berry-Esseen theorem to the non-asymptotic converse bound in \eqref{nshot}.

Let 
\begin{align}
\log M:=nR(P_{XY},D_1+D_2)+L\sqrt{n}+\frac{1}{2}\log n.
\end{align}

Invoking \eqref{nshot} with $\epsilon=\frac{\log n}{2n}$, we obtain that
\begin{align}
\varepsilon_n^{\rmK}(D_1,D_2)
&\geq \Pr\Big(\sum_{i=1}^n\jmath_{\rmK}(x,y|D_1,D_2,P_{XY})\geq nR_{\rmK}(P_{XY},D_1,D_2)+L\sqrt{n}\Big)-\exp\Big(-\frac{\log n}{2}\Big)\\
&\geq \rmQ\bigg(\frac{L}{\sqrt{\rmV(D_1,D_2|P_{XY})}}\bigg)-\frac{6\rmT(D_1,D_2|P_{XY})}{\sqrt{n}\rmV^{3/2}(D_1,D_2|P_{XY})}-\frac{1}{\sqrt{n}}\label{conversefinal},
\end{align}
where \eqref{conversefinal} follows the from Berry-Esseen theorem. Hence, if $L<\sqrt{\rmV(D_1,D_2|P_{XY})}Q^{-1}(\epsilon)$, then 
\begin{align}
\limsup_{n\to\infty} \varepsilon_n^{\rmK}(D_1,D_2)>\epsilon.
\end{align}
Hence, the optimal second-order coding rate satisfies that $L^*(\epsilon,D_1,D_2)\geq\sqrt{\rmV(D_1,D_2|P_{XY})}Q^{-1}(\epsilon)$.

\section{Proof of Second-order Asymptotics for the Fu-Yeung Problem (Theorem \ref{mddsecondregion})}
\label{sec:proof4fy}
\subsection{Achievability Coding Scheme}
In this subsection, we first present a type covering lemma tailored to the Fu-Yeung problem, based on which, we can derive an upper bound on the joint excess-distortion and error probability. Finally, invoking Taylor expansions and the Berry-Esseen Theorem, we derive an achievable second-order coding region for DMSes.

Define 
\begin{align}
c_1&=|\calX|\cdot |\calY|\cdot |\hat{\calX_1}|+2,\\
c_2&=7|\calX|\cdot |\calY| \cdot |\hat{\calX}_1|\cdot|\hat{\calX_2}|+4.
\end{align}

We are now ready to present the type covering lemma.
\begin{lemma}
\label{mddtypecovering}
Let type $Q_X\in\calP_n(\calX)$. Let $Q_Y$ be induced by $Q_X$ and the deterministic function $g:\calX\to\calY$. Let $R_1\geq R(Q_X,D_1)$ be given.
\begin{itemize}
\item There exists a set $\calB\in\calX_1^n$ such that
for each $x^n\in\calT_{Q_X}$, if we let $z^*=\argmin_{z\in\calB} d_1(x^n,z)$, then $d_1(x^n,z^*)\leq D_1$;
\item There exists a set $\calB(z^*)\in\calX_2^n$ such that
\begin{align}
\min_{\hatx_2^n\in\calB(z^*)}d_2(x^n,\hatx_2^n)\leq D_2.
\end{align}
\item There exists a set $\calB_Y\in\hat{\calY}^n$ satisfying that $\frac{1}{n}\log |\calB_Y|\leq H(Q_Y)$ and there exists $\haty^n\in\calB_Y$ such that $\haty^n=g(x^n)$.
\end{itemize}
The sizes of $\calB$ and $\calB(z^*)$ are bounded as follows:
\begin{align}
\frac{1}{n}\log |\calB|&\leq R_1+c_1\log (n+1)\\
\frac{1}{n}\log (|\calB|\cdot|\calB(z^*)|)&\leq \rvR_{\rm{FY}}(R_1,D_1,D_2|Q_X)+(c_1+c_2)\log(n+1).
\end{align}
\end{lemma}
The proof of Lemma \ref{mddtypecovering} is similar to \cite[Lemma 10]{zhou2016second} and thus omitted.

Let 
\begin{align}
R_{1,n}:&=\frac{1}{n}\big(\log M_1 -(c_1+|\calX|)\log (n+1)\big)\label{def:r1n},\\
R_{2,n}:&=\frac{1}{n}\big(\log M_2-(c_2+|\calY|)\log (n+1)\big)\label{def:r2n},
\end{align}

Invoking Lemma \ref{mddtypecovering}, we can upper bound the joint excess-distortion and error probability for an $(n,M_1,M_2)$-code. 

\begin{lemma}
\label{mdduppjoint}
There exists an $(n,M_1,M_2)$-code such that 
\begin{align}
\varepsilon_n^{\rm{FY}}(D_1,D_2)
\nn&\leq \Pr\bigg(R_{1,n}<R(\hatT_{X^n},D_1)~\mathrm{or}~R_{2,n}+\frac{c_2\log (n+1)}{n}<H(\hatT_{g(X^n)})\\
&\qquad \qquad ~\mathrm{or}~R_{1,n}+R_{2,n}<\rvR_{\rm{FY}}(R_{1,n},D_1,D_2|\hatT_{X^n})+H(\hatT_{g(X^n)})\bigg).
\end{align}
\end{lemma}
\begin{proof}
Set $(R_1,R_2)=(R_{1,n},R_{2,n})$. Consider the following coding scheme. Given a source $x^n$, the encoder $f_2$ calculates its type $\hatT_{x^n}$. Then, the encoder $f_2$ obtain $y^n$ using the deterministic function $y_i=g(x_i)$ and its type $\hatT_{y^n}$. Now encoder $f_2$ calculates $R(\hatT_{x^n},D_1)$ and $\rvR(R_{1,n},D_1,D_2|\hatT_{x^n})$. If $\log M_1<nR(\hatT_{x^n},D_1)+(c_1+|\calX|)\log (n+1)$ or  $\log M_2<nH(\hatT_{y^n})+|\calY|\log (n+1)$ or $\log M_1M_2<n\rvR_{\rm{FY}}(R_{1,n},D_1,D_2|\hatT_{x^n})+nH(\hatT_{y^n})+(c_1+c_2+|\calX|+|\calY|)\log (n+1)$, then the system declares an error. 

Otherwise, the encoder $f_1$ sends the type of $x^n$ with at most $|\calX|\log (n+1)$ nats and the encoder $f_2$ sends the type of $y^n$ using at most $|\calY|\log (n+1)$ nats. Further, the encoder $f_2$ sends the index of $y^n=g(x^n)$ in the type class $\calT_{\hatT_{y^n}}$.
Now, choose $\calB\in\calX_1^n$ in Lemma \ref{mddtypecovering} and let $z^*=\argmin_{z\in\calB}d_1(x^n,z)$. Given $z^*$, choose $\calB(z^*)$ in Lemma \ref{mddtypecovering} and let $z_2^*=\argmin_{z_2\in\calB(z^*)}d_2(x^n,z_2)$. Finally, we use the encoder $f_1$ to send the index of $z_1^*$ and use either $f_1$ or $f_2$ to send out the index of $z_2^*$. Invoking Lemma \ref{mddtypecovering}, we conclude that no error will be made if $\log M_1\geq nR(\hatT_{x^n},D_1)+(c_1+|\calX|)\log (n+1)$, $\log M_2\geq nH(\hatT_{y^n})+|\calY|\log (n+1)$ and $\log M_1M_2\geq n\rvR_{\rm{FY}}(R_{1,n},D_1,D_2|\hatT_{x^n})+nH(\hatT_{y^n})+(c_1+c_2+|\calX|+|\calY|)\log (n+1)$. The proof is now complete.
\end{proof}

Let 
\begin{align}
\log M_1&=nR_1+L_1\sqrt{n}+(c_1+|\calX|)\log (n+1)\\
\log M_2&=nR_2+L_2\sqrt{n}+(c_2+|\calY|)\log (n+1).
\end{align}
Hence, invoking \eqref{def:r1n} and \eqref{def:r2n}, we obtain
\begin{align}
R_{i,n}&=R_i+\frac{L_i}{\sqrt{n}},~i=1,2.
\end{align}

Define the typical set
\begin{align}
\calA_n(P_X)&:=\bigg\{Q_X\in\calP_n(\calX):\|Q_X-P_X\|_{\infty}\leq \sqrt{\frac{\log n}{n}}\bigg\}.
\end{align}
According to \cite[Lemma 22]{tan2014state}, we obtain
\begin{align}
\Pr\Big(\hatT_{X^n}\notin \calA_n(P_X)\Big)\leq \frac{2|\calX|}{n^2}.
\end{align}
Let $P_Y$ be induced by the source distribution $P_X$ and the deterministic function $g:\calX\to\calY$. Then, for any $x^n$, we have that for any $y$
\begin{align}
\hatT_{y^n}(y)-P_Y(y)&=
\hatT_{g(x^n)}(y)-P_Y(y)\\
&=\sum_{x:g(x)=y} \Big(\hatT_{x^n}(x)-P_X(x)\Big).
\end{align}
Thus, if $\hatT_{X^n}\in\calA_n(P_X)$, we obtain
\begin{align}
\|\hatT_{Y^n}-P_Y\|_{\infty}&\leq |\calX|\sqrt{\frac{\log n}{n}}.
\end{align}

For $x^n$ such that $\hatT_{x^n}\in\calA_n(P_X)$, applying Taylor's expansions and noting that $y^n=g(x^n)$, we obtain
\begin{align}
H(\hatT_{g(x^n)})
&=H(\hatT_{y^n})\\
&=H(P_Y)+\sum_{y} \Big(\hatT_{y^n}(y)-P_Y(y)\Big)(-\log P_Y(y))+O\Big(\|\hatT_{y^n}-P_Y\|^2\Big)\\
&=\sum_{y} -\hatT_{y^n}(y)\log P_Y(y)+O\Bigg(\frac{\log n}{n}\Bigg)\\
&=\frac{1}{n}\sum_{i=1}^n -\log P_Y(y_i)+O\Bigg(\frac{\log n}{n}\Bigg),
\end{align}
and
\begin{align}
R(\hatT_{x^n},D_1)
&=R(P_X,D_1)+\sum_x (\hatT_{x^n}(x)-P_X(x))\jmath(x,D_1|P_X)+O\Big(\|\hatT_{x^n}-P_X\|^2\Big)\\
&=\frac{1}{n}\sum_{i=1}^n \jmath(x_i,D_1|P_X)+O\Big(\frac{\log n}{n}\Big),
\end{align}
and
\begin{align}
\nn &\rvR_{\rm{FY}}(R_{1,n},D_1,D_2|\hatT_{x^n})\\
&=\rvR_{\rm{FY}}(R_1,D_1,D_2|P_X)-s^*\frac{L_1}{\sqrt{n}}+\sum_x \Big(\hatT_{x^n}-P_X(x)\Big)\jmath_{\rm{FY}}(x,g(x)|R_1,D_1,D_2,P_X)+O\Big(\|\hatT_{x^n}-P_X\|^2\Big),\label{usederivative4fy}\\
&=\frac{1}{n}\sum_{i=1}^n \jmath_{\rm{FY}}(x_i,g(x_i)|R_1,D_1,D_2,P_X) -\frac{s^*L_1}{\sqrt{n}}+O\Big(\frac{\log n}{n}\Big),
\end{align}
where in \eqref{usederivative4fy}, we use the fact that given $Q_X$ in the neighborhood of $P_X$ satisfying some regularity condition (this condition is satisfied by $\hatT_{x^n}\in\calA_n(P_X)$), for arbitrary $a\in\calX$,
\begin{align}
\frac{\partial \rvR_{\rm{FY}}(R_1,D_1,D_2|Q_X)}{\partial Q_X(a)}\Bigg|_{Q_X=P_X}=\jmath_{\rm{FY}}(x,g(x)|R_1,D_1,D_2,P_X)-(1+s^*)\label{derivative4fy}.
\end{align}
The proof of \eqref{derivative4fy} can be done similarly as Lemma \ref{lemmataylor}, \cite[Theorem 2.2]{kostina2013lossy}, \cite{zhou2016second} and thus omitted.

Recall that $\xi_n=\frac{\log n}{n}$. Therefore, invoking Lemma \ref{mdduppjoint}, we obtain
\begin{align}
\nn&\varepsilon_n^{\rm{FY}}(D_1,D_2)\\
\nn&\leq \Pr\bigg(R_{1,n}<R(\hatT_{X^n},D_1)~\mathrm{or}~R_{2,n}+\frac{c_2\log (n+1)}{n}< H(\hatT_{g(X^n)})\\
&\qquad \qquad \mathrm{or}~R_{1,n}+R_{2,n}<H(\hatT_{g(X^n)})+\rvR_{\rm{FY}}(R_{1,n},D_1,D_2|\hatT_{X^n}),~\hatT_{X^n}\in\calA_n(P_{XY})\bigg)+\Pr\Big(\hatT_{X^n}\notin \calA_n(P_{XY})\Big)\\
\nn&\leq \Pr\bigg(R_1+\frac{L_1}{\sqrt{n}}<\frac{1}{n}\sum_{i=1}^n \jmath(X_i,D_1|P_X)+O(\xi_n)~\mathrm{or}~R_2<\frac{1}{n}\sum_{i=1}^n \log \frac{1}{P_Y(Y_i)}+O(\xi_n)\\
&\qquad \quad \mathrm{or}~R_1+R_2+\frac{(1+s^*)L_1+L_2}{n}<\frac{1}{n}\sum_{i=1}^n\Big(\jmath_{\rm{FY}}(X_i,g(X_i)|R_1,D_1,D_2,P_X)-\log P_Y(Y_i)\Big)+O(\xi_n)\bigg)+\frac{2|\calX|}{n^2}\label{upperbd1}.
\end{align}

Then, we need to evaluate the upper bound in \eqref{upperbd1} given each case of $(R_1,R_2)$.
\begin{itemize}
\item Case (i) $R_1=R(P_X,D_1)$ and $R_1+R_2>\rvR_{\rm{FY}}(R_1,D_1,D_2|P_X)+H(P_Y)$

In this case we have that $R_2=\rvR_{\rm{FY}}(R_1,D_1,D_2|P_X)-R_1>H(P_Y)$. Thus, invoking the weak law of large numbers, we obtain
\begin{align}
\kappa_{1,n}&:=\Pr\Big(R_2<\frac{1}{n}\sum_{i=1}^n \log \frac{1}{P_Y(Y_i)}+O(\xi_n)\Big)\to 0,\\
\kappa_{2,n}&:=\Pr\Big(R_1+R_2+\frac{(1+s^*)L_1+L_2}{n}<\frac{1}{n}\sum_{i=1}^n\Big(\jmath_{\rm{FY}}(X_i,g(X_i)|R_1,D_1,D_2,P_X)-\log P_Y(Y_i)\Big)+O(\xi_n)\Big)\to 0.
\end{align}

Invoking \eqref{upperbd1}, we obtain
\begin{align}
\varepsilon_n^{\rm{FY}}(D_1,D_2)
&\leq \Pr\Big(R_1+\frac{L_1}{\sqrt{n}}<\frac{1}{n}\sum_{i=1}^n \jmath(X_i,D_1|P_X)+O(\xi_n)\Big)+\frac{2|\calX|}{n^2}+\kappa_{1,n}+\kappa_{2,n}\label{step1casei}\\
&\leq \rmQ\Bigg(\frac{L_1+O(\sqrt{n}\xi_n)}{\sqrt{\rmV(P_X,D_1)}}\Bigg)+\frac{6\rmT(P_X,D_1)}{\sqrt{n}\rmV(P_X,D_1)}+\frac{2|\calX|}{n^2}+\kappa_{1,n}+\kappa_{2,n}\label{laststepcasei},
\end{align}
where $\rmT(P_X,D_1)$ is the third absolute moment of $\jmath(X,D_1|P_X)$ (which is finite for DMSes) and \eqref{laststepcasei} follows by applying the Berry-Esseen theorem to the first term in \eqref{step1casei}.

Hence, if we choose $(L_1,L_2)$ such that
\begin{align}
L_1\geq \sqrt{\rmV(P_X,D_1)}\rmQ^{-1}(\epsilon),
\end{align}
then we have $\limsup_{n\to\infty} \varepsilon_n^{\rm{FY}}(D_1,D_2)\leq \epsilon$.

\item Case (ii) $R_1=R(P_X,D_1)$ and $R_1+R_2=\rvR_{\rm{FY}}(R_1,D_1,D_2|P_X)+H(P_Y)$

In this case, we still have $R_2>H(P_Y)$. Hence, invoking \eqref{upperbd1}, we obtain
\begin{align}
\nn&1-\varepsilon_n^{\rm{FY}}(D_1,D_2)\\
\nn&\geq \Pr\bigg(R_1+\frac{L_1}{\sqrt{n}}\geq \frac{1}{n}\sum_{i=1}^n \jmath(X_i,D_1|P_X)+O(\xi_n),\\
&\qquad R_1+R_2+\frac{(1+s^*)L_1+L_2}{\sqrt{n}}<\frac{1}{n}\sum_{i=1}^n \Big(\jmath_{\rm{FY}}(X_i,g(X_i)|R_1,D_1,D_2,P_X)-\log P_Y(Y_i)\Big)+O(\xi_n)\bigg)-\frac{2|\calX|}{n^2}-\kappa_{1,n}\\
&\geq 1-\Psi\big(L_1+O(\xi_n),(1+s^*)L_1+L_2+O(\xi_n);\mathbf{0};\mathbf{V}_1(R_1,D_1,D_2|P_X)\Big)+O\Big(\frac{1}{\sqrt{n}}\big)-\frac{2|\calX|}{n^2}-\kappa_{1,n}.
\end{align}
Hence, if we choose $(L_1,L_2)$ such that 
\begin{align}
\Psi\big(L_1,(1+s^*)L_1+L_2;\mathbf{0};\mathbf{V}_1(R_1,D_1,D_2|P_X)\big)\geq 1-\epsilon,
\end{align}
we have $\limsup_{n\to\infty} \varepsilon_n^{\rm{FY}}(D_1,D_2)\leq \epsilon$.

\item Case (iii) $R(P_X,D_1)<R_1<\rvR_{\rm{FY}}(R_1,D_1,D_2|P_X)$, $R_1+R_2=\rvR_{\rm{FY}}(R_1,D_1,D_2|P_X)+H(P_Y)$

In this case, it is still true that $R_2>H(P_Y)$. The analysis is similar to Case (i). It can be verified that if we choose $(L_1,L_2)$ such that 
\begin{align}
(1+s^*)L_1+L_2\geq \sqrt{\rmV(R_1,D_1,D_2|P_X)}\rmQ^{-1}(\epsilon),
\end{align}
we obtain $\limsup_{n\to\infty} \varepsilon_n^{\rm{FY}}(D_1,D_2)\leq \epsilon$.

\item Case (iv) $R_1=\rvR_{\rm{FY}}(R_1,D_1,D_2|P_X)$ and $R_2=H(P_Y)$

The analysis is similar to Case (ii). It can be verified that if 
\begin{align}
\Psi((1+s^*)L_1+L_2,L_2;\mathbf{0};\mathbf{V}_2(R_1,D_1,D_2|P_X))\geq 1-\epsilon,
\end{align}
we have $\limsup_{n\to\infty}\varepsilon_n^{\rm{FY}}(D_1,D_2)\leq \epsilon$.

\item Case (v) $R_1>\rvR_{\rm{FY}}(R_1,D_1,D_2|P_X)$ and $R_2=H(P_Y)$

The analysis is similar to Case (i). It can be verified that if
\begin{align}
L_2\geq \sqrt{\rmV(P_Y)}\rmQ^{-1}(\epsilon),
\end{align}
we have $\limsup_{n\to\infty}\varepsilon_n^{\rm{FY}}(D_1,D_2)\leq \epsilon$.
\end{itemize}
The proof of the achievability part is now complete.

\subsection{Converse Coding Theorem}
The major step stone is the type-based ``strong converse''.
\begin{lemma}
\label{mddstrongconverse}
Fix $c>0$ and a type $Q_X\in\calP_n(P_X)$. For any $(n,M_1,M_2)$-code such that
\begin{align}
\Pr\Big(d_1(X^n,\hatX_1^n)\leq D_1,~d_2(X^n,\hatX_2^n)\leq D_2,~\hatY^n=Y^n|X^n\in\calT_{Q_X}\Big)&\geq \exp(-nc)\label{mddassump},
\end{align}
there exists a conditional distribution $Q_{\hatX_1\hatX_2|X}$ such that
\begin{align}
\frac{1}{n}\log M_1&\geq I(Q_X,Q_{\hatX_1|X})-\xi_{1,n},\\
\frac{1}{n}\log M_2&\geq H(Q_Y)-\xi_{2,n},\\
\frac{1}{n}\log M_1M_2&\geq H(Q_Y)+I(Q_Y,Q_{\hatX_1|Y})+I(Q_{X|Y},Q_{\hatX_1\hatX_2|XY}|Q_Y)-\xi_{1,n}-\xi_{2,n},
\end{align}
where 
\begin{align}
\xi_{1,n}&=\frac{|\calX|\log (n+1)+\log n+nc}{n},\\
\xi_{2,n}&=2\xi_{1,n}+\frac{2(\log n+nc)+|\calX|\cdot|\calY|\log (n+1)+\log |\calY|+h_b(1/n)}{n}.
\end{align}
and $Q_{X|Y}$, $Q_{\hatX_1|Y}$, $Q_{\hatX_1\hatX_2|XY}$ are induced by $Q_X$, $Q_{\hatX_1\hatX_2|X}$ and the deterministic function $y=g(x)$.

Further, the expected distortions are bounded as
\begin{align}
\mathbb{E}_{Q_X\times Q_{\hatX_1\hatX_2|X}}[d_1(X,\hatX_1)]&\leq D_1+\frac{\overline{d}_1}{n}:=D_{1,n},\\
\mathbb{E}_{Q_X\times Q_{\hatX_1\hatX_2|X}}[d_2(X,\hatX_2)]&\leq D_2+\frac{\overline{d}_2}{n}:=D_{2,n}.
\end{align}
\end{lemma}
The proof of Lemma \ref{mddstrongconverse} is similar to \cite[Lemma 6]{watanabe2015second}, \cite[Lemma 12]{zhou2016second} and thus omitted. The main technique is the perturbation approach by Gu and Effros~\cite{wei2009strong} and the generalization with method of types~\cite{watanabe2015second}.

Let $c=\log n/n$, then we have
\begin{align}
\xi_{1,n}&=\frac{|\calX|\log (n+1)+2\log n}{n},\\
\xi_{2,n}&=\frac{8\log n+(|\calX|\cdot|\calY|+2|\calX|)\log (n+1)+\log |\calY|+h_b(1/n)}{n}.
\end{align}
Define
\begin{align}
R_{i,n}=\frac{1}{n}(\log M_i+n\xi_{i,n}),~i=1,2.
\end{align}

Invoking Lemma \ref{mddstrongconverse}, we can prove a lower bound on the joint excess-distortion and error probability for any $(n,M_1,M_2)$-code.
\begin{lemma}
\label{mddlbjoint}
Any $(n,M_1,M_2)$-code satisfies that
\begin{align}
\varepsilon_n^{\rm{FY}}(D_1,D_2)
\nn&\geq \Pr\Big(R_{1,n}<R(\hatT_{X^n},D_{1,n})~\mathrm{or}~R_{2,n}<H(\hatT_{g(X^n)})~\mathrm{or}\\
&\qquad \qquad R_{1,n}+R_{2,n}<\rvR_{\rm{FY}}(R_{1,n},D_{1,n},D_{2,n}|\hatT_{X^n})+H(\hatT_{g(X^n)})\Big).
\end{align}
\end{lemma}
The proof of Lemma \ref{mddlbjoint} is similar to \cite[Lemma 7]{watanabe2015second}, \cite[Lemma 13]{zhou2016second} and thus omitted.

The rest of the converse proof is omitted since it is analogous to the achievability proof where we Taylor expand several quantities and then apply (multi-variate) Berry-Esseen theorems for each case.

\section{Conclusion}
\label{sec:conclusion}
In this paper, we revisit two lossy source coding problems: the Kaspi problem~\cite[Theorem 1]{kaspi1994} of the lossy source coding problem with side information available at the encoder and one of the two decoders and the Fu-Yeung problem~\cite{fu2002rate} of the multiple description coding problem with one semi-deterministic distortion measure. For both problems, we present the properties of optimal test channels for certain rate-distortion functions, non-asymptotic converse bounds and the refined asymptotics for DMSes under bounded distortion measures. 

In the future, one may derive exact second-order asymptotics for GMSes under quadratic distortion measures for the Kaspi problem~\cite{perron2005kaspi} and the multiple description coding problem~\cite{ozarow1980source}. For the Kaspi problem, the converse part of second-order asymptotics for abstract memoryless sources follows directly from the non-asymptotic converse bound. Thus, one just need to prove a matching achievability second-order coding rate. Another future work is to extend the ideas of deriving non-asymptotic converse bounds in this paper to other multi-terminal source coding problems, such as the Gray-Wyner problem~\cite{gray1974source} and the multiple description coding problem~\cite{wolf1980source}.

\appendix
\subsection{Proof of Lemma \ref{nule1}}
\label{proofnule1}

For simplicity, we prove the claim for the discrete case. The proof can be extended to continuous case easily by replacing summations with integrals and log-sum inequality with its continuous version~\cite[Lemma 2.14]{donsker1975asymptotic}. Given any distributions $(P_{\hatX_1|XY},P_{\hatX_2|XY\hatX_1})$ and $(Q_{\hatX_1},Q_{\hatX_2|Y\hatX_1})$, let $P_{\hatX_1}$, $P_{X\hatX_1}$, $P_{Y\hatX_1}$, $P_{\hatX_2|Y\hatX_1}$, $P_{XY\hatX_1}$, $P_{X\hatX_2}$ be induced by $(P_{\hatX_1|XY},P_{\hatX_2|XY\hatX_1})$.

First, for $\hatx_1$ such that $P_{\hatX_1}^*(\hatx_1)>0$, using the result in \eqref{optcond14kaspi} and the definition of $\nu_1(\cdot)$ in \eqref{def:nu1}, we conclude that
\begin{align}
\nu_1(\hatx_1)
&=\sum_{x,y}P_{XY}(x,y)\frac{P_{\hatX_1|XY}^*(\hatx_1|x,y)}{P_{\hatX_1}^*(\hatx_1)}=1\label{nu2typical}.
\end{align}

Further, using the results in \eqref{optcond14kaspi} and \eqref{optcond24kaspi}, and the definition of $\nu_2(\cdot)$ in \eqref{def:nu2}, we conclude that for all $\hatx_1$ satisfying that $P_{\hatX_1}^*(\hatx_1)>0$,
\begin{align}
\nn&\nu_2(\hatx_1,P_{\hatX_2|Y\hatX_1}^*)\\
&=\sum_{x,y,\hatx_2}P_{XY}(x,y)P_{\hatX_2|Y\hatX_1}^*(\hatx_2|y,\hatx_1)\alpha(x,y)\exp(-\lambda_1^*d_1(x,\hatx_1)-\lambda_2^*d_2(x,\hatx_2))\\
&=\sum_{x,y,\hatx_2}P_{XY}(x,y)P_{\hatX_2|Y\hatX_1}^*(\hatx_2|y,\hatx_1)\frac{P_{\hatX_2|XY\hatX_1}^*(\hatx_2|x,y,\hatx_1)}{P_{\hatX_2|Y\hatX_1}^*(\hatx_2|y,\hatx_1)}\frac{P_{\hatX_1|XY}^*(\hatx_1|x,y)}{P_{\hatX_1}^*(\hatx_1)}\\
&=\sum_{x,y,\hatx_2}\frac{P_{XY}(x,y)P_{\hatX_1\hatX_2|XY}^*(\hatx_1,\hatx_2|x,y)}{P_{\hatX_1}^*(\hatx_1)}\\
&=1\label{sumeq2fortypicalhatx1}.
\end{align}
 
Then, we will show that $\nu_2(\hatx_1,P_{\hatX_2|Y\hatX_1}^*)\leq 1$ holds for all $\hatx_1$ such that $P_{\hatX_1}^*(\hatx_1)>0$ and arbitrary distribution $Q_{\hatX_2|Y\hatX_1}$. 

Given $(y,\hatx_1)$, using the definition of $\alpha_2(\cdot|\cdot)$ in \eqref{def:a2q4kaspi} and $\alpha_2(\cdot)$ in \eqref{def:a24kaspi}, and the log-sum inequality, we obtain that
\begin{align}
\nn&\inf_{P_{\hatX_2|X,y,\hatx_1}}\inf_{Q_{\hatX_2|y,\hatx_1}}D(P_{\hatX_2|X,\hatx_1,y}\|Q_{\hatX_2|y,\hatx_1}|P_{X|y,\hatx_1}^*)+\lambda_2^*\mathbb{E}_{P_{X\hatX_2|y,\hatx_1}}[d_2(X,\hatX_2)]\\
&=\inf_{P_{\hatX_2|X,y,\hatx_1}}D(P_{\hatX_2|X,\hatx_1,y}\|P_{\hatX_2|y,\hatx_1}|P_{X|y,\hatx_1}^*)+\lambda_2^*\mathbb{E}_{P_{X\hatX_2|y,\hatx_1}}[d_2(X,\hatX_2)]\\
&=\mathbb{E}_{P_{X|y,\hatx_1}^*}[\log \alpha_2(X,y,\hatx_1)]\label{optalpha2},
\end{align}
and
\begin{align}
\nn&\inf_{P_{\hatX_2|X,y,\hatx_1}}
D(P_{\hatX_2|X,y,\hatx_1}\|Q_{\hatX_2|y,\hatx_1}|P_{X|y,\hatx_1}^*)+\lambda_2^*\mathbb{E}_{P_{X\hatX_2|y,\hatx_1}}[d_2(X,\hatX_2)]\\
&=\mathbb{E}_{P_{X|y,\hatx_1}^*}[\log \alpha_2(X,y,\hatx_1|Q_{\hatX_2|Y\hatX_1})]\label{suboptalpha2}.
\end{align}
From \eqref{optalpha2} and \eqref{suboptalpha2}, we have
\begin{align}
\inf_{Q_{\hatX_2|y,\hatx_1}}\mathbb{E}_{P_{X|y,\hatx_1}^*}[\log \alpha_2(X,y,\hatx_1|Q_{\hatX_2|Y\hatX_1})]
&=\mathbb{E}_{P_{X|y,\hatx_1}^*}[\log \alpha_2(X,y,\hatx_1)]\label{optqhatx2gyhatx1}.
\end{align}

Fix $\epsilon\in[0,1]$ and let $b\in\hat{\calX}_2$. Choose $Q_{\hatX_2|y,\hatx_1}=(1-\epsilon)P_{\hatX_2|x,\hatx_1}^*(\hatx_2)+\epsilon1\{\hatx_2=b\}$. Using the definitions in \eqref{def:a2q4kaspi} and \eqref{def:a24kaspi}, we obtain that
\begin{align}
\frac{1}{\alpha_2(x,y,\hatx_1|Q_{\hatX_2|Y\hatX_1})}
&=\frac{1-\epsilon}{\alpha_2(x,y,\hatx_1)}+\epsilon \exp(-\lambda_2^*d_2(x,b)).
\end{align}
Thus,
\begin{align}
\nn&\frac{\partial \mathbb{E}_{P_{X|y,\hatx_1}^*}[\log \alpha_2(X,y,\hatx_1|Q_{\hatX_2|Y\hatX_1})]}{\partial \epsilon}\bigg|_{\epsilon=0}\\
&=\mathbb{E}_{P_{X|y,\hatx_1}^*}\left[\alpha_2(X,y,\hatx_1|Q_{\hatX_2|Y\hatX_1})\Big(\frac{1}{\alpha_2(X,y,\hatx_1)}-\exp(-\lambda_2^*d_2(x,b))\Big)\right]\bigg|_{\epsilon=0}\\
&=1-\mathbb{E}_{P_{X|y,\hatx_1}^*}[\alpha_2(X,y,\hatx_1)\exp(-\lambda_2^*d_2(x,b))]\\
&\geq 0\label{nonnegativeopt1},
\end{align}
where \ref{nonnegativeopt1} follows from \eqref{optqhatx2gyhatx1} which indicates that the derivative of the quantity at $\epsilon=0$ is nonnegative since the minimum is achieved by the test channel $P_{\hatX_2|Y\hatX_1}^*$. 
Using the result in \eqref{optcond14kaspi}, we obtain that for $\hatx_1$ such that $P_{\hatX_1}^*(\hatx_1)>0$,
\begin{align}
P_{X|y,\hatx_1}(x)
&=\frac{P_{XY}(x,y)P_{\hatX_1|XY}^*(\hatx_1|x,y)}{\sum_{x}P_{XY}(x,y)P_{\hatX_1|XY}^*(\hatx_1|x,y)}\\
&=\frac{P_{XY}(x,y)P_{\hatX_1}^*(\hatx_1)\alpha(x,y)\exp(-\lambda_1^*d_1(x,\hatx_1))/\alpha_2(x,y,\hatx_1)}{\sum_x P_{XY}(x,y)P_{\hatX_1}^*(\hatx_1)\alpha(x,y)\exp(-\lambda_1^*d_1(x,\hatx_1))/\alpha_2(x,y,\hatx_1)}\\
&=\frac{P_{XY}(x,y)\alpha(x,y)\exp(-\lambda_1d_1(x,\hatx_1))/\alpha_2^*(x,y,\hatx_1)}{\sum_x P_{XY}(x,y)\alpha(x,y)\exp(-\lambda_1^*d_1(x,\hatx_1))/\alpha_2(x,y,\hatx_1)}.
\end{align}
Hence, from \eqref{nonnegativeopt1}, we obtain
\begin{align}
1
&\geq \mathbb{E}_{P_{X|y,\hatx_1}^*}[\alpha_2(X,y,\hatx_1)\exp(-\lambda_2^*d_2(x,b))]\\
&=\sum_x \frac{P_{XY}(x,y)\alpha(x,y)\exp(-\lambda_1^*d_1(x,\hatx_1))/\alpha_2(x,y,\hatx_1)}{\sum_x P_{XY}(x,y)\alpha(x,y)\exp(-\lambda_1^*d_1(x,\hatx_1))/\alpha_2(x,y,\hatx_1)}\alpha_2(x,y,\hatx_1)\exp(-\lambda_2^*d_2(x,b)).
\end{align}
Thus, we have
\begin{align}
\nn&\sum_x P_{XY}(x,y)\alpha(x,y)\exp(-\lambda_1^*d_1(x,\hatx_1)-\lambda_2^*d_2(x,b))\\
&\leq \sum_x P_{XY}(x,y)\frac{\alpha(x,y)}{\alpha_2(x,y,\hatx_1)}\exp(-\lambda_1^*d_1(x,\hatx_1))\label{nowstep1}.
\end{align}
Fix a arbitrary distribution $Q_{\hatX_2|Y\hatX_1}$. Multiplying $Q_{\hatX_2|Y\hatX_1}(\hatx_2|y,\hatx_1)$ to both sides of \eqref{nowstep1}, summing over $(y,\hatx_2)$, and using the definitions of $\nu_1(\cdot)$ in \eqref{def:nu1}, $\nu_2(\cdot)$ in \eqref{def:nu2}, from \eqref{nu2typical}, we obtain that
\begin{align}
\nu_2(\hatx_1,Q_{\hatX_2|Y\hatX_1})\leq \nu_1(\hatx_1)\leq 1.
\end{align}

Finally, we will show that Lemma \ref{nule1} holds for all $\hatx_1$ such that $P_{\hatX_1}^*(\hatx_1)=0$ and arbitrary $Q_{\hatX_2|Y\hatX_1}$. Given $y$, define
\begin{align}
\nn&F(P_{\hatX_1|XY},P_{\hatX_2|XY\hatX_1},Q_{\hatX_1},Q_{\hatX_2|Y\hatX_1},y)\\
&:=D(P_{\hatX_1|X,y}\|Q_{\hatX_1}|P_{X|y})+\lambda_1^*\mathbb{E}_{P_{X\hatX_1|y}}[d_1(X,\hatX_1)]+D(P_{\hatX_2|Xy\hatX_1}\|Q_{\hatX_2|y,\hatX_1}|P_{X\hatX_1|y})+\lambda_2^*\mathbb{E}_{P_{X\hatX_2|y,\hatX_1}}[d_2(X,\hatX_2)]\\
\nn&=D(P_{\hatX_1|X,y}\|P_{\hatX_1}|P_{X|y})+D(P_{\hatX_1}|Q_{\hatX_1})+\lambda_1^*\mathbb{E}_{P_{X\hatX_1|y}}[d_1(X,\hatX_1)]\\
&\quad+D(P_{\hatX_2|Xy\hatX_1}\|P_{\hatX_2|y,\hatX_1}|P_{X\hatX_1|y})+D(P_{\hatX_2|y,\hatX_1}\|Q_{\hatX_2|y,\hatX_1}|P_{\hatX_1|y})+\lambda_2^*\mathbb{E}_{P_{X\hatX_2|y}}[d_2(X,\hatX_2)]\label{def:barL2}.
\end{align}
Using the definitions of $\alpha(\cdot|\cdot)$ in \eqref{def:a1q4kaspi} and $\alpha(\cdot)$ in \eqref{def:a4kaspi}, and the log-sum inequality, we obtain 
\begin{align}
\inf_{P_{\hatX_1|XY},P_{\hatX_2|XY\hatX_1},Q_{\hatX_1},Q_{\hatX_2|Y\hatX_1}}F(P_{\hatX_1|XY},P_{\hatX_2|XY\hatX_1},Q_{\hatX_1},Q_{\hatX_2|Y\hatX_1},y)
&=\mathbb{E}_{P_{X|y}}[\log \alpha(X,y)]\label{op1ta1},
\end{align}
and
\begin{align}
\inf_{P_{\hatX_1|XY},P_{\hatX_2|XY\hatX_1}}F(P_{\hatX_1|XY},P_{\hatX_2|XY\hatX_1},Q_{\hatX_1},Q_{\hatX_2|Y\hatX_1},y)
&=\mathbb{E}_{P_{X|y}}[\log \alpha(X,y|Q_{\hatX_1},Q_{\hatX_2|Y\hatX_1})]\label{subopta1}.
\end{align}
From \eqref{op1ta1} and \eqref{subopta1}, we conclude that
\begin{align}
\inf_{Q_{\hatX_1},Q_{\hatX_2|Y\hatX_1}}\mathbb{E}_{P_{X|y}}[\log \alpha(X,y|Q_{\hatX_1},Q_{\hatX_2|Y\hatX_1})]=\mathbb{E}_{P_{X|y}}[\log \alpha(X,y)]\label{tobeused*}.
\end{align}

Fix $\epsilon\in[0,1]$ and $a\in\hat{\calX_1}$ such that $P_{\hatX_1}^*(a)=0$. Let $\bar{Q}_{\hatX_2|Y\hatX_1}$ be arbitrary distribution.
Now choose 
\begin{align}
Q_{\hatX_1}(\hatx_1)&=(1-\epsilon)P_{\hatX_1}^*(\hatx_1)+\epsilon1\{\hatx_1=a\}\label{chooseqhx1}
\end{align}
and
\begin{align}
Q_{\hatX_2|Y\hatX_1}(\hatx_2|y,\hatx_1)
&=\left\{
\begin{array}{ll}
P_{\hatX_2|Y\hatX_1}^*(\hatx_2|y,\hatx_1)&\mathrm{if}~P_{\hatX_1}^*(\hatx_1)>0\\
\bar{Q}_{\hatX_2|Y\hatX_1}(\hatx_2|y,\hatx_1)&\mathrm{otherwise}
\end{array}
\right.\label{chooseqhx1x2}
\end{align}

Using the definitions of $\alpha_2(\cdot|\cdot)$ in \eqref{def:a2q4kaspi}, $\alpha_2(\cdot)$ in \eqref{def:a24kaspi}, and the choices of distributions in \eqref{chooseqhx1} and \eqref{chooseqhx1x2}, we obtain that
\begin{align}
\frac{1}{\alpha(x,y|Q_{\hatX_1},Q_{\hatX_2|Y\hatX_1})}
&=\frac{(1-\epsilon)}{\alpha(x,y)}+\epsilon\frac{\exp(-\lambda_1^*d_1(x,a))}{\alpha_2(x,y,\hatx_1|\bar{Q}_{\hatX_2|Y\hatX_1})}\label{aetoa}.
\end{align}
Thus,
\begin{align}
\nn&\frac{\partial \mathbb{E}_{P_{X|y}}[\log\alpha(X,y|Q_{\hatX_1},Q_{\hatX_2|Y\hatX_1})]}{\partial \epsilon}\bigg|_{\epsilon=0}\\
&=\mathbb{E}_{P_{X|y}}\left[\alpha(X,y|Q_{\hatX_1},Q_{\hatX_2|Y\hatX_1})]\Big(\frac{1}{\alpha(X,y)}-\frac{\exp(-\lambda_1^*d_1(X,a))}{\alpha_2(X,y,\hatx_1|\bar{Q}_{\hatX_2|Y\hatX_1})}\Big)\right]\bigg|_{\epsilon=0}\\
&=1-\mathbb{E}_{P_{X|y}}\left[\frac{\alpha(X,y)\exp(-\lambda_1^*d_1(X,a))}{\alpha_2(X,y,\hatx_1|\bar{Q}_{\hatX_2|Y\hatX_1})}\right]\label{usechoiceq}\\
&\geq 0\label{optqhx1x2},
\end{align}
where \eqref{usechoiceq} follows from \eqref{aetoa}; \eqref{optqhx1x2} follows from \eqref{tobeused*} which implies that the minimum of $\mathbb{E}_{P_{X|y}}[\log \bar{\alpha}_1(X,y|Q_{\hatX_1},Q_{\hatX_2|Y\hatX_1})]$ is achieved at $\epsilon=0$. Note that \eqref{optqhx1x2} holds for arbitrary distribution $\bar{Q}_{\hatX_2|Y\hatX_1}$.

Choose $\bar{Q}_{\hatX_2|Y\hatX_1}=P_{\hatX_2|Y\hatX_1}^*$ and invoking \eqref{def:nu2} and \eqref{optqhx1x2}, we conclude that for any $a\in\hat{\calX_1}$ satisfying $P_{\hatX_1}^*(a)=0$, 
\begin{align}
\nu_1(a)
&=\mathbb{E}_{P_Y}\Big[\mathbb{E}_{P_{X|Y}}\Big[\frac{\alpha(X,Y)\exp(-\lambda_1^*d_1(X,a))}{\alpha_2(X,Y,\hatx_1)}\Big|Y\Big]\Big]\\
&\leq 1\label{sumineq4anyQhx2}.
\end{align}

Further, using the definition of $\alpha_2(\cdot|\cdot)$ in \eqref{def:a2q4kaspi} and the result in \eqref{optqhx1x2}, we obtain that for any $\bar{Q}_{\hatX_2|Y\hatX_1}$,
\begin{align}
1
&\geq \mathbb{E}_{P_{X|y}}\left[\frac{\alpha(X,y)\exp(-\lambda_1^*d_1(X,a))}{\alpha_2(X,y,\hatx_1|\bar{Q}_{\hatX_2|Y\hatX_1})}\right]\\
&=\sum_x P_{X|y}(x)\alpha(x,y)\exp(-\lambda_1^*d_1(x,a))\Big\{\sum_{\hatx_2}\bar{Q}_{\hatX_2|y,\hatx_1}(\hatx_2)\exp(-\lambda_2^*d_2(x,\hatx_2))\Big\}\\
&=\sum_{x,\hatx_2}P_{X|y}(x)\alpha(x,y)\bar{Q}_{\hatX_2|y,\hatx_1}(\hatx_2)\exp(-\lambda_1^*d_1(x,a)-\lambda_2^*d_2(x,\hatx_2))\label{nowstep2}.
\end{align}

The proof of Lemma \ref{opttest4kaspi} is complete by using the definition of $\nu_2(\cdot)$ in \eqref{def:nu2} and the results in \eqref{sumineq4anyQhx2}, \eqref{nowstep2}.

\subsection{Proof of Lemma \ref{oneshotconverse4kaspi}}
\label{proofoneshot}
Let $S$ be a random variable taking values in $\calM$. Let the encoding function $f$ be the random transformations  $P_{S|X}$. Let the decoding function $g_1$ be  $P_{\hatX_1|S}$ and let $g_2$ be $P_{\hatX_2|SY}$. Let $Q_S$ be the uniform distribution over $\calM$ and let $Q_{\hatX_1}$ and $Q_{\hatX_2|Y\hatX_1}$ be chosen such that
\begin{align}
Q_{\hatX_1}(\hatx_1)Q_{\hatX_2|Y}(\hatx_2|y)
&:=\sum_s Q_S(s)P_{\hatX_1|U}(\hatx_1|s)P_{\hatX_2|SY}(\hatx_2|s,y)\label{choseqhatx1},\\
Q_{\hatX_2|Y\hatX_1}(\hatx_2|y,\hatx_1)&:=Q_{\hatX_2|Y}(\hatx_2|y)\label{chooseqhatx2gyhatx1}.
\end{align}

Then, we have that for any $\gamma>0$,
\begin{align}
\nn&\Pr\left(\jmath_{\rm{K}}(X,Y|D_1,D_2,P_{XY})\geq \log M+\gamma\right)\\
\nn&=\Pr\left(\jmath_{\rm{K}}(X,Y|D_1,D_2,P_{XY})\geq \log M+\gamma,~d_1(X,\hatX_1)\leq D_1,~d_2(X,\hatX_2)\leq D_2\right)\\
&\quad+\Pr\left(\jmath_{\rm{K}}(X,Y|D_1,D_2,P_{XY})\geq \log M+\gamma,~d_1(X,\hatX_1)> D_1,~\mathrm{or}~d_2(X,\hatX_2)> D_2\right)\\
&\leq 
\Pr\left(\jmath_{\rm{K}}(X,Y|D_1,D_2,P_{XY})\geq \log M+\gamma,~d_1(X,\hatX_1)\leq D_1,~d_2(X,\hatX_2)\leq D_2\right)+\varepsilon_1^{\rmK}(D_1,D_2)\label{tobeupperbd}.
\end{align}
The first term in \eqref{tobeupperbd} can be upper bounded as follows:
\begin{align}
\nn&\Pr\left(\jmath_{\rm{K}}(X,Y|D_1,D_2,P_{XY})\geq \log M+\gamma,~d_1(X,\hatX_1)\leq D_1,~d_2(X,\hatX_2)\leq D_2\right)\\
&=\Pr\left(M\leq \exp(\jmath_{\rm{K}}(X,Y|D_1,D_2,P_{XY})-\gamma)1\{d_1(X,\hatX_1)\leq D_1,~d_2(X,\hatX_2)\leq D_2\}\right)\\
&\leq \frac{\exp(-\gamma)}{M}\mathbb{E}[\exp(\jmath_{\rm{K}}(X,Y|D_1,D_2,P_{XY}))1\{d_1(X,\hatX_1)\leq D_1,~d_2(X,\hatX_2)\leq D_2\}]\label{usemarkovineq}\\
&\leq \frac{\exp(-\gamma)}{M}\mathbb{E}\Big[\exp\Big(\jmath_{\rm{K}}(X,Y|D_1,D_2,P_{XY})+\lambda_1^*(D_1-d_1(X,\hatX_1))+\lambda_2^*(D_2-d_2(X,\hatX_2))\Big)\Big]\\
&\leq \frac{\exp(-\gamma)}{M}\mathbb{E}_{P_{XY}}\Big[
\sum_s P_{S|X}(s|X)\mathbb{E}\Big[\alpha(X,Y)\exp(-\lambda_1^*d_1(X,\hatX_1)-\lambda_2^*d_2(X,\hatX_2)|S=s,Y)\Big]\Big]\label{oneshotusea1*}\\
&\leq \exp(-\gamma)\mathbb{E}_{P_{XY}}\Big[\mathbb{E}_{Q_{\hatX_1}\times Q_{\hatX_2|Y\hatX_1}}\Big[\alpha(X,Y)\exp(-\lambda_1^*d_1(X,\hatX_1)-\lambda_2^*d_2(X,\hatX_2)\Big|Y)\Big]\Big]\label{proplessthan1}\\
&=\exp(-\gamma)\mathbb{E}_{Q_{\hatX_1}}\Big[\nu_2(\hatX_1,Q_{\hatX_2|Y\hatX_1})\Big]\label{oneshotusenu1*}\\
&\leq \exp(-\gamma)\label{usesumine},
\end{align}
where \eqref{usemarkovineq} follows from Markov's inequality; \eqref{oneshotusea1*} follows from \eqref{def:kaspitilt}; \eqref{proplessthan1} follows from the fact that $P_{S|X}(s|x)\leq 1$ for any $(s,x)$ and the choice of $(Q_{\hatX_1},Q_{\hatX_2|Y\hatX_1})$ in \eqref{choseqhatx1} and \eqref{chooseqhatx2gyhatx1}; \eqref{oneshotusenu1*} follows from the definition of $\nu_2(\cdot)$ in \eqref{def:nu2}; and \eqref{usesumine} follows from \eqref{nu2le1}.

The proof of Lemma \ref{oneshotconverse4kaspi} is complete by invoking \eqref{tobeupperbd} and \eqref{usesumine}.

\subsection{Proof of Lemma \ref{proptest4fy}}
\label{proofproptest4fy}
Note that the Markov chain $Y-X-(\hatX_1,\hatX_2)$ holds in the Fu-Yeung problem due to the fact that $Y$ is a function of $X$. From \eqref{equivalent}, we obtain that
\begin{align}
I(\hatX_1;Y)+I(X;\hatX_1\hatX_2|Y)
&=I(XY;\hatX_1)+I(X;\hatX_2|Y\hatX_1)\\
&=I(X;\hatX_1)+I(X;\hatX_2|Y\hatX_1).
\end{align}

Hence, the optimization problem in \eqref{rvrmin} is equivalent to 
\begin{align}
\rvR_{\mathrm{FY}}(R_1,D_1,D_2|P_{XY})&=\min_{\substack{(P_{\hatX_1|X},P_{\hatX_2|X\hatX_1}):\\R_1\geq I(X;\hatX_1)\\D_1\geq \mathbb{E}[d_1(X,\hatX_1)]\\D_2\geq\mathbb{E}[d_2(X,\hatX_2)]}} I(X;\hatX_1)+I(X;\hatX_2|Y\hatX_1)\label{rvrmin2}.
\end{align}
Therefore, we can prove the properties of the optimal channels for $\rvR_{\mathrm{FY}}(R_1,D_1,D_2|P_{XY})$ using \eqref{rvrmin2}. Further, due to the Markov chain $Y-X-(\hatX_1,\hatX_2)$, we have that $P_{\hatX_2|XY\hatX_1}=P_{\hatX_2|X\hatX_1}$ and thus
\begin{align}
I(X;\hatX_2|Y\hatX_1)
&=D(P_{\hatX_2|X\hatX_1}\|P_{\hatX_2|Y\hatX_1}|P_{XY\hatX_1})\\
&=\sum_{x,y,\hatx_1,\hatx_2}P_X(x)P_{\hatX_1|X}(\hatx_1|x)P_{\hatX_2|X\hatX_1}(\hatx_2|x,\hatx_1)1\{y=g(x)\}\log \frac{P_{\hatX_2|X\hatX_1}(\hatx_2|x,\hatx_1)}{P_{\hatX_2|Y\hatX_1}(\hatx_2|y,\hatx_1)}\\
&=\sum_{x,\hatx_1,\hatx_2}P_X(x)P_{\hatX_1|X}(\hatx_1|x)P_{\hatX_2|X\hatX_1}(\hatx_2|x,\hatx_1)\log \frac{P_{\hatX_2|X\hatX_1}(\hatx_2|x,\hatx_1)}{P_{\hatX_2|Y\hatX_1}(\hatx_2|g(x),\hatx_1)}\label{mutualygx}.
\end{align}

Given any distributions $(Q_{\hatX_1},Q_{\hatX_2|Y\hatX_1})$ and non-negative numbers $(s,t_1,t_2)$, define
\begin{align}
\nn&F(P_{\hatX_1|X},P_{\hatX_2|X\hatX_1},Q_{\hatX_1},Q_{\hatX_2|Y\hatX_1},s,t_1,t_2)\\
&:=(1+s)D(P_{\hatX_1|X}\|Q_{\hatX_1}|P_X)+D(P_{\hatX_2|X\hatX_1}\|Q_{\hatX_2|Y\hatX_1}|P_{XY\hatX_1})+t_1\mathbb{E}_{P_{X\hatX_1}}[d_1(X,\hatX_1)]+t_2\mathbb{E}_{P_{X\hatX_2}}[d_2(X,\hatX_2)]\\
\nn&=(1+s)D(P_{\hatX_1|X}\|P_{\hatX_1}|P_X)+D(P_{\hatX_1}\|Q_{\hatX_1})+
t_1\mathbb{E}_{P_{X\hatX_1}}[d_1(X,\hatX_1)]\\
&\qquad+D(P_{\hatX_2|X\hatX_1}\|P_{\hatX_2|Y\hatX_1}|P_{XY\hatX_1})+D(P_{\hatX_2|Y\hatX_1}\|Q_{\hatX_2|Y\hatX_1}|P_{Y\hatX_1})+t_2\mathbb{E}_{P_{X\hatX_2}}[d_2(X,\hatX_2)],\label{def:Fppqqstt}
\end{align}
and
\begin{align}
F(P_{\hatX_1|X},P_{\hatX_2|X\hatX_1},s,t_1,t_2)
&:=F(P_{\hatX_1|X},P_{\hatX_2|X\hatX_1},P_{\hatX_1},P_{\hatX_2|Y\hatX_1},s,t_1,t_2)\\
&=\inf_{Q_{\hatX_1},Q_{\hatX_2|Y\hatX_1}}F(P_{\hatX_1|X},P_{\hatX_2|X\hatX_1},Q_{\hatX_1},Q_{\hatX_2|Y\hatX_1},s,t_1,t_2)\label{def:Fppstt},\\
F(s,t_1,t_2)&:=\inf_{P_{\hatX_1|XY},P_{\hatX_2|XY\hatX_1}}F(P_{\hatX_1|X},P_{\hatX_2|X\hatX_1},s,t_1,t_2).\label{def:Fsst}
\end{align}
Considering the dual problem of \eqref{rvrmin2} and using the definitions of $s^*$ in \eqref{def:s*}, $t_1^*$ in \eqref{def:t1*}, $t_2^*$ in \eqref{def:t2*}, we conclude that 
\begin{align}
\rvR_{\rm{FY}}(R_1,D_1,D_2|P_X)
&=\max_{(s,t_1,t_2)\in\bbR_+}F(s,t_1,t_2)-s R_1-t_1D_1-t_2D_2\\
&=F(s^*,t_1^*,t_2^*)-s^*R_1-t_1^*D_1-t_2^*D_2\label{dualeqcop}.
\end{align}

Using the definition of $\beta_2(\cdot|\cdot)$ in \eqref{def:b2q} and the log-sum inequality, we obtain that
\begin{align}
\nn&D(P_{\hatX_2|X\hatX_1}\|Q_{\hatX_2|Y\hatX_1}|P_{XY\hatX_1})+t_2^*\mathbb{E}_{P_{X\hatX_2}}[d_2(X,\hatX_2)]\\
&=\sum_{(x,y,\hatx_1)} P_{X}(x)1\{y=g(x)\}P_{\hatX_1|X}(\hatx_1|x)\Big\{\sum_{\hatx_2}P_{\hatX_2|X\hatX_1}(\hatx_2|x,\hatx_1)\Big(\log\frac{P_{\hatX_2|X\hatX_1}(\hatx_2|x,\hatx_1)}{Q_{\hatX_2|Y\hatX_1}(\hatx_2|y,\hatx_1)}+t_2^*d_2(x,\hatx_2)\Big)\Big\}\\
&\geq \sum_{x,y,\hatx_1} P_X(x)1\{y=g(x)\}P_{\hatX_1|X}(\hatx_1|x)\log \beta_2(x,y,\hatx_1|Q_{\hatX_2|Y\hatX_1})\label{log-sum},
\end{align}
where \eqref{log-sum} holds with equality if and only if
\begin{itemize}
\item For $(x,y,\hatx_1,\hatx_2)$ such that $P_{X}(x)1\{y=g(x)\}P_{\hatX_1|X}(\hatx_1|x)>0$,
\begin{align}
P_{\hatX_2|X\hatX_1}(\hatx_2|x,\hatx_1)=Q_{\hatX_2|Y\hatX_1}(\hatx_2|y,\hatx_1)\exp(-t_2^*d_2(x,\hatx_2))\beta_2(x,y,\hatx_1|Q_{\hatX_2|Y\hatX_1})\label{condition1}.
\end{align}
\item For $(x,y,\hatx_1,\hatx_2)$ such that $P_{X}(x)1\{y=g(x)\}P_{\hatX_1|X}(\hatx_1|x)=0$, $P_{\hatX_2|X\hatX_1}(\cdot|x,\hatx_1)$ is absolutely continuous with respect to $Q_{\hatX_2|Y\hatX_1}(\cdot|g(x),\hatx_1)$, i.e., for any $(x,y,\hatx_1,\hatx_2)$, if $Q_{\hatX_2|Y\hatX_1}(\hatx_2|y,\hatx_1)=0$, then we have $P_{\hatX_2|X\hatX_1}(\hatx_2|x,\hatx_1)=0$.
\end{itemize}
From \eqref{def:Fppqqstt} and \eqref{log-sum}, we obtain that
\begin{align}
\nn&F(P_{\hatX_1|X},P_{\hatX_2|X\hatX_1},Q_{\hatX_1},Q_{\hatX_2|Y\hatX_1},s^*,t_1^*,t_2^*)\\
&\geq (1+s^*)\sum_{x,y} P_X(x)1\{y=g(x)\}\Bigg\{\sum_{\hatx_1}P_{\hatX_1|X}(\hatx_1|x)\bigg(\log \frac{P_{\hatX_1|X}(\hatx_1|x)}{Q_{\hatX_1}(\hatx_1)}+\frac{t_1^*d_1(x,\hatx_1)+\log\beta_2(x,y,\hatx_1|Q_{\hatX_2|XY\hatX_1})}{1+s^*}\bigg)\Bigg\}\label{usea2q2}\\
&\geq (1+s^*)\sum_{x,y} P_X(x)1\{y=g(x)\}\log\beta(x,y|Q_{\hatX_1},Q_{\hatX_2|Y\hatX_1})\label{uselogsumagain},
\end{align}
where \eqref{uselogsumagain} follows by using the definition of $\beta_2(\cdot|\cdot)$ in \eqref{def:bq} and the log-sum inequality; and \eqref{uselogsumagain} holds with equality if and only if for all $(x,y)$ such that $P_X(x)1\{y=g(x)\}>0$,
\begin{align}
P_{\hatX_1|X}(\hatx_1|x)=\beta(x,y|Q_{\hatX_1},Q_{\hatX_2|Y\hatX_1})Q_{\hatX_1}(\hatx_1)\exp\Bigg(-\frac{t_1^*d_1(x,\hatx_1)+\log\beta_2(x,y,\hatx_1|Q_{\hatX_2|Y\hatX_1})}{1+s^*}\Bigg)\label{condition2}.
\end{align}

Thus, combing \eqref{def:Fsst}, \eqref{dualeqcop}, \eqref{log-sum} and \eqref{uselogsumagain} and noting that $P_{XY}(x,y)=P_X(x)1\{y=g(x)\}$, we obtain that
\begin{align}
\rvR_{\rm{FY}}(R_1,D_1,D_2)=(1+s^*)\mathbb{E}_{P_{XY}}[\log \beta(X,Y|P_{\hatX_1}^*,P_{\hatX_2|Y\hatX_1}^*)]-s^*R_1-t_1^*D_1-t_2^*D_2.
\end{align}
Further, from \eqref{condition1} and \eqref{condition2}, we obtain the properties of the optimal test channels in \eqref{optcond1} and \eqref{optcond2}. The result in \eqref{expantypical} follows directly from \eqref{optcond1} and \eqref{optcond2}. The proof of Lemma \ref{proptest4fy} is now complete.

\subsection{Proof of Lemma \ref{wleq1}}
\label{proofwleq1}
For any $\hatx_1$ such that $P_{\hatX_1}^*(\hatx_1)>0$, using the result in \eqref{optcond2} and the definition of $w_1(\cdot)$ in \eqref{def:w1}, we obtain that
\begin{align}
w_1^*(\hatx_1)=\sum_{x}P_X(x)\frac{P_{\hatX_1|X}^*(\hatx_1|x)}{P_{\hatX_1}^*(\hatx_1)}=1\label{nu1typical}.
\end{align}

Next, we will show that $w_1(\hatx_1)\leq 1$ for all $\hatx_1$ such that $P_{\hatX_1}^*(\hatx_1)=0$. Let $P_{X|Y}$ be induced by $P_X$ and the deterministic function $g$. For simplicity, we use $P_{A|b,C}(\cdot|\cdot)$ and $P_{A|BC}(\cdot|b,\cdot)$ interchangeable. Given $y$, define
\begin{align}
\nn&F(P_{\hatX_1|X},P_{\hatX_2|X\hatX_1},Q_{\hatX_1},Q_{\hatX_2|Y\hatX_1},y)\\
\nn&:=\sum_{x,\hatx_1,\hatx_2}P_{X|Y}(x|y)P_{\hatX_1|X}(\hatx_1|x)\Bigg\{(1+s^*)\log\frac{P_{\hatX_1|XY}(\hatx_1|x)}{Q_{\hatX_1}(\hatx_1)}+t_1^*d_1(x,\hatx_1)\\
&\qquad+P_{\hatX_2|X\hatX_1}(\hatx_2|x,\hatx_1)\Bigg(\log\frac{P_{\hatX_2|X\hatX_1}(\hatx_2|x,\hatx_1)}{Q_{\hatX_2|Y\hatX_1}(\hatx_2|y,\hatx_1)}+t_2^*d_2(x,\hatx_2)\Bigg)\Bigg\}\label{def:Fyppqqstt}.
\end{align}
Using the definitions of $\beta(\cdot|\cdot)$ in \eqref{def:bq}, $\beta(\cdot)$ in \eqref{def:b}, and the log-sum inequality, we conclude that
\begin{align}
\inf_{\substack{P_{\hatX_1|XY},P_{\hatX_2|XY\hatX_1}\\Q_{\hatX_1},Q_{\hatX_2|Y\hatX_1}}}F(P_{\hatX_1|XY},P_{\hatX_2|XY\hatX_1},Q_{\hatX_1},Q_{\hatX_2|Y\hatX_1},y)&=(1+s^*)\mathbb{E}_{P_{X|y}}[\log \beta(X,y)],\label{a1optimal}
\end{align}
and
\begin{align}
\inf_{\substack{P_{\hatX_1|XY},P_{\hatX_2|XY\hatX_1}}}F(P_{\hatX_1|XY},P_{\hatX_2|XY\hatX_1},Q_{\hatX_1},Q_{\hatX_2|Y\hatX_1},y)&=(1+s^*)\mathbb{E}_{P_{X|y}}[\log \beta(X,y|Q_{\hatX_1},Q_{\hatX_2|Y\hatX_1})].\label{a1suboptimal}
\end{align}
From \eqref{a1optimal} and \eqref{a1suboptimal}, we conclude that 
\begin{align}
\inf_{Q_{\hatX_1},Q_{\hatX_2|Y\hatX_1}}\mathbb{E}_{P_{X|y}}[\log \beta(X,y|Q_{\hatX_1},Q_{\hatX_2|Y\hatX_1})]=\mathbb{E}_{P_{X|y}}[\log \beta(X,y)].\label{a1minq}
\end{align}
Fix $\epsilon\in[0,1]$ and let $a\in\hat{\calX}_1$ be chosen arbitrarily such that $P_{\hatX_1}^*(a)=0$. Let $\barQ_{\hatX_2|Y\hatX_1}$ be an arbitrary distribution. Now choose 
\begin{align}
Q_{\hatX_1}(\hatx_1)=(1-\epsilon)P_{\hatX_1}^*(\hatx_1)+\epsilon1\{\hatx_1=a\},
\end{align}
and
\begin{align}
Q_{\hatX_2|Y\hatX_1}(\hatx_2|y,\hatx_1)
&=\left\{
\begin{array}{ll}
P_{\hatX_2|Y\hatX_1}^*(\hatx_2|y,\hatx_1)&\mathrm{if}~P_{\hatX_1}^*(\hatx_1)>0\\
\barQ_{\hatX_2|Y\hatX_1}(\hatx_2|y,\hatx_1)&\mathrm{otherwise}
\end{array}
\right.
\end{align}
Then, from the definitions of $\beta(\cdot|\cdot)$ in \eqref{def:bq} and $\beta(\cdot)$ in \eqref{def:b}, we obtain that
\begin{align}
\frac{1}{\beta(x,y|Q_{\hatX_1},Q_{\hatX_2|Y\hatX_1})}
&=\frac{1-\epsilon}{\beta(x,y)}+\epsilon\beta_2^{-\frac{1}{1+s^*}}(x,y,a|\barQ_{\hatX_2|Y\hatX_1})\exp\bigg(-\frac{t_1^*}{1+s^*}d_1(x,a)\bigg).
\end{align}
Hence, we have
\begin{align}
\nn&\frac{\partial \mathbb{E}_{P_{X|y}}[\log \beta(X,y|Q_{\hatX_1},Q_{\hatX_2|Y\hatX_1})]}{\partial \epsilon}\Bigg|_{\epsilon=0}\\
&=\mathbb{E}_{P_{X|y}}\Bigg[\beta(X,y|Q_{\hatX_1},Q_{\hatX_2|Y\hatX_1})]\Bigg\{\frac{1}{\beta(x,y)}-\beta_2^{-\frac{1}{1+s^*}}(X,y,a|Q_{\hatX_2|Y\hatX_1})\exp\bigg(-\frac{t_1^*}{1+s^*}d_1(X,a)\bigg)\Bigg\}\Bigg]\Bigg|_{\epsilon=0}\\
&=\mathbb{E}_{P_{X|y}}\Bigg[1-\beta(X,y)\beta_2^{-\frac{1}{1+s^*}}(X,y,a|Q_{\hatX_2|Y\hatX_1})\exp\bigg(-\frac{t_1^*}{1+s^*}d_1(X,a)\bigg)\Bigg]\\
&\geq 0\label{nonnegative2},
\end{align}
where \eqref{nonnegative2} follows from \eqref{a1minq}, which implies that the derivative of $\mathbb{E}_{P_{X|y}}[\log \beta(X,y|Q_{\hatX_1},Q_{\hatX_2|Y\hatX_1})]$ is non-negative at $\epsilon=0$ due to the fact that $(P_{\hatX_1}^*,P_{\hatX_2|Y\hatX_1}^*)$ achieves the minimum of the left hand side of \eqref{a1minq}. Letting $\barQ_{\hatX_2|Y\hatX_1}=P_{\hatX_2|Y\hatX_1}^*$, we obtain that $w(\hatx_1)\leq 1$ for all $a\in\hat{\calX}_1$ such that $P_{\hatX_1}^*(a)=0$. Till now, we have shown that $w(\hatx_1)\leq 1$ for all $\hatx_1\in\hat{\calX}_1$.

In the following, we will show that $w_2(\hatx_1,Q_{\hatX_2|Y\hatX_1})\leq 1$ holds for all $\hatx_1$ and arbitrary $Q_{\hatX_2|Y\hatX_1}$. First, let us focus on the $\hatx_1$ such that $P_{\hatX_1}^*(\hatx_1)>0$. Let $P_{Y\hatX_1}^*$ and $P_{X|Y\hatX_1}^*$ be induced by $P_X$, $P_{\hatX_1|X}^*$ and the deterministic function $g:\calX\to\calY$. Given $(y,\hatx_1)$ such that $P_{Y\hatX_1}^*(y,\hatx_1)>0$, let
\begin{align}
\nn&\psi(y,\hatx_1|P_{\hatX_2|X\hatX_1},Q_{\hatX_2|Y\hatX_1})\\
&:=\sum_{x,\hatx_2}P_{X|Y\hatX_1}^*(x|y,\hatx_1)P_{\hatX_2|X,\hatX_1}(\hatx_2|x,\hatx_1)\Bigg(\log\frac{P_{\hatX_2|X,\hatX_1}(\hatx_2
|x,\hatx_1)}{Q_{\hatX_2|Y\hatX_1}(\hatx_2|y,\hatx_1)}+t_2^*d_2(x,\hatx_2)\Bigg)\\
&=D(P_{\hatX_2|X,\hatx_1}\|Q_{\hatX_2|y,\hatx_1}|P_{X|y,\hatx_1}^*)+t_2^*\mathbb{E}_{P_{X|y,\hatx_1}^*\times P_{\hatX_2|X,\hatx_1}}[d_2(X,\hatX_2)]\label{def:psi}.
\end{align}

Using the definitions of $\beta_2(\cdot|\cdot)$ in \eqref{def:b2q}, $\beta(\cdot)$ in \eqref{def:b2} and $\psi(\cdot)$ in \eqref{def:psi}, and the log-sum inequality, we obtain that for any $(y,\hatx_1)$ such that $P_{Y\hatX_1}^*(y,\hatx_1)>0$,
\begin{align}
\nn&\inf_{P_{\hatX_2|X,\hatX_1}}\inf_{Q_{\hatX_2|Y,\hatX_1}}\psi(x,\hatx_1|P_{\hatX_2|X\hatX_1},Q_{\hatX_2|Y\hatX_1})\\
&=\inf_{P_{\hatX_2|X,\hatX_1}}D(P_{\hatX_2|X,\hatx_1}\|P_{\hatX_2|y,\hatx_1}|P_{Y|\hatx_1}^*)+t_2^*\mathbb{E}_{P_{X|y,\hatx_1}^*\times P_{\hatX_2|X,\hatx_1}}[d_2(X,\hatX_2)]\\
&=\mathbb{E}_{P_{X|y,\hatx_1}^*}[\log \beta_2(X,y,\hatx_1)]\label{optb2},
\end{align}
and 
\begin{align}
\inf_{P_{\hatX_2|X\hatX_1}}\psi(y,\hatx_1|P_{\hatX_2|X\hatX_1},Q_{\hatX_2|Y\hatX_1})=\mathbb{E}_{P_{X|y,\hatx_1}^*}[\log \beta_2(X,y,\hatx_1|Q_{\hatX_2|Y\hatX_1})]\label{suboptb2}.
\end{align}

From \eqref{optb2} and \eqref{suboptb2}, we obtain that
\begin{align}
\inf_{Q_{\hatX_2|Y,\hatX_1}}\mathbb{E}_{P_{X|y,\hatx_1}^*}[\log \beta_2(X,y,\hatx_1|Q_{\hatX_2|Y\hatX_1})]
&=\mathbb{E}_{P_{X|y,\hatx_1}^*}[\log \beta_2(X,y,\hatx_1)]\label{optqhatx2gyhatx14fy}.
\end{align}

Fix $\epsilon\in[0,1]$ and let $b\in\hat{\calX}_2$. Let $\barQ_{\hatX_2|y,\hatx_1}=(1-\epsilon)P_{\hatX_2|x,\hatx_1}^*(\hatx_2)+\epsilon1\{\hatx_2=b\}$. Using the definitions of $\beta_2(\cdot|\cdot)$ in \eqref{def:b2q} and $\beta_2(\cdot)$ in \eqref{def:b2}, we obtain that
\begin{align}
\frac{1}{\beta_2(x,y,\hatx_1|\barQ_{\hatX_2|Y\hatX_1})}
&=\frac{1-\epsilon}{\beta_2(x,y,\hatx_1)}+\epsilon \exp(-t_2^*d_2(x,b))\label{fyreuse}.
\end{align}
Thus,
\begin{align}
\nn&\frac{\partial \mathbb{E}_{P_{X|y,\hatx_1}^*}[\log \beta_2(X,y,\hatx_1|\barQ_{\hatX_2|Y\hatX_1})]}{\partial \epsilon}\Bigg|_{\epsilon=0}\\
&=\mathbb{E}_{P_{X|y,\hatx_1}^*}\Bigg[\beta_2(X,y,\hatx_1|\barQ_{\hatX_2|Y\hatX_1})\Big\{\frac{1-\epsilon}{\beta_2(X,y,\hatx_1)}+\epsilon \exp(-t_2^*d_2(x,b))\Big\}\Bigg]\Bigg|_{\epsilon=0}\\
&=1-\mathbb{E}_{P_{X|y,\hatx_1}^*}[\beta_2(X,y,\hatx_1)\exp(-t_2^*d_2(X,b))]\\
&\geq 0\label{nonnegativeopt14fy},
\end{align}
where \eqref{nonnegativeopt14fy} follows from \eqref{optqhatx2gyhatx14fy} which indicates that the derivative of the quantity at $\epsilon=0$ is nonnegative since the minimum is achieved by the test channel $P_{\hatX_2|Y\hatX_1}^*$. From the result in \eqref{optcond2} and the definitions in \eqref{def:b2}, \eqref{def:b}, we obtain that for $(y,\hatx_1)$ such that $P_{Y\hatX_1}^*(y,\hatx_1)>0$,
\begin{align}
P_{X|Y\hatX_1}^*(x|y,\hatx_1)
&=\frac{P_X(x)1\{y=g(x)\}P_{\hatX_1|X}^*(\hatx_1|x)}{\sum_x P_X(x)1\{y=g(x)\}P_{\hatX_1|X}^*(\hatx_1|x)}\\
&=\frac{P_X(x)1\{y=g(x)\}\beta(x,g(x))P^*_{\hatX_1}(\hatx_1)\exp\Bigg(-\frac{t_1^*d_1(x,\hatx_1)+\log\beta_2(x,g(x),\hatx_1)}{1+s^*}\Bigg)}{\sum_x  P_X(x)1\{y=g(x)\}\beta(x,g(x))P^*_{\hatX_1}(\hatx_1)\exp\Bigg(-\frac{t_1^*d_1(x,\hatx_1)+\log\beta_2(x,g(x),\hatx_1)}{1+s^*}\Bigg)}\\
&=\frac{P_X(x)1\{y=g(x)\}\beta(x,g(x))\exp\Bigg(-\frac{t_1^*d_1(x,\hatx_1)+\log\beta_2(x,g(x),\hatx_1)}{1+s^*}\Bigg)}{\sum_x P_X(x)1\{y=g(x)\}\beta(x,g(x))\exp\Bigg(-\frac{t_1^*d_1(x,\hatx_1)+\log\beta_2(x,g(x),\hatx_1)}{1+s^*}\Bigg)}.
\end{align}
Hence, from \eqref{nonnegativeopt14fy}, we obtain that
\begin{align}
1
&\geq \mathbb{E}_{P_{X|y,\hatx_1}^*}[\beta_2(X,y,\hatx_1)\exp(-t_2^*d_2(x,b))]\\
&=\sum_x \frac{P_X(x)1\{y=g(x)\}\beta(x,g(x))\exp\Bigg(-\frac{t_1^*d_1(x,\hatx_1)+\log\beta_2(x,g(x),\hatx_1)}{1+s^*}\Bigg)}{\sum_x P_X(x)1\{y=g(x)\}\beta(x,g(x))\exp\Bigg(-\frac{t_1^*d_1(x,\hatx_1)+\log\beta_2(x,g(x),\hatx_1)}{1+s^*}\Bigg)}\beta_2^*(x,y,\hatx_1)\exp(-t_2^*d_2(x,b)),
\end{align}
i.e.,
\begin{align}
\nn&\sum_xP_X(x)1\{y=g(x)\}\beta(x,g(x))\beta_2^{\frac{s^*}{1+s^*}}(x,g(x),\hatx_1)\exp\Bigg(-\frac{t_1}{1+s^*}d_1(x,\hatx_1)-t_2^*d_2(x,\hatx_2)\Bigg)\\
&\leq \sum_x P_X(x)1\{y=g(x)\}\beta(x,g(x))\exp\Bigg(-\frac{t_1^*d_1(x,\hatx_1)+\log\alpha_2(x,g(x),\hatx_1)}{1+s^*}\Bigg)\label{nowstep14fy}.
\end{align}
Fix an arbitrary distribution $Q_{\hatX_2|Y\hatX_1}$. Multiplying $Q_{\hatX_2|Y\hatX_1}(\hatx_2|y,\hatx_1)$ to both sides of \eqref{nowstep14fy}, summing over $(y,\hatx_2)$, and using the definitions of $w_1(\cdot)$ in \eqref{def:w1}, $w_2(\cdot)$ in \eqref{def:w2}, from \eqref{nu1typical}, we obtain that for any $\hatx_1$ such that $P_{\hatX_1}^*(\hatx_1)>0$ and arbitrary $Q_{\hatX_2|Y\hatX_1}$,
\begin{align}
w_2(\hatx_1,Q_{\hatX_2|Y\hatX_1})&\leq w_1(\hatx_1)\leq 1.
\end{align}

Finally, we will show that for that $w_2(\hatx_1,Q_{\hatX_2|Y\hatX_1})\leq w_1(\hatx_1)$ holds for arbitrary $(\hatx_1,Q_{\hatX_2|Y\hatX_1})$ satisfying $P_{\hatX_1}^*(\hatx_1)=0$. Recall \eqref{fyreuse} and the choice of $\bar{Q}_{\hatX_2|Y\hatX_1}$ above \eqref{fyreuse}. Thus, from \eqref{chooseopt} and \eqref{nonnegative2}, we have that for any $a\in\hat{\calX}_1$ satisfying $P_{\hatX_1}^*(a)=0$,
\begin{align}
\nn&\frac{\partial \mathbb{E}_{P_{X|y}}\Bigg[\beta(X,y)\beta_2^{-\frac{1}{1+s^*}}(X,y,a|\barQ_{\hatX_2|Y\hatX_1})\exp\Big(-\frac{t_1^*}{1+s^*}d_1(X,a)\Big)\Bigg]}{\partial \epsilon}\Bigg|_{\epsilon=0}\\
&=\mathbb{E}_{P_{X|y}}\Bigg[\beta(X,y)\beta_2^{\frac{s^*}{1+s^*}}(X,y,a|\barQ_{\hatX_2|Y\hatX_1})\exp\Big(-\frac{t_1^*}{1+s^*}d_1(X,a)\Big)\Big(\exp(-t_2^*d_2(X,b))-\frac{1}{\beta_2(x,y,a)}\Big)\Bigg]\Bigg|_{\epsilon=0}\\
\nn&=\mathbb{E}_{P_{X|y}}\Bigg[\beta(X,y)\beta_2^{\frac{s^*}{1+s^*}}(X,y,a)\exp\Big(-\frac{t_1^*}{1+s^*}d_1(X,a)-t_2^*d_2(X,b)\Big)\Bigg]\\
&\qquad-\mathbb{E}_{P_{X|y}}\Bigg[\beta(X,y)\beta_2^{-\frac{1}{1+s^*}}(X,y,a)\exp\Big(-\frac{t_1^*}{1+s^*}d_1(X,a)\Big)\Bigg]\label{usechooseopt}\\
&\leq 0\label{cauoptimal},
\end{align}
where \eqref{usechooseopt} follows from the fact that when $\epsilon=0$, $\barQ_{\hatX_2|Y\hatX_1}=P_{\hatX_2|Y\hatX_1}^*$; and \eqref{cauoptimal} follows since we chooses $P_{\hatX_2|Y\hatX_1}^*(\cdot|y,a)$ as in \eqref{chooseopt} when $P_{Y\hatX_1}^*(y,a)=0$.

Let $Q_{\hatX_2|Y\hatX_1}$ be arbitrary distribution. Multiplying $P_Y(y)Q_{\hatX_2|Y\hatX_1}(\hatx_2|y,a)$ on both sides of \eqref{cauoptimal}, summing over $(y,\hatx_2)$, and using the definitions in \eqref{def:w1}, \eqref{def:w2}, we obtain that
\begin{align}
w_2(a,Q_{\hatX_2|Y\hatX_1})\leq w_1(a)\leq 1,
\end{align}
for any $a\in\hat{\calX}_1$ satisfying $P_{\hatX_1}^*(a)=0$ and arbitrary $Q_{\hatX_2|Y\hatX_1}$. The proof of Lemma \ref{wleq1} is now complete.

\subsection{Proof of Lemma \ref{oneshotconverse4fy}}
\label{proofoneshot4fy}
Let us consider stochastic encoders and decoders. Let encoders $f_1,f_2$ be conditional distributions $P_{S_1|X}$ and $P_{S_2|X}$ respectively. Let the decoders $\phi_1,\phi_2,\phi_3$ be conditional distributions $P_{\hatX_1|S_1}$, $P_{\hatX_2|S_1S_2}$ and $P_{\hatY|S_2}$. Let $Q_{S_1S_2}$ be uniform over $\calM_1\times\calM_2$. Further, let $Q_{S_1}$ and $Q_{S_2}$ be induced by $Q_{S_1S_2}$. Then, we have
\begin{align}
\nn&P_{XYS_1S_2\hatX_1\hatX_2\hatY}(x,y,s_1,s_2,\hatx_1,\hatx_2,\haty)\\
&=P_X(x)1\{y=g(x)\}P_{S_1|X}(s_1|x)P_{S_2|X}(s_2|x)P_{\hatX_1|S_1}(\hatx_1|s_1)P_{\hatX_2|S_1S_2}(\hatx_2|s_1,s_2)P_{\hatY|S_2}(\haty|s_2)
\end{align}
In the following, whenever we use $\Pr$ and $\mathbb{E}$, we mean the probability calculated with respect to $P_{XYS_1S_2\hatX_1\hatX_2\hatY}$ unless otherwise stated. 

Then, we have
\begin{align}
\nn&\Pr((X,Y)\in(\calA_1\cup\calA_2\cup\calA_3))\\
\nn&=\Pr((X,Y)\in(\calA_1\cup\calA_2\cup\calA_3),~d_1(X,\hatX_1)>D_1\mathrm{~or~}d_2(X,\hatX_2)\leq D_2\mathrm{~or~}\hatY\neq Y)\\
&\qquad+\Pr((X,Y)\in(\calA_1\cup\calA_2\cup\calA_3),~d_1(X,\hatX_1)\leq D_1,~d_2(X,\hatX_2)\leq D_2,~\hatY=Y)\\
\nn&\leq \varepsilon_1^{\rm{FY}}(D_1,D_2)+\Pr((X,Y)\in\calA_1,~d_1(X,\hatX_1)\leq D_1)+\Pr((X,Y)\in\calA_2,~\hatY=Y)\\
&\qquad+\Pr((X,Y)\in\calA_3,~d_1(X,\hatX_1)\leq D_1,~d_2(X,\hatX_2)\leq D_2)\label{upponeshot4fy}.
\end{align}
From \cite[Theorem 1]{kostina2012converse}, we obtain that
\begin{align}
\Pr((X,Y)\in\calA_1,~d_1(X,\hatX_1)\leq D_1)\leq \exp(-\epsilon)\label{uppa14fy}.
\end{align}
The third term in \eqref{upponeshot4fy} can be upper bounded as follows:
\begin{align}
\nn&\Pr((X,Y)\in\calA_2,~\hatY=Y)\\
&=\sum_{\substack{(x,y,s_2,\haty):~\haty=y\\P_Y(y)\leq \frac{\exp(-\gamma)}{M_2}}}P_{Y}(y)P_{X|Y}(x|y)P_{S_2|X}(s_2|x)P_{\hatY|S_2}(\haty|s_2)\\
&\leq \exp(-\gamma)\sum_{x,y,s_2}\frac{1}{M_2}P_{X|Y}(x|y)P_{S_2|X}(s_2|x)P_{\hatY|S_2}(y|s_2)\\
&\leq \exp(-\gamma)\sum_{x,y,s_2}Q_{S_2}(s_2)P_{X|Y}(x|y)P_{\hatY|S_2}(y|s_2)\label{upa24fy}\\
&\leq \exp(-\gamma)\label{hatY=y},
\end{align}
where \eqref{upa24fy} follows since $P_{S_2|X}(s_2|x)\leq 1$; \eqref{hatY=y} follows since $\sum_yP_{\hatY|S_2}(y|s_2)=\sum_{\haty}P_{\hatY|S_2}(y|s_2)=1$.

In the following, we will show that 
\begin{align}
\Pr((X,Y)\in\calA_3,~d_1(X,\hatX_1)\leq D_1,~d_2(X,\hatX_2)\leq D_2)\leq 2\exp(-\gamma).
\end{align}

Let $Q_{\hatX_1}$ be induced by $Q_{S_1}$ and $P_{\hatX_1|S_1}$. In a similar manner as \cite[Theorem 1]{kostina2012converse}, using the defintion of $\jmath_1(\cdot)$ in \eqref{def:j14fy}, we obtain that
\begin{align}
\nn&\Pr\Big\{\jmath_1(X,Y,\hatX_1,D_1)\geq \log M_1+\gamma,~d_1(X,\hatX_1)\leq D_1\Big\}\\
&=\Pr\Big\{\exp\big\{1\{d_1(X,\hatX_1)\leq D_1\}\jmath_1(X,Y,\hatX_1,D_1)-\gamma\big\}\geq  M_1\Big\}\\
&\leq \frac{\mathbb{E}\Bigg[\exp\bigg\{1\{d_1(X,\hatX_1)\leq D_1\}\jmath_1(X,Y,\hatX_1,D_1)-\gamma\bigg\}\Bigg]}{M_1}\label{usemarkov4fy}\\
&\leq \frac{\exp(-\gamma)}{M_1}\mathbb{E}\Bigg[\exp\bigg(\jmath_1(X,Y,\hatX_1,D_1)+\frac{t_1^*(D_1-d_1(X,\hatX_1))}{1+s^*}\bigg)\Bigg]\\
&=\frac{\exp(-\gamma)}{M_1}\mathbb{E}\Bigg[\exp\bigg(\imath_1(X,Y,\hatX_1)-\frac{t_1^*d_1(X,\hatX_1)}{1+s^*}\bigg)\Bigg]\label{usedefj14fy}\\
&=\exp(-\gamma)\mathbb{E}_{P_{XY}}\Bigg[\sum_{s_1,\hatx_1}\frac{1}{M_1}P_{S_1|X}(s_1|X)P_{\hatX_1|S}(\hatx_1|s)\exp\bigg(\imath_1(X,Y,\hatx_1)-\frac{t_1^*d_1(X,\hatx_1)}{1+s^*}\bigg)\Bigg]\\
&\leq\exp(-\gamma) \mathbb{E}_{P_{XY}}\Bigg[\sum_{s_1,\hatx_1}Q_{S_1}(s_1)P_{\hatX_1|S_1}(\hatx_1|s_1)\exp\bigg(\imath_1(X,Y,\hatx_1)-\frac{t_1^*d_1(X,\hatx_1)}{1+s^*}\bigg)\Bigg]\label{ple14fy}\\
&=\exp(-\gamma)\mathbb{E}_{Q_{\hatX_1}}\bigg[\mathbb{E}_{P_{XY}}\Big[\exp(\imath_1(X,Y,\hatX_1)\Big]\bigg]\\
&\leq \exp(-\gamma)\label{usewle1},
\end{align}
where \eqref{usemarkov4fy} follows from Markov's inequality; \eqref{usedefj14fy} follows from \eqref{def:j14fy}; \eqref{ple14fy} follows since $P_{S_1|X}(s_1|x)\leq 1$ for all $(x,s_1)$; and \eqref{usewle1} follows from the second-inequality in \eqref{w2leqw1le1}.

Let $Q_{\hatX_1X_2}$ be induced by $Q_{S_1S_2}$, $P_{\hatX_1|S_1}$ and $P_{\hatX_2|S_2}$. Further, let $Q_{\hatX_2|\hatX_1}$ be induced by $Q_{\hatX_1\hatX_2}$ and let $Q_{\hatX_2|Y\hatX_1}$ be chosen such that for every $(y,\hatx_1,\hatx_2)$, we have $Q_{\hatX_2|Y\hatX_1}(\hatx_2|y,\hatx_1)=Q_{\hatX_2|\hatX_1}(\hatx_2|\hatx_2)$. In a similar manner as \eqref{usewle1}, using the definition of $\jmath_2(\cdot)$ in \eqref{def:j24fy}, we can show that
\begin{align}
\nn&\Pr(\jmath_2(X,Y,\hatX_1,D_1,D_2)\geq \log M_1M_2+\gamma,~d_1(X,\hatX_1)\leq D_1,~d_2(X,\hatX_2)\leq D_2)\\
&\leq \frac{\exp(-\gamma)}{M_1M_2}\mathbb{E}\Bigg[\exp\bigg(\imath_2(X,Y,\hatX_1)-\frac{s^*d_1(X,\hatX_1)}{1+s^*}+t_2^*d_2(X,\hatX_2)\bigg)\Bigg]\\
&\leq \exp(-\gamma)\mathbb{E}_{Q_{\hatX_1}}\Bigg[\mathbb{E}_{P_{XY}\times Q_{\hatX_2|Y\hatX_1}}\bigg[\exp\bigg(\imath_2(X,Y,\hatX_1)-\frac{s^*d_1(X,\hatX_1)}{1+s^*}+t_2^*d_2(X,\hatX_2)\bigg|\hatX_1\bigg]\Bigg]\\
&\leq \exp(-\gamma)\label{usew2}.
\end{align}
Using the results in \eqref{j=j1j24fy}, \eqref{usewle1} and \eqref{usew2}, we conclude that
\begin{align}
\nn&\Pr((X,Y)\in\calA_3,~d_1(X,\hatX_1)\leq D_1,~d_2(X,\hatX_2)\leq D_2)\\
&=\Pr(\jmath_{\rm{FY}}(X,Y|R_1,D_1,D_2,P_X)\geq \log M_1M_2+s^*\log M_1+(1+s^*)\gamma,~d_1(X,\hatX_1)\leq D_1,~d_2(X,\hatX_2)\leq D_2)\\
\nn&\leq \Pr(s^*\jmath_1(X,Y,\hatX_1,D_1)+\jmath_2(X,Y,\hatX_1,D_1,D_2)\geq \log M_1M_2+s^*\log M_1+(1+s^*)\gamma,\\
&\qquad\qquad\qquad~\mathrm{and}~d_1(X,\hatX_1)\leq D_1,~d_2(X,\hatX_2)\leq D_2)\label{upj14fy}\\
\nn&\leq \Pr\Big\{\jmath_1(X,Y,\hatX_1,D_1)\geq \log M_1+\gamma,~d_1(X,\hatX_1)\leq D_1\Big\}\\
&\qquad\qquad+ \Pr(\jmath_2(X,Y,\hatX_1,D_1,D_2)\geq \log M_1M_2+\gamma,~d_1(X,\hatX_1)\leq D_1,~d_2(X,\hatX_2)\leq D_2)\label{up2j14fy}\\
&\leq 2\exp(-\gamma)\label{upa34fy}.
\end{align}
where \eqref{upj14fy} follows since $\Pr(A-c\geq B)\leq \Pr(A\geq B)$ for any random variables $(A,B)$ and non-negative constant $c$; \eqref{up2j14fy} follows since that for any random variables $(A,B)$ and constants $(c,d)$, we have $\Pr(A+B\geq c+d)\leq \Pr(A\geq c)+\Pr(B\geq d)$.

The proof of Lemma \ref{oneshotconverse4fy} is complete by combining \eqref{upponeshot4fy}, \eqref{uppa14fy}, \eqref{upa24fy} and \eqref{upa34fy}.

\subsection{Proof of Lemma \ref{lemmataylor}}
\label{prooflemmataylor}

The following lemma plays a central role in the proof of Lemma \ref{lemmataylor}. We relate the defined distortion-tilted information density to the first derivative of the Kaspi rate-distortion function. Considering the fact that our correlated source may not be fully supported, we need the parametric notation above \eqref{def:calan} (See also \cite{watanabe2015second}). 

\begin{lemma}
\label{kaspifirstderive}
Suppose for all $Q_{XY}$ in the neighborhood of $P_{XY}$, $\mathrm{supp}(Q_{\hat{X}_1\hat{X}_2}^*)=\mathrm{supp}(P_{\hat{X}_1\hat{X}_2}^*)$. Then, for all $i\in[1:m]$,
\begin{align}
\frac{\partial R(\Theta(Q_{XY}),D_1,D_2)}{\partial \Theta_i(Q_{XY})}\Big|_{Q_{XY}=P_{XY}}&=\jmath(x_i,y_i|D_1,D_2,P_{XY})-\jmath(x_m,y_m|D_1,D_2,P_{XY}).
\end{align}
\end{lemma}
The proof of Lemma \ref{kaspifirstderive} is given in Appendix \ref{proofkaspifirst}.

Note that for $(x^n,y^n)$ such that $\hat{T}_{x^ny^n}\in\calA_n(P_{XY})$, invoking Lemma \ref{kaspifirstderive} and applying Taylor's expansion, we obtain that
\begin{align}
&\nn R(\hat{T}_{x^ny^n},D_1,D_2)\\
\nn&=R_{\rmK}(P_{XY},D_1,D_2)+\sum_{i=1}^m\big(\Theta_i(\hat{T}_{x^ny^n})-\Theta_i(P_{XY})\big)\big(\jmath(x_i,y_i|D_1,D_2,P_{XY})-\jmath(x_m,y_m|D_1,D_2,P_{XY})\big)\\
&\qquad+O(\|\hat{T}_{x^ny^n}-P_{XY}\|^2)\\
\nn&=R_{\rmK}(P_{XY},D_1,D_2)+\sum_{x,y}\big(\hat{T}_{x^ny^n}(x,y)-P_{XY}(x,y)\big)\jmath(x,y|D_1,D_2,P_{XY})+O(\|\hat{T}_{x^ny^n}-P_{XY}\|^2)\\
&=\sum_{x,y} \hat{T}_{x^ny^n}(x,y)\jmath(x,y|D_1,D_2,P_{XY})+O\Big(\frac{\log n}{n}\Big)\\
&=\frac{1}{n}\sum_{i=1}^n \jmath(x_i,y_i|D_1,D_2,P_{XY})+O\Big(\frac{\log n}{n}\Big).
\end{align}

\subsection{Proof of Lemma \ref{kaspifirstderive}}
\label{proofkaspifirst}

From the assumption in Lemma \ref{kaspifirstderive}, we conclude that $\Theta(Q_{XY})$ is supported on $\calS=\mathrm{supp}(P_{XY})$. Let $(Q_{\hatX_1|XY}^*,Q_{\hatX_2|XY\hatX_1}^*)$ be an optimal test channel achieving $R(Q_{XY},D_1,D_2)$. Let $Q_{\hatX_1}^*,Q_{\hatX_2|Y\hatX_1}^*$ be induced by $Q_{XY}$ and $Q_{\hat{X}_1\hat{X}_2|XY}^*$.

Invoking Lemma \ref{opttest4kaspi} and \eqref{def:kaspitilt}, we obtain
\begin{align}
R_{\rmK}(\Theta(Q_{XY}),D_1,D_2)
&=\mathbb{E}_{Q_{XY}}\Big[\jmath_{\rmK}(X,Y|D_1,D_2,\Theta(Q_{XY}))\Big]\\
&=\sum_{i=1}^m \Theta_i(Q_{XY})\jmath_{\rmK}(x_i,y_i|D_1,D_2,\Theta(Q_{XY}))\label{thetaq},
\end{align}
and for every $(\hat{x}_1,\hat{x}_2)$ such that $Q_{\hat{X}_1}^*(\hat{x}_1)Q_{\hat{X}_2|Y\hat{X}_1}^*(\hat{x}_2|y,\hat{x}_1)>0$, we have
\begin{align}
\jmath_{\rmK}(x,y|D_1,D_2,\Theta(Q_{XY}))\nn&=\log \frac{Q_{\hat{X}_1|XY}^*(\hat{x}_1|x,y)}{Q_{\hat{X}_1}^*(\hat{x}_1)}+\lambda_{1,Q}^*(d_1(x,\hat{x}_1)-D_1)\\
&\qquad+\log \frac{Q_{\hat{X}_2|XY\hat{X}_1}^*(\hat{x}_2|x,y,\hat{x}_1)}{Q_{\hat{X}_2|Y\hat{X}_1}^*(\hat{x}_2|y,\hat{x}_1)}+\lambda_{2,Q}^*(d_2(x,\hat{x}_2)-D_2),\label{jmathqs}
\end{align}
where $\lambda_{1,Q}^*$ and $\lambda_{2,Q}^*$ are defined similarly as $\lambda_1^*$ in \eqref{dualoptimal1} and $\lambda_2^*$ in \eqref{dualoptimal2}.

Hence, invoking \eqref{thetaq}, we obtain
\begin{align}
\nn&\frac{\partial R_{\rmK}(\Theta(Q_{XY}),D_1,D_2)}{\partial \Theta_i(Q_{XY})}\Big|_{Q_{XY}=P_{XY}}\\
&=\sum_{j=1}^m \frac{\Theta_j(Q_{XY})}{\partial \Theta_i(Q_{XY})}\jmath_{\rmK}(x_j,y_j|D_1,D_2,\Theta(P_{XY}))+\sum_{j=1}^m \Theta_j(P_{XY})\frac{\jmath_{\rmK}(x_j,y_j|D_1,D_2,\Theta(Q_{XY}))}{\partial \Theta_i(Q_{XY})}\\
&=\jmath_{\rmK}(x_i,y_i|D_1,D_2,P_{XY})-\jmath_{\rmK}(x_m,y_m|D_1,D_2,P_{XY})+\sum_{j=1}^m \Theta_j(P_{XY})\frac{\jmath_{\rmK}(x_j,y_j|D_1,D_2,\Theta(Q_{XY}))}{\partial \Theta_i(Q_{XY})}\Bigg|_{Q_{XY}=P_{XY}}
\label{derive2},
\end{align}
where \eqref{derive2} holds since $\Theta_m(P_{XY})=1-\sum_j^{m-1} \Theta_j(P_{XY})$. Then we focus on the second term in \eqref{derive2}. Invoking \eqref{jmathqs}, we obtain
\begin{align}
\nn&\jmath_{\rmK}(x,y|D_1,D_2,\Theta(Q_{XY}))\\
\nn&=\sum_{\hat{x}_1:Q_{\hat{X}_1}^*(\hat{x}_1)>0}Q_{\hat{X}_1|XY}^*(\hat{x}_1|x,y)\Bigg\{\Bigg(\log \frac{Q_{\hat{X}_1|XY}^*(\hat{x}_1|x,y)}{Q_{\hat{X}_1}^*(\hat{x}_1)}+\lambda_{1,Q}^*(d_1(x,\hat{x}_1)-D_1)\Bigg)\\
&\qquad+\sum_{\hat{x}_2:Q_{\hat{X}_2|Y\hat{X}_1}(\hat{x}_2|y,\hat{x}_1)>0}Q_{\hat{X}_2|XY\hat{X}_1}^*(\hat{x}_2|x,y,\hat{x}_1)\Bigg(\log \frac{Q_{\hat{X}_2|XY\hat{X}_1}^*(\hat{x}_2|x,y,\hat{x}_1)}{Q_{\hat{X}_2|Y\hat{X}_1}^*(\hat{x}_2|y,\hat{x}_1)}+\lambda_{2,Q}^*(d_2(x,\hat{x}_2)-D_2)\Bigg)\Bigg\}\label{jmathq2}.
\end{align}

Invoking \eqref{jmathq2}, in a similar manner as \cite[Lemma 3]{watanabe2015second} and \cite[Lemma 3]{zhou2015second}, we obtain that
\begin{align}
\nn&\sum_{j=1}^m \Theta_j(P_{XY})\frac{\jmath_{\rmK}(x_j,y_j|D_1,D_2,\Theta(Q_{XY}))}{\partial \Theta_i(Q_{XY})}\Bigg|_{Q_{XY}=P_{XY}}\\
\nn&=\Bigg\{\sum_{j=1}^m \Theta_j(P_{XY})\sum_{\hat{x}_1:Q_{\hat{X}_1}^*(\hat{x}_1)>0}Q_{\hat{X}_1|XY}^*(\hat{x}_1|x_j,y_j)\Bigg\{\Bigg( \frac{\partial }{\partial \Theta_i(Q_{XY})}\Bigg(\log \frac{Q_{\hat{X}_1|XY}^*(\hat{x}_1|x_j,y_j)}{Q_{\hat{X}_1}^*(\hat{x}_1)}\Bigg)\\
&\nn\qquad+\frac{\partial }{\partial \Theta_i(Q_{XY})}\Big(\lambda_{1,Q}^*\big(d_1(x,\hat{x}_1)-D_1\big)\Big)\Bigg)\\
&\nn\qquad+\sum_{\hat{x}_2:Q_{\hat{X}_2|Y\hat{X}_1}(\hat{x}_2|y,\hat{x}_1)>0}Q_{\hat{X}_2|XY\hat{X}_1}^*(\hat{x}_2|x_j,y_j,\hat{x}_1)\Bigg(\frac{\partial }{\partial \Theta_i(Q_{XY})}\Bigg(\log \frac{Q_{\hat{X}_2|XY\hat{X}_1}^*(\hat{x}_2|x_j,y_j,\hat{x}_1)}{Q_{\hat{X}_2|Y\hat{X}_1}^*(\hat{x}_2|y_j,\hat{x}_1)}\Bigg)\\
&\qquad+\frac{\partial \lambda_{2,Q}^*}{\partial \Theta_i(Q_{XY})}\bigg(d_2(x,\hat{x}_2)-D_2\bigg)\Bigg)\Bigg\}\Bigg\}\Bigg|_{Q_{XY}=P_{XY}}\\
&=0\label{finalderive}.
\end{align}

\bibliographystyle{IEEEtran}
\bibliography{IEEEfull_lin}
 
\end{document}